\documentclass{article}
\usepackage{myijcai16}
\pdfoutput=1

\usepackage{times}
\usepackage{latexsym}
\usepackage{tabularx,booktabs,multirow}
\usepackage{subcaption}

\usepackage{amsmath}
 \usepackage{tikz}
\usepackage{xspace}
\usepackage{amsthm}
\usepackage{amssymb}
\usepackage{amsfonts}
\usepackage{paralist}
\usepackage{url}

\newcommand{\fo}{\kw{FO}}

\newcommand{\types}{\kw{Types}}
\newcommand{\instance}{{\cal F}}

\newcommand{\fgtgd}{\kw{FGTGD}}
\newcommand{\afgtgd}{\kw{BaseFGTGD}}
\newcommand{\acfgtgd}{\kw{BaseCovFGTGD}}

\newcommand{\agnf}{\kw{BaseGNF}}
\newcommand{\acgnf}{\kw{BaseCovGNF}}
\newcommand{\mapsfrom}{:=}

\newcommand{\gnf}{\kw{GNF}}

\newcommand{\gf}{\kw{GF}}

\newcommand{\decode}{\kw{decode}}

\newcommand{\dom}{\kw{Dom}}
\newcommand{\child}{\kw{Child}}
\newcommand{\incd}{\kw{ID}}
\newcommand{\aincd}{\kw{BaseID}}
\newcommand{\did}{\kw{DID}}
\newcommand{\owqa}{\kw{QA}}

\newcommand{\owqalin}{\owqa{\kw{lin}}}
\newcommand{\owqatr}{\owqa{\kw{tr}}}
\newcommand{\owqatrans}{\owqatr}

\newcommand{\owqatc}{\owqa{\kw{tc}}}

\newcommand{\eg}{e.g.~}

\newcommand{\kw}[1]{{\mathsf{#1}}\xspace}

\newtheorem{theorem}{Theorem}

\newtheorem{proposition}{Proposition}
\newtheorem{corollary}{Corollary}

  \newtheorem{definition}{Definition}

\renewcommand{\phi}{\varphi}

\newtheorem{lemma}{Lemma}
 
\newcommand{\myparagraph}[1]{\paragraph{#1.}}
\newcommand{\myeat}[1]{}

\makeatletter
\def\@Opargbegintheorem#1#2#3#4{#4\trivlist
      \item[\hskip\labelsep{#3#1}]{#3#2\@thmcounterend\ }}

\newcommand{\definerep}[2]{%
\spnewtheorem*{#1rp}{#2}{\bf}{\itshape}
\newenvironment{#1rep}[2]{%
  \ifthenelse{\equal{##1}{*}}
  {\begin{#1rp}[\ref{##2}]}
  {\begin{#1}\label{##2}}}
{\ifthenelse{\equal{\@currenvir}{#1}}{\end{#1}}{\end{#1rp}}}
}

\newcommand{\calA}{\mathcal{A}}
\newcommand{\calB}{\mathcal{B}}
\newcommand{\calC}{\mathcal{C}}
\newcommand{\calD}{\mathcal{D}}

\newcommand{\calF}{\mathcal{F}}
\newcommand{\calG}{\mathcal{G}}

\newcommand{\calN}{\mathcal{N}}

\newcommand{\frakA}{\mathfrak{A}}

\newcommand{\sigmab}{\sigma_{\calB}}

\newcommand{\sigmad}{\sigma_{\calD}}

\newcommand{\arity}[1]{\kw{arity}(#1)}

\makeatletter
\newcommand*{\defeq}{\mathrel{\rlap{%
  \raisebox{0.3ex}{$\m@th\cdot$}}%
  \raisebox{-0.3ex}{$\m@th\cdot$}}%
  =}
\makeatother

\newcommand{\guarded}{\kw{guarded}}
\newcommand{\guardedb}{\kw{guarded}_{\sigmab}}
\newcommand{\guardedbg}{\kw{guarded}_{\sigmab \cup \set{G}}}
\newcommand{\agnfk}{\agnf^k}
\newcommand{\acgnfk}{\acgnf^k}
\newcommand{\drel}{S}

\newcommand{\acgnfplus}{\acgnf^{+}}
\newcommand{\acgnfminus}{\acgnf^{-}}
\newcommand{\tree}{T}
\newcommand{\mysize}[1]{|#1|}

\newcommand{\witness}{\kw{witness}}
\newcommand{\twoexp}{\kw{2EXPTIME}}

\newcommand{\exptime}{\kw{EXPTIME}}
\newcommand{\conp}{\kw{CoNP}}
\newcommand{\np}{\kw{NP}}
\newcommand{\ptime}{\kw{PTIME}}

\newcommand{\trans}{+}

\newcommand{\true}{\mathfrak{t}}
\newcommand{\false}{\mathfrak{f}}

\newcommand{\gn}{\text{GN}}
\newcommand{\fA}{{\cal F}}
\newcommand{\fB}{{\cal G}}
\renewcommand{\restriction}{\mathord{\upharpoonright}}
\newcommand{\restrict}[2]{ {#1}\restriction_{#2} }

\newcommand{\set}[1]{\{ #1 \}}
\newcommand{\elems}[1]{\kw{elems}(#1)}
\newcommand{\structureunravelk}[1]{#1^k}
\newcommand{\structureunravelkint}[1]{#1^k_{\cal B}}

\newcommand{\Or}{\mathrm{Or}}

\newcommand{\etal}{\emph{et al.}}

\newcommand{\card}[1]{\left|#1\right|}

\title{Query Answering with Transitive and Linear-Ordered Data}
\author{
Antoine Amarilli \\
LTCI, CNRS, T{\'e}l{\'e}com ParisTech, Universit{\'e} Paris-Saclay
\And
Michael Benedikt \\
University of Oxford
\AND
Pierre Bourhis \\
CNRS CRIStAL, Universit\'e Lille 1, INRIA Lille
\And
Michael Vanden Boom \\ 
University of Oxford
}

\begin{document}

\maketitle

\begin{abstract}
We consider entailment problems involving powerful constraint
languages such as \emph{guarded existential rules}, in which additional semantic restrictions
are put on a set of distinguished relations. We consider 
restricting a relation to be transitive,  restricting a relation  to be the transitive closure of another relation, 
and restricting a relation to be a linear order. We give some natural generalizations of guardedness
that allow inference to be decidable in each case, and isolate the complexity of the corresponding decision problems. Finally
we show that slight changes in our conditions lead to undecidability.

\end{abstract}

\section{Introduction} \label{sec:intro}

The \emph{query answering problem} (or certain
answer problem), abbreviated here as $\owqa$, 
is a fundamental reasoning problem
in both knowledge representation and databases.
It asks
whether a query (e.g.~given by an existentially-quantified conjunction of atoms) is entailed
by a set of constraints and a set of facts.
A common class of constraints used
for $\owqa$ are the \emph{existential
rules}, also known as \emph{tuple generating dependencies}
(TGDs).
Although query answering is known to be undecidable  for general TGDs,
there are a number of subclasses 
that admit decidable $\owqa$, such as those
based on \emph{guardedness}.
For instance, \emph{guarded} TGDs
 require all variables in the body
 of the dependency to appear in a single body atom (the \emph{guard}).
 \emph{Frontier-guarded} TGDs ($\fgtgd$s)
relax this condition
and require only that some
guard atom 
contains the variables that occur in both head and body \cite{bagetcomplexityfg}.
This includes standard
SQL referential constraints as well as important constraint
classes (e.g.~role inclusions) arising in knowledge representation.
Guarded existential rules can be generalized to
\emph{guarded logics} that allow disjunction and negation and still
enjoy decidable $\owqa$, \eg the
guarded fragment of first-order logic ($\gf$)
\cite{andreka1998modal} and the Guarded Negation Fragment
($\gnf$) \cite{gnficalp}.

A key challenge is to extend these results to capture
additional semantics of the relations. For example, the property
that a binary relation is \emph{transitive} cannot be expressed
directly in guarded logics, and neither can the property that one 
relation is the \emph{transitive closure} of another.
Going beyond transitivity, one cannot express that a binary
relation is a \emph{linear order}.
Since ordered data is common in applications, this means that a key part of data semantics
is being lost.

There has been extensive work on decidability results for guarded logics
thus extended with semantic restrictions.
We first review known results for 
the
\emph{satisfiability problem}.

Ganzinger \etal~\shortcite{undecidgf2} showed that satisfiability is not
decidable for $\gf$ when two relations
are restricted to be transitive, even on \emph{arity-two} signatures
(i.e.~with only unary and binary relations).
For linear orders, \cite{kieronski2011decidability} showed $\gf$ is
undecidable when three relations are restricted to be (non-strict) linear
orders, even with only two variables (so on arity-two signatures). 
Otto \shortcite{ottofo2order} showed satisfiability is decidable for
two-variable logic with
one relation restricted to be a linear order.
For transitive relations,
one way to regain decidability for $\gf$ satisfiability
was shown by Szwast and Tendera~\shortcite{gftgdecid}:
allow transitive
relations \emph{only} in guards.

We now turn to the $\owqa$ problem.
Gottlob \etal~\shortcite{andreaslidia} showed that query answering for
$\gf$ with transitive relations only in guards
is undecidable,
even on arity-two signatures.
Baget \etal~\shortcite{mugnier15} studied
$\owqa$ 
with respect to a collection
of linear TGDs (those with
only a single atom in the body and the head).  They showed that the query
answering problem is decidable with such TGDs and transitive relations, if the
signature is binary or if other additional restrictions are obeyed.

The case of TGDs mentioning relations with a restricted interpretation 
has been studied in the database community mainly in the setting of acyclic schemas, such as those
that map source data to target data. Transitivity restrictions have not been studied,
but there has been work on inequalities \cite{abdus} and TGDs with arithmetic \cite{afratiarith}.
Due to the acyclicity assumptions, $\owqa$ is still decidable, and has data complexity in $\conp$.
The fact that the data complexity can be $\conp$-hard
is shown in \cite{abdus}, while polynomial cases
have been isolated in \cite{abdus} (in the presence of inequalities)
and \cite{afratiarith} (in the presence of arithmetic).

Transitivity has also been studied in description
logics, where the signature contains unary relations (concepts) and
binary relations (roles). 
In this arity-two context, $\owqa$ is decidable for many description logics
featuring expressive operators as well as transitivity, such as $\mathcal{ZIQ}$,
$\mathcal{ZOQ}$, $\mathcal{ZOI}$ \cite{calvanese2009regular}, 
Horn-$\mathcal{SROIQ}$ \cite{ortiz2011query},
or
regular-$\mathcal{EL}^{++}$ \cite{kroetzsch2007conjunctive}, but they often restrict
the interaction between transitivity and some features such as role inclusions
and Boolean role combinations. $\owqa$ becomes undecidable for more expressive
description logics with transitivity such as 
$\mathcal{ALCOIF^*}$ \cite{ortiz2010query2}
and $\mathcal{ZOIQ}$ \cite{ortiz2010query}, and the problem is open for
$\mathcal{SROIQ}$ and $\mathcal{SHOIQ}$ \cite{ortiz2012reasoning}.

The main contribution
of this work is to introduce a broad class of constraints over arbitrary-arity
vocabularies where query answering is decidable when additional semantics are
imposed on some \emph{distinguished relations}. We show that transitivity restrictions
can be handled in guarded and frontier-guarded constraints,
as long as these distinguished relations are \emph{not} used as guards
--- we call this new kind of restriction \emph{base-guardedness} (and similarly, we extend frontier-guarded to ``base-frontier-guarded'' and so forth). The base-guarded restriction is orthogonal to
the prior decidable cases such as transitive guards \cite{gftgdecid} or linear rules
\cite{mugnier15}.  

On the one hand, we show that the condition allows us to define very expressive and flexible decidable logics,  capable
of expressing  not only guarded existential rules, but guarded rules with negation and disjunction in the head.
They can express not only integrity constraints
but also conjunctive queries and their negations.
On the other hand, a by-product of our results is  new query answering schemes for some previously-studied
classes of guarded existential rules with extra semantic restrictions.
For example, our base-frontier-guarded constraints encompass
all \emph{frontier-one TGDs} (where at most one variable is shared between the body and head)
\cite{baget2009extending}.
Hence, our results imply that $\owqa$ is decidable with
transitive closure and frontier-one constraints, which answers a question of
\cite{mugnier15}.
Our results even extend to frontier-one  TGDs with distinguished
relations that are required to be the transitive closure of other relations.

Our results are shown by arguing that it is enough to consider entailment over  ``tree-like'' sets of facts. By representing
the set of witness representations
as a  tree automaton,
we derive upper bounds for the combined complexity of the problem. 
The sufficiency of tree-like examples also enable a refined analysis of \emph{data complexity} (when the query and
constraints are fixed).
Further, we use a set of coding techniques 
to  show  matching lower bounds
within our fragment.
We also show that loosening our conditions leads to undecidability.

Finally, we solve the $\owqa$ problem
when the distinguished relations are \emph{linear orders}.
We show that it is undecidable even assuming base-guardedness, so we introduce a
stronger condition called \emph{base-coveredness}:
not only are distinguished relations never used as guards, they are always \emph{covered}
by a non-distinguished atom.
Our  decidability technique works by ``compiling away'' linear order restrictions, obtaining an entailment problem
without any special restrictions.
The correctness proof for our reduction to classical $\owqa$ again relies
on the ability to  restrict reasoning to  sets of facts with tree-like representations.
To our knowledge,
these are the first
decidability results for the $\owqa$ problem with linear orders.
Again, we give tight complexity bounds, and show that  weakening
the base-coveredness condition leads to undecidability.

\begin{figure*}
\renewcommand{\thefigure}{} 
  \begin{subfigure}[b]{.64\linewidth}
  \begin{tabular}{l@{\quad}l@{~~}l@{\quad}l@{~~}l@{\quad}l@{~~}l}
\toprule
\multirow{ 2}{*}{\begin{tabular}[t]{@{}l@{}}{\bf Fragment}\\\null\end{tabular}} & \multicolumn{2}{c}{\bf $\!\!\!\!\!\!\!\!\!\!$$\owqatr$} &
\multicolumn{2}{c}{\bf $\!\!\!\!\!\!\!\!\!\!$$\owqatc$} & \multicolumn{2}{c}{\bf
$\!\!\!\!\!\!\!\!\!\!$$\owqalin$}\\
& \bf data & \bf combin.
& \bf data & \bf combin.
& \bf data & \bf combin. \\
\midrule
{$\agnf$}
& coNP-c & 2EXP-c
& coNP-c & 2EXP-c
& \multicolumn{2}{c}{{undecidable}} \\
{$\acgnf$}
& coNP-c & 2EXP-c
& coNP-c & 2EXP-c
& coNP-c & 2EXP-c \\
{$\afgtgd$}
& in coNP & 2EXP-c
& coNP-c & 2EXP-c
& \multicolumn{2}{c}{{undecidable}} \\
{$\acfgtgd$}
& P-c & 2EXP-c
& coNP-c & 2EXP-c
& coNP-c & 2EXP-c \\
\bottomrule
\end{tabular}
\caption{Summary of $\owqa$ results (for
base-covered fragments, queries are also base-covered)}
\label{tab:complexity}
\end{subfigure}
\begin{subfigure}[b]{.35\linewidth}
\begin{tikzpicture}[xscale=2.1,yscale=.6,inner sep=1pt]
\node (GNF) at (0, 0) {$\gnf$};
\node (AGNF) at (0, -1) {$\agnf$};
\node (ACGNF) at (0, -2) {$\acgnf$};
\node (FGTGD) at (1, -1) {$\fgtgd$};
\node (AFGTGD) at (1, -2) {$\afgtgd$};
\node (ACFGTGD) at (1, -3) {$\acfgtgd$};
\node (ID) at (2.13, -2) {$\incd$};
\node (AID) at (2.13, -4) {$\aincd$};
\node (FR1) at (1.9, -2.8) {fr-one};

\draw[->] (ACGNF) -- (AGNF);
\draw[->] (AGNF) -- (GNF);
\draw[->] (ACFGTGD) -- (AFGTGD);
\draw[->] (AFGTGD) -- (FGTGD);
\draw[->] (FGTGD) -- (GNF);
\draw[->] (AFGTGD) -- (AGNF);
\draw[->] (ACFGTGD) -- (ACGNF);
\draw[->] (ID) -- (FGTGD);
\draw[->] (AID) -- (ACFGTGD);
\draw[->] (AID) -- (ID);
\draw[->] (FR1) -- (AFGTGD);
\end{tikzpicture}
\caption{Taxonomy of fragments}
\label{fig:tax}
\end{subfigure}
\end{figure*}

\section{Preliminaries} \label{sec:prelims}

We work on a \emph{relational signature}  $\sigma$ where each relation $R \in\nolinebreak
\sigma$ has an associated \emph{arity}
(written $\arity{R}$); we write $\arity{\sigma} \defeq \max_{R \in \sigma}
\arity{R}$.
A \emph{fact} $R(\vec{a})$ (or \emph{$R$-fact}) consists of a relation $R\in \sigma$ and domain
elements $\vec{a}$, with $\card{\vec{a}} = \arity{R}$. We denote a (finite
or infinite) set of facts over~$\sigma$ by $\instance$.
We write $\elems{\instance}$ for the set of elements
mentioned in the facts in $\instance$.

We consider \emph{constraints} and \emph{queries} given in fragments of first-order
logic ($\fo$).
For simplicity, we disallow constants in constraints and queries, although our results
extend with them.
Given
a set of facts $\instance$ 
and 
a sentence $\phi$ in $\fo$, we talk of $\instance$ \emph{satisfying} $\phi$ in the
usual way.
The \emph{size} of $\varphi$, written $\mysize{\varphi}$, is defined
to be the number of symbols in $\varphi$.

The queries that we will use are \emph{conjunctive queries} (CQ), namely, existentially quantified
conjunction of atoms, which we restrict for simplicity to be Boolean.
We also allow \emph{unions of
conjunctive queries} (UCQs), namely, disjunctions of CQs.

\myparagraph{Problems considered}
Given a 
\emph{finite} set of facts $\instance_0$, constraints $\Sigma$ and query $Q$ (given as $\fo$
sentences), we say that 
$\instance_0$ and $\Sigma$ \emph{entail} $Q$  if for  every (possibly infinite) $\instance \supseteq \instance_0$
satisfying $\Sigma$,
 $\instance$ satisfies $Q$.
This amounts to asking whether $\calF_0 \wedge \Sigma \wedge \neg Q$ is
satisfiable (by a possibly infinite set of facts).
We write $\owqa(\instance_0, \Sigma, Q)$ for this decision problem, called
the \emph{query answering} problem.

In this paper, we study the $\owqa$ problem when imposing semantic constraints
on some \emph{distinguished} relations.
We thus split the signature as $\sigma \defeq \sigmab \sqcup \sigmad$, where $\sigmab$ is the \emph{base signature}
(its relations are the \emph{base} relations), and
$\sigmad$ is the \emph{distinguished} signature. All distinguished relations are
required to be binary,
and they will be
assigned special semantics. We study three kinds of special semantics.

We say \emph{$\instance_0, \Sigma$ entails $Q$ over transitive relations}, and
write $\owqatrans(\instance_0, \Sigma, Q)$ for the corresponding problem, if 
$\instance_0 \wedge \Sigma \wedge \neg Q$ is satisfied by some set of facts
$\instance$
where each distinguished relation $R_i^\trans$ in $\sigmad$ is required to be
\emph{transitive}\footnote{
Note that we work with \emph{transitive}
relations, which may not be \emph{reflexive}, unlike, e.g., $R^*$  roles in
$\mathcal{ZOIQ}$ description logics \cite{calvanese2009regular}.
However, our results can be adapted to the case of reflexive and transitive
predicates (and reflexive transitive closure).}
in~$\instance$.

We say \emph{$\instance_0, \Sigma$ entails $Q$ over transitive closure}, and
write
$\owqatc(\instance_0, \Sigma, Q)$ for this problem,
if the same holds on some $\instance$ where
each relation $R^\trans_i$
of~$\sigmad$ is interpreted as the transitive closure of a corresponding binary base
relation $R_i \in \sigmab$.

We say \emph{$\instance_0, \Sigma$ entails $Q$ over linear orders},
and write $\owqalin(\instance_0, \Sigma, Q)$,
if the same
holds on some $\instance$ where each relation $<_i \in \sigmad$ is
required to be a strict linear order on the elements
of~$\instance$.

We now define the constraint languages (which are all fragments of $\fo$)
for which we study these $\owqa$ problems.

\myparagraph{Dependencies}
The first constraint languages we study are
restricted classes of
\emph{tuple-generating dependencies} (TGDs).
A TGD is a first-order sentence $\tau$ of the form
$\forall \vec x~ (\bigwedge_i \gamma_i(\vec x) \rightarrow \exists \vec y ~ \bigwedge_i
\rho_i(\vec x, \vec y) ~)$
where $\bigwedge_i \gamma_i$ and $\bigwedge_i \rho_i$ are conjunctions of atoms
respectively called the \emph{body} and \emph{head} of~$\tau$.
 
We will be interested in TGDs that are \emph{guarded}
in various ways.
A \emph{guard} for $\vec{x}$ is an atom from $\sigma$
using every variable in $\vec{x}$.
We will be particularly interested in \emph{base-guards},
which are guards coming from the base relations $\sigmab$.

 A \emph{frontier-guarded TGD} ($\fgtgd$) is a TGD whose body contains a guard
for the \emph{frontier variables} --- variables that occur in both head and body.
It is a \emph{base frontier-guarded TGD} ($\afgtgd$) if there is a base atom
including all  the frontier variables.
We allow equality atoms $x=x$ to be guards, so $\afgtgd$ subsumes 
\emph{frontier-one TGDs}, which have
one frontier variable.
Frontier-guarded and frontier-one TGDs have been shown to have
an attractive combination of expressivity and computational properties~\cite{bagetcomplexityfg}.

We also introduce the more restricted class of 
\emph{base-covered frontier-guarded TGDs} ($\acfgtgd$):
they are the $\afgtgd$s
where, for every $\sigmad$-atom in the body,
there is a base atom in the body containing its variables (a \emph{base guard for the atom}).
Note that each $\sigmad$-atom may have a different base guard.

An important special case of frontier-guarded TGDs for applications are
\emph{inclusion dependencies} ($\incd$). An $\incd$ is a $\fgtgd$ where the body and head
contain a single atom, and no variable occurs twice in the same atom. A
\emph{base inclusion dependency} $\aincd$ is an $\incd$ where
the body atom is in  $\sigmab$, so the body atom serves as  the base-guard for the frontier
variables, while the  constraint is trivially base-covered.

\myparagraph{Guarded logics}
Moving beyond TGDs, the second kind of constraint that we study
are \emph{guarded logics}.

The \emph{guarded negation fragment} ($\gnf$) is the fragment of $\fo$
which contains all atoms, and is closed under
conjunction, disjunction, existential quantification,
and the following form of negation: for any $\gnf$ formula $\phi(\vec{x})$
and atom $A(\vec{x}, \vec{y})$ with free variables exactly as
indicated, the formula $ A(\vec{x}, \vec{y}) \wedge
\neg \phi(\vec{x})$ is in $\gnf$.
That is, existential quantification may be unguarded,
but the free variables in any negated subformula must be guarded; universal
quantification must be expressed with existential quantification and negation.
$\gnf$ can express all $\fgtgd$s, as well as non-TGD constraints and UCQs.
For instance, as it allows disjunction, $\gnf$ can express
\emph{disjunctive
inclusion dependencies}, $\did$s, which  generalize IDs: their body is a
single atom with no repeated variables, and their head is a disjunction of atoms
with no repeated variables. 

We introduce the \emph{base-guarded negation fragment} $\agnf$ over $\sigma$: it
is defined
like $\gnf$, but requires \emph{base  guards} instead of guards.
The \emph{base-covered guarded negation fragment} $\acgnf$
over $\sigma$
consists of $\agnf$ formulas such that
every $\sigmad$-atom $A$ that appears negatively (i.e., under the scope of an odd
number of negations) appears conjoined with a
base guard --- i.e., a $\sigmab$-atom containing all variables of $A$.
This technical condition is designed to generalize $\acfgtgd$s.
Indeed, a $\acfgtgd$ of the form $\forall \vec{x} (\bigwedge \gamma_i \rightarrow \exists \vec{y} \bigwedge \rho_i)$ can be written in $\acgnf$
as
$\neg \exists \vec{x} ( \bigwedge \gamma_i  \wedge \neg \exists \vec{y} \bigwedge \rho_i )$.

We call a CQ $Q$ \emph{base-covered} if each $\sigmad$-atom in $Q$ has a $\sigmab$-atom of $Q$
containing its variables.
This is the same as saying that $\neg Q$ is in $\acgnf$.
A UCQ is base-covered if each disjunct is.

\myparagraph{Examples}
We conclude the preliminaries by giving a few examples.
Consider a signature with a binary base relation $B$,
a ternary base relation $C$, and a distinguished relation $R^{\trans}$.
\begin{compactitem}
\item $\forall x y z \big( (R^{\trans}(x,y) \wedge R^{\trans}(y,z)) \rightarrow R^{\trans}(x,z) \big)$ is a TGD,
but is not a $\fgtgd$ since the frontier variables $x,z$ are not guarded.
It cannot even be expressed in $\gnf$.
\item $\forall x y  \big(R^{\trans}(x,y) \rightarrow B(x,y) \big)$
is an $\incd$, hence a $\fgtgd$.  It is not a $\aincd$ or even in $\agnf$, since the frontier
variables are not base-guarded.
\item $\forall x y z \big((B(z,x) \wedge R^{\trans}(x,y) \wedge R^{\trans}(y,z)) \rightarrow R^{\trans}(x,z) \big)$
is a $\afgtgd$. It is not a $\acfgtgd$
since there are no base atoms in the body to cover $x,y$ and $y,z$.
\item $\exists w x y z \big( R^{\trans}(w,x) \wedge R^{\trans}(x,y) \wedge R^{\trans}(y,z) \wedge R^{\trans}(z,w) \wedge C(w,x,y) \wedge C(y,z,w) \big)$
is a base-covered CQ.
\item $\exists x y
\big(
B(x,y) \wedge \neg (R^{\trans}(x,y) \wedge R^{\trans}(y,x) ) \wedge  (R^{\trans}(x,y) \vee R^{\trans}(y,x) )
\big)$ cannot be rewritten as a TGD. But it is  in $\acgnf$.
\end{compactitem}

Our main results are summarized in Table~\ref{tab:complexity}, and
the languages that we study are illustrated in Figure~\ref{fig:tax}.
In particular, $\owqatr$ and $\owqatc$ are decidable for $\agnf$.
This includes  base-frontier-guarded
rules, which allow one to use a transitive relation
such as ``part-of'' (or even its transitive closure) whenever
only one variable is to be exported to the head. This latter condition
holds in the translations of many classical description logics.
Our results also imply that $\owqalin$ is decidable for $\acgnf$,
which allows constraints  that arise from data integration and
data exchange over attributes with linear orders --- e.g.~views
defined by selecting rows of a table where some inequality involving
the attributes is satisfied.

\section{Decidability results for transitivity} \label{sec:decid}

We first consider $\owqatc$,
where
$\sigmab$ includes binary relations $R_1, \dots, R_n$,
and $\sigmad$ consists of binary relations $R^\trans_1, \dots, R^\trans_n$
such that $R^\trans_i$ is the transitive
closure of $R_i$.

\begin{theorem}
  \label{thm:decidtransautomata}
We can decide $\owqatc(\instance_0, \Sigma,Q)$ in $\twoexp$,
where $\instance_0$ ranges over finite sets of facts,
$\Sigma$ over $\agnf$ constraints, and $Q$ over UCQs.
In particular, this holds when $\Sigma$ consists of $\afgtgd$s.
\end{theorem}

In order to prove Theorem~\ref{thm:decidtransautomata},
we give a decision procedure to determine whether
$\instance_0 \wedge  \Sigma \wedge
\neg Q$
is satisfiable, when $R^\trans_i$ is interpreted as the transitive closure of $R_i$.
When $\Sigma \in \agnf$ and $Q$ is a Boolean UCQ, then
$\Sigma \wedge \neg Q$ is in $\agnf$.
So it suffices to show that $\agnf$ satisfiability is decidable,
when properly interpreting $R^\trans_i$.

As mentioned in the introduction, our proofs rely heavily on the fact that
in query answering problems for these constraint languages, one
can restrict to sets of facts that have a ``tree-like'' structure. We now make this notion precise.
A \emph{tree decomposition} of $\instance$ consists
of a tree $(T, \child)$  and a labelling function $\lambda$ associating
each node of the tree $T$ to a set of facts of $\instance$, called the \emph{bag} of that node,
that satisfies the following conditions:
\begin{inparaenum}[(i)]
\item each fact of $\instance$ must be in the image of $\lambda$;
\item for each element $e \in \elems{\instance}$,
the set of nodes  whose bag uses $e$ is a connected
subset of $T$.
\end{inparaenum}
It is \emph{$\instance_0$-rooted} if the root node is associated with $\instance_0$.
It has \emph{width $k-1$}
if each bag other than the root mentions at most $k$ elements.

For a number $k$, a $\sigma$ sentence $\phi$  is said to have
\emph{transitive-closure friendly $k$-tree-like witnesses} if:
for every finite set of facts $\instance_0$,
if there is an $\instance$
extending $\instance_0$ with additional $\sigmab$-facts
such that $\instance$ satisfies $\phi$
when each $R^\trans$ is interpreted as the transitive closure of $R$,
then there is such an $\instance$
that has an $\instance_0$-rooted $(k-1)$-width
tree decomposition.
We can show that $\agnf$ sentences have this kind of $k$-tree-like witness
for an easily computable $k$.
The proof uses a standard technique,
involving an unravelling based on ``guarded negation bisimulation''
\cite{gnficalp}:

\begin{proposition}
  \label{prop:transdecomp}
  Every sentence $\phi$ in $\agnf$ has  transitive-closure friendly $k$-tree-like witnesses,
    where $k \leq \mysize{\phi}$.
\end{proposition}

Here $k$ can be taken to be the ``width'' of $\phi$ \cite{gnficalp}, which is roughly the maximum number of free variables
in any subformula. 
Hence,
it suffices to test satisfiability for $\agnf$ restricted
to sets of facts with tree decompositions
of width $\mysize{\phi}-1$.
It is well known that sets of facts of bounded tree-width
can be encoded as trees over a finite alphabet that depends only on the signature and the tree-width.
This makes the problem amenable to tree automata
techniques,
since we can design a tree automaton that runs on
representations of these tree decompositions
and checks whether some sentence holds in
the corresponding set of facts.

\begin{theorem}\label{thm:automata}
Let $\phi$ be a sentence in $\agnf$,
and let $\instance_0$ be a finite set of facts.
We can construct in $\twoexp$
a 2-way alternating parity tree automaton $\calA_{\phi,\instance_0}$
such that\\[.3em]
\null\hfill$ \text{$\instance_0 \wedge \phi$ is satisfiable
\quad iff \quad
$L(\calA_{\phi,\instance_0}) \neq \emptyset$}$\hfill\null\\[.3em]
when each $R_i^\trans \in \sigmad$ is interpreted as the transitive closure
of $R_i \in \sigmab$.
The number of states of $\calA_{\phi,\instance_0}$ is exponential in $\mysize{\phi} \cdot \mysize{\instance_0}$
and the number of priorities is linear in $\mysize{\phi}$.
\end{theorem}

The construction can be viewed
as an extension of \cite{CalvaneseGV05},
and incorporates ideas from automata for guarded logics
(see, e.g., \cite{GradelW99}).

Because 2-way tree automata emptiness is decidable in time
exponential in the number of states and priorities \cite{Vardi98},
this yields the $\twoexp$ bound
for Theorem~\ref{thm:decidtransautomata}.

\myparagraph{Consequences for $\owqatr$}
We can derive results
for $\owqatrans$
by observing that the
$\owqatc$ problem subsumes it:
to enforce that $R^\trans \in \sigmad$ is transitive, simply interpret it as the
transitive closure of a relation~$R$ that is never otherwise used.
Hence:

\begin{corollary} \label{cor:decidetransgnf}
We can decide $\owqatr(\instance_0, \Sigma,Q)$ in $\twoexp$,
where $\instance_0$ ranges over finite sets of facts,
$\Sigma$ over $\agnf$ constraints (in particular, $\afgtgd$),
and $Q$ over UCQs.
\end{corollary}

In particular, this result holds for 
\emph{frontier-one TGDs} (those with a single frontier variable),
as a single variable is always base-guarded.
This answers a question of \cite{mugnier15}.

\myparagraph{Data complexity}
Our results in Theorem~\ref{thm:decidtransautomata} and
Corollary~\ref{cor:decidetransgnf} show upper bounds on the \emph{combined
complexity} of the $\owqatr$ and $\owqatc$ problems. We now turn to the complexity when
the query and constraints are fixed but the initial set of facts varies --- the \emph{data complexity}.

We first show
a $\conp$ data complexity upper bound for $\owqatc$ for $\agnf$ constraints.
The algorithm uses the fact
that a counterexample to $\owqatc$
can be taken to have a $\instance'$-rooted tree  decomposition, for some $\instance'$ that
does not add new elements to $\instance_0$, only  new facts. While such a decomposition could be large, it suffices
to guess $\instance'$  and annotations describing, for each $\mysize{\phi}$-tuple $\vec c$ in $\instance'$,  sufficiently many formulas
holding in the  subtree that interfaces with $\vec c$. 
The technique generalizes an
analogous result in \cite{vldb12}.

\begin{theorem}
  \label{thm:conptransdataupper}
  For any fixed $\agnf$ constraints $\Sigma$ and UCQ $Q$, given a finite
  set of facts $\instance_0$, we can decide $\owqatc(\instance_0, \Sigma, Q)$ in
  $\conp$ data complexity. 
\end{theorem}

For $\fgtgd$s, the data complexity of $\owqa$ is in $\ptime$ \cite{bagetcomplexityfg}.
We can show that the same holds, but only for $\acfgtgd$s, and for
$\owqatrans$ rather than $\owqatc$:
\newcommand{\ptimetransdataupper}{
  For any fixed $\acfgtgd$ constraints $\Sigma$ and base-covered UCQ $Q$, given a finite
  set of facts $\instance_0$, we can decide $\owqatrans(\instance_0, \Sigma, Q)$ in
  $\ptime$ data complexity.
}
\begin{theorem}
  \label{thm:ptimetransdataupper}
  \ptimetransdataupper
\end{theorem}
The proof uses a reduction to the standard $\owqa$ problem for
$\fgtgd$s, and then applies the $\ptime$ result of \cite{bagetcomplexityfg}. The reduction again makes use of tree-likeness to show that we can replace the requirement  that the 
$R_i^\trans$ are transitive by the weaker requirement of transitivity within small sets (intuitively, within 
bags of a decomposition). We will also use this idea for linear orders (see Proposition~\ref{prop:rewritelin}).

Restricting to $\owqatrans$ is in fact essential to make
data complexity tractable, as hardness holds otherwise.

\myparagraph{Hardness}
We now show complexity lower bounds. We already know that all our variants of
$\owqa$ are $\twoexp$-hard in combined complexity, and $\conp$-hard in data
complexity, when $\gnf$ constraints are allowed: this follows from existing bounds on $\gnf$
reasoning \cite{vldb12} even without distinguished predicates.
However, in the case of the $\owqatc$ problem, we can show the same result for the much
weaker language of $\aincd$s.

We do this via a reduction from $\owqa$ with
\emph{disjunctive inclusion dependencies}, which
is known to be $\twoexp$-hard in combined complexity \cite[Theorem
2]{bourhispieris} and $\conp$-hard in data complexity
\cite{calvanese2006data,bourhispieris}, even without distinguished relations. We
use the transitive closure to emulate disjunction (as was already suggested in
the description logic context \cite{oldhorrocks}) by creating an
$R^\trans_i$-fact and limiting the length of a witness $R_i$-path (this limit is imposed 
by $Q'$).
The choice of the length of the witness path among two possibilities is used to
mimic the disjunction.  
We thus show:

\newcommand{\tcdisj}{
  For any finite set of facts $\instance_0$,
  $\did$s $\Sigma$,
  and UCQ $Q$ on a signature $\sigma$,
  we can compute in $\ptime$ a set of facts $\instance_0'$,
  $\aincd$s $\Sigma'$,
  and a base-covered CQ $Q'$ on a signature $\sigma'$
  (with a single distinguished relation),
  such that $\owqa(\instance_0, \Sigma,
  Q)$ iff $\owqatc(\instance_0', \Sigma', Q')$.
}

\begin{theorem}
  \label{thm:tcdisj}
  \tcdisj
\end{theorem}

This implies the following, contrasting with  Theorem \ref{thm:ptimetransdataupper}:

\begin{corollary}
  The $\owqatc$ problem with $\aincd$s and base-covered CQs is $\conp$-hard in data
  complexity and $\twoexp$-hard in combined complexity.
\end{corollary}

In fact, the data complexity lower bound for $\owqatc$ even holds in the absence of
constraints:

\newcommand{\lindatacompltrans}{
  There is a base-covered CQ $Q$ such that the data complexity of
  $\owqatc(\instance_0, \emptyset, Q)$ is $\conp$-hard.
}
\begin{proposition} \label{prop:lindatacompltrans}
  \lindatacompltrans
\end{proposition} 

We prove this by reducing the problem of $3$-coloring a directed graph, known to be
$\np$-hard, to the complement of $\owqatc$. It is well-known how to do this
using TGDs that have disjunction in the head.  As in the proof of
Theorem~\ref{thm:tcdisj}, we simulate this disjunction by using a choice of  the length of
paths that realize transitive closure facts asserted in~$\instance_0$.

In all of these hardness results, we first prove them for UCQs, and then show how the use of disjunction
can be eliminated, using a prior ``trick'' (see, e.g., \cite{georgchristos}) 
to code the intermediate truth values of disjunctions within a CQ.

\section{Decidability results for linear orders} \label{sec:decidlin}

\renewcommand{\drel}{<}
We now move to $\owqalin$,
the setting where the distinguished relations $<_i$ of~$\sigmad$ are \emph{linear}
(total) strict orders, i.e., they are transitive, irreflexive, and total.
We consider constraints and queries
that are base-covered.
We prove the following result.

\begin{theorem}
  \label{thm:decidelindirect}
We can decide $\owqalin(\instance_0, \Sigma, Q)$ in $\twoexp$,
where
$\instance_0$ ranges over finite sets of facts,
$\Sigma$ over $\acgnf$,
and $Q$ over base-covered UCQs.
In particular, this holds when $\Sigma$ consists of $\acfgtgd$s.
\end{theorem}

Our technique here is to reduce
this to traditional $\owqa$
where no additional restrictions (like
being transitive or a linear order) are imposed.
Starting with $\acgnf$ constraints,
we reduce
to a traditional $\owqa$ problem with $\gnf$ constraints,
and hence prove decidability in $\twoexp$ using \cite{vldb12}.
However, the reduction is quite simple, and hence could be applicable to   other  constraint classes.

The idea behind the reduction is to include additional constraints
that enforce the linear order conditions.
However, we cannot express transitivity or totality in $\gnf$.
Hence, we will only add
a weakening of these properties that is expressible in $\gnf$,
and then argue that this is sufficient for our purposes.

The reduction is described in the following proposition.
\begin{proposition}
  \label{prop:rewritelin}
  For any finite set of facts $\instance_0$,
  constraints $\Sigma \in \acgnf$,
  and
  base-covered UCQ $Q$,
  we
  can compute $\instance_0'$ and $\Sigma' \in \agnf$ in $\ptime$
  such that
  $\owqalin(\instance_0, \Sigma, Q) \text{ iff } 
  \owqa(\instance_0', \Sigma', Q)$.
\end{proposition}

In particular,
$\instance_0'$ is $\instance_0$ together with facts $G(a,b)$ for every pair $a,b \in \elems{\instance_0}$,
where $G$ is some fresh binary base relation.
We define $\Sigma'$ as $\Sigma$ together with the \emph{$k$-guardedly linear axioms}
  for each distinguished relation~$\drel$, where $k$ is $\max(\mysize{\Sigma \wedge \neg
  Q}, \arity{\sigma \cup \{G\}})$; namely:

\newcommand{\glinaxioms}[1]{
\begin{#1}
\item guardedly total: \\
$\forall x y ( (\guardedbg(x,y) \wedge x \neq y) \rightarrow x \drel y \vee y \drel x )$
\item irreflexive: 
$\neg \exists x (x \drel x)$
\item $k$-guardedly transitive: 
for $1 \leq l \leq k-1$:\\
$\neg \exists x y ( \psi_l(x,y) \wedge \guardedbg(x,y) \wedge \neg (x \drel y)
  )$\\
and, for $1 \leq l \leq k$:
$\neg \exists x ( \psi_l(x,x) \wedge x = x \wedge \neg (x \drel x) )$
\end{#1}
where:
\begin{#1}
  \item $\guardedbg(x,y)$ is the formula expressing that $x,y$ is base-guarded (an existentially-quantified disjunction over
all possible base-guards containing $x$ and $y$);
  \item $\psi_1(x,y)$ is just $x < y$; and
  \item $\psi_l(x,y)$ for $l \geq 2$ is:
$\exists x_2 \dots x_{l} ( x \drel x_2 \wedge \dots \wedge x_{l} \drel y
)$.\end{#1}}
\glinaxioms{compactitem}
Unlike the property of being a linear order,
the $k$-guardedly linear axioms can be expressed in $\agnf$.

We now sketch the argument for the correctness
of the reduction.
The easy direction is where we assume $\owqa(\instance_0',\Sigma',Q)$ holds,
so any $\instance' \supseteq \instance_0'$ satisfying $\Sigma'$ must satisfy $Q$.
Now consider $\instance \supseteq \instance_0$ that satisfies $\Sigma$
and where all $\drel$ in $\sigmad$
are strict linear orders.
We must show that $\instance$ satisfies $Q$.
First, observe that $\instance$ satisfies $\Sigma'$
since the \mbox{$k$-guardedly linear} axioms for $\drel$
are clearly satisfied for all~$k$
when $\drel$ is a strict linear order.
Now consider the extension of $\instance$ to $\instance'$
with facts $G(a,b)$ for all $a,b \in \elems{\instance_0}$.
This must still satisfy $\Sigma'$:
adding these facts
means there are additional $k$-guardedly linear requirements
on the elements from~$\instance_0$,
but these requirements already hold
since $\drel$ is a strict linear order.
Hence, by our initial assumption, $\instance'$ must satisfy $Q$.
Since $Q$ does not mention $G$,
the restriction of $\instance'$ back to~$\instance$ still satisfies $Q$ as well.
Therefore, $\owqalin(\instance_0,\Sigma,Q)$ holds.

For the harder direction, suppose for the sake of contradiction
that $\owqa(\instance_0',\Sigma',Q)$ does not hold,
but $\owqalin(\instance_0,\Sigma,Q)$ does.
Then there is some $\instance' \supseteq \instance_0'$ such that
$\instance'$ satisfies $\Sigma' \wedge \neg Q$.
We will again rely on the ability to  restrict to tree-like $\instance'$,
but with
a slightly different notion of tree-likeness.

We say a set $E$ of elements from $\elems{\instance}$
are \emph{base-guarded}
in $\instance$
if there is some $\sigmab$-fact or $G$-fact in $\instance$ that mentions all of the elements in $E$.
A \emph{base-guarded-interface tree decomposition} $(T, \child ,\lambda)$ for $\instance$
is a tree decomposition satisfying the following additional property:
for all nodes $n_1$ that are not the root of $T$,
if $n_2$ is a child of $n_1$
and $E$ is the set of elements mentioned in both $n_1$ and $n_2$,
then $E$ is base-guarded in $\instance$.
A sentence $\phi$ has \emph{base-guarded-interface $k$-tree-like witnesses} if
for any finite
set of facts $\instance_0$,
  if there is some
  $\instance \supseteq \instance_0$ satisfying $\phi$
 then there  is
 such an $\instance$
  with an $\instance_0$-rooted $(k-1)$-width
   base-guarded-interface tree decomposition.

Although the transformation from $\Sigma$ to $\Sigma'$ makes the formula
larger, it does not increase the ``width'' that controls the bag size of tree-like witnesses.
Hence, we can show:

\begin{lemma}
  \label{lemma:guarded-interface-dec}
  The sentence $\Sigma' \wedge \neg Q$ has base-guarded-interface $k$-tree-like witnesses for
   $k = \max(\mysize{\Sigma \wedge \neg Q}, \arity{\sigma \cup \{G\}})$.
\end{lemma}

Using this lemma, we can assume that we have some $\instance' \supseteq \instance_0'$
which has a $(k-1)$-width base-guarded-interface tree decomposition
and witnesses $\Sigma' \wedge \neg Q$.
If every $\drel$ in $\sigmad$ is a strict linear order in $\instance'$,
then restricting $\instance'$ to the set of $\sigma$-facts
yields some $\instance$ that would satisfy $\Sigma \wedge \neg Q$, a contradiction.
Hence, there are some distinguished relations
$\drel$ that are not strict linear orders in $\instance'$.
We can show that such an $\instance'$ can actually
be extended to some $\instance''$ that still satisfies $\Sigma' \wedge \neg Q$
but where all $\drel$ in $\sigmad$ are strict linear orders,
which we already argued is impossible.

The crucial part of the argument is thus about extending
$k$-guardedly linear counterexamples to genuine linear orders:

\begin{lemma} \label{lemma:reducelin}
If there is $\instance' \supseteq \instance_0'$
that satisfies $\Sigma' \wedge \neg Q$
and has a $\instance_0'$-rooted base-guarded-interface $(k-1)$-width tree decomposition,
then there is $\instance'' \supseteq \instance'$
that satisfies $\Sigma' \wedge \neg Q$
where each distinguished relation is a strict linear order.
\end{lemma}

The proof of Lemma \ref{lemma:reducelin} proceeds
by showing that sets of facts that have $(k-1)$-width base-guarded-interface tree decompositions
and satisfy $k$-guardedly linear axioms
must already be cycle-free with respect to $\drel$.
Hence, by taking the transitive closure of $\drel$ in $\instance$,
we get a new set of facts
where every $\drel$ is a strict \emph{partial} order.
Any strict partial order can be further extended to a strict linear order using known techniques,
so we can obtain $\instance'' \supseteq \instance'$
where $\drel$ is a strict partial order.
This $\instance''$ may have more $\drel$-facts than $\instance'$,
but the $k$-guardedly linear axioms
ensure that these new $\drel$-facts are only about pairs of elements
that are not base-guarded.

It remains to show that $\instance''$ satisfies $\Sigma' \wedge \neg Q$.
It is clear that $\instance''$ still satisfies the $k$-guardedly linear axioms,
but it could no longer satisfy $\Sigma \wedge \neg Q$.
However,
this is where the base-covered assumption on $\Sigma \wedge \neg Q$ is used:
it can be shown that satisfiability of $\Sigma \wedge \neg Q$ in $\acgnf$
is not affected by adding new $\drel$-facts
about pairs of elements that are not base-guarded.

\myparagraph{Data complexity}
Again, the result of Theorem~\ref{thm:decidelindirect} is a combined complexity
upper bound. However, as it
works by reducing to traditional $\owqa$ in $\ptime$,
data complexity upper bounds follow
from \cite{vldb12}.

\begin{corollary}
  For any $\acgnf$ constraints $\Sigma$ and base-covered UCQ $Q$, given a finite
  set of facts $\instance_0$, we can decide $\owqalin(\instance_0, \Sigma, Q)$ in
  $\conp$ data complexity. 
\end{corollary}

This is similar to the way data complexity bounds were shown for $\owqatr$
(in Theorem~\ref{thm:ptimetransdataupper}).
However, unlike for the $\owqatr$ problem, the constraint rewriting in this
section introduces disjunction, so rewriting a $\owqalin$ problem for $\acfgtgd$s does not produce a classical query answering problem for
$\fgtgd$s. Thus the rewriting does not  imply a $\ptime$ data complexity upper bound for
$\acfgtgd$. Indeed, we will see in Proposition~\ref{prop:lindatacompl} that this is $\conp$-hard.

\myparagraph{Hardness}
$\owqalin$ for $\acgnf$ constraints is
again immediately $\conp$-hard in data complexity, and $\twoexp$-hard in combined
complexity, from the corresponding bounds on $\gnf$ \cite{vldb12}. However, we can
show that hardness holds for the much weaker constraint language $\aincd$,
by a reduction from $\did$ reasoning, as in Section~\ref{sec:decid}.

\newcommand{\lindisj}{
  For any finite set of facts $\instance_0$,
  $\did$s $\Sigma$,
  and UCQ $Q$ on a signature $\sigma$,
  we can compute in $\ptime$ a set of facts $\instance_0'$,
  $\aincd$s $\Sigma'$,
  and CQ $Q'$ on a signature $\sigma'$ (with a single distinguished relation),
  such that $\owqa(\instance_0, \Sigma,
  Q)$ iff $\owqalin(\instance_0', \Sigma', Q')$.
}

\begin{theorem}
  \label{thm:lindisj}
  \lindisj
\end{theorem}

The reduction allows us to transfer hardness results for $\did$ from \cite{calvanese2006data,bourhispieris}, exactly as
was done in Theorem \ref{thm:tcdisj}, to conclude:

\begin{corollary}
  The $\owqalin$ problem with $\aincd$ and
  base-covered CQs is $\conp$-hard in data
  complexity and $\twoexp$-hard in combined complexity.
\end{corollary}

Again, as in the previous section, the data complexity lower bound even holds in the absence of
constraints:

\newcommand{\lindatacompl}{
  There is a base-covered CQ $Q$ such that the data complexity of
  $\owqalin(\instance, \emptyset, Q)$ is $\conp$-hard.
}
\begin{proposition} \label{prop:lindatacompl}
  \lindatacompl
\end{proposition}

\section{Undecidability results for transitivity} \label{sec:undecidtrans}

We have shown in Section~\ref{sec:decid}
that query answering is decidable with transitive relations (even
with transitive closure), $\afgtgd$s, and UCQs
(Theorem~\ref{thm:decidtransautomata}).
Removing our base-guarded
condition leads to undecidability of $\owqatc$, even when
constraints are inclusion dependencies:

\newcommand{\undectrans}{
  There is a signature $\sigma = \sigmab \sqcup \sigmad$ with a single
  distinguished predicate $S^\trans$ in~$\sigmad$,
  a set $\Sigma$ of
  $\incd$s on~$\sigma$, and a CQ $Q$ on $\sigmab$,
  such that the following problem is undecidable:
  given a finite set of facts $\instance_0$,
  decide $\owqatc(\instance_0, \Sigma, Q)$.
}

\begin{theorem}
  \label{thm:undectrans}
  \undectrans
\end{theorem}

The proof is by reduction from a tiling problem. The constraints use a
transitive successor relation to define a grid of integer pairs. It then uses
transitive closure to emulate disjunction, as in Theorem~\ref{thm:tcdisj}, and
relies on $Q$ to test for forbidden adjacent tile patterns.

An analogous result can be shown for $\owqatr$, using
(non-base-guarded) disjunctive inclusion dependencies:

\newcommand{\undectransb}{
  There is an arity-two signature $\sigma = \sigmab \sqcup \sigmad$
  with a single distinguished predicate $S^\trans$ in~$\sigmad$,
  a set $\Sigma$ of $\did$s on~$\sigma$,
  a CQ $Q$ on $\sigmab$, such that the following problem is
  undecidable:
  given a finite set of facts $\instance_0$,
  decide $\owqatr(\instance_0, \Sigma, Q)$.
}
\begin{theorem}
  \label{thm:undectransb}
  \undectransb
\end{theorem}

These results complement the undecidability results of 
\cite[Theorem~2]{andreaslidia}, which showed that, on arity-two signatures, 
$\owqatr$ is undecidable with guarded TGDs and atomic CQs,
even when transitive relations occur only in guards.
Our results also contrast with the decidability results of 
\cite{mugnier15} which apply to $\owqatr$: our Theorem~\ref{thm:undectrans} shows that
their results cannot extend to $\owqatc$.
\section{Undecidability results for linear orders} \label{sec:undecidlin}

Section~\ref{sec:decidlin} has shown that $\owqalin$ is decidable for
base-covered CQs and $\acgnf$ constraints. Dropping the
base-covered requirement on the query leads to undecidability:

\newcommand{\undeccq}{
  There is a signature $\sigma = \sigmab \sqcup \sigmad$ where
  $\sigmad$ is a single strict linear order
  relation, a CQ $Q$ on~$\sigma$, and a set $\Sigma$ of
  inclusion dependencies on $\sigmab$ (i.e., not mentioning the linear order, so in
  particular base-covered), such that the following problem is undecidable:
  given a finite set of facts $\instance_0$, decide
  $\owqalin(\instance_0, \Sigma, Q)$.
}
\begin{theorem}
  \label{thm:undeccq}
  \undeccq
\end{theorem}

This result is close to \cite[Theorem 3]{egoretal}, which
deals not with a  linear order, but inequalities in queries, which we
can express with a linear order.
However, this requires a
UCQ. As in our prior hardness and undecidability results, we can adapt the technique to use a CQ.

By letting $\Sigma' \defeq \Sigma \wedge \neg Q$ where $\Sigma$ and $Q$ are as in the
previous theorem, we obtain base-guarded constraints
for which $\owqalin$ is undecidable.
In fact, $\Sigma'$ can be expressed as a set of $\afgtgd$s.
This implies that the base-covered requirement is necessary for the constraint
language:

\begin{corollary}\label{cor:undeclin}
  There is a signature $\sigma = \sigmab \sqcup \sigmad$ where $\sigmad$ is a single strict
  linear order relation, and a set $\Sigma'$ of $\afgtgd$ constraints,
  such that, letting $\top$ be the tautological query, the following problem is
  undecidable:
  given a finite set of facts $\instance_0$, decide $\owqalin(\instance_0,
  \Sigma', \top)$.
\end{corollary}

\section{Conclusion} \label{sec:conc}
We have given a detailed picture of the impact of transitivity, transitive closure,
and linear order restrictions on query answering problems for a broad class of
guarded constraints.
We have shown that transitive relations and transitive closure restrictions
can be handled in  guarded constraints as long
as they are not used in guards. For linear
orders, the same is true if order atoms are covered by base atoms. This
implies the analogous results for frontier-guarded TGDs, in particular
frontier-one. But in the linear order case we show that $\ptime$ data complexity cannot always be preserved.

We leave open the question of entailment over \emph{finite} sets of facts. 
There are few techniques for deciding entailment over finite sets of facts for logics where
it does not coincide with general entailment (and for the constraints considered here it does not coincide).
An exception is \cite{vincemikolaj}, but it is not clear if the techniques there can be extended to our constraint languages.

\section*{Acknowledgments}
Amarilli was partly funded by the T\'{e}l\'{e}com ParisTech Research Chair on Big Data
and Market Insights.
Bourhis was supported by CPER Nord-Pas de Calais/FEDER DATA Advanced Data Science and Technologies 2015-2020 and the ANR Aggreg Project ANR-14-CE25-0017, INRIA Northern European Associate Team Integrated Linked Data.
Benedikt's work was sponsored by the Engineering and Physical Sciences
Research Council of the United Kingdom (EPSRC), grants EP/M005852/1 and EP/L012138/1.
Vanden Boom was partially supported by EPSRC grant EP/L012138/1.

\bibliographystyle{named}
\bibliography{algs}

\appendix

\clearpage

\noindent The proofs for the results stated in the main paper
are provided in this appendix.

\section{Normal form}\label{app:nf}
\newcommand{\nf}{\text{normal form}\xspace}

The proofs make use of the fact that the fragments of $\gnf$
that we consider can be converted into a normal form
that is related to the $\gn$ normal form
introduced in the original paper on $\gnf$ \cite{gnficalp}.
The idea is that $\gnf$ formulas can be seen as being built
up from atoms using guarded negation, disjunction, and CQs.
We introduce this normal form here,
and discuss some related notions
we will use in the proofs.

First, the
\emph{guardedness predicate} $\guarded(\vec{x})$ asserts that
$\vec{x}$ is guarded by some $\sigma$-atom;
it can be seen as an abbreviation for the disjunction
of existentially quantified relational atoms from~$\sigma$ involving
all of the variables from~$\vec{x}$.
We write $\guardedb(\vec{x})$ for the corresponding guardedness predicate
restricted to $\sigmab$.

The \emph{$\nf$ for $\agnf$} over $\sigma$
starts with $\sigmab$-atoms
and builds up via the following rules:
\begin{itemize}
\item If $\phi_1(\vec x)$ and $\phi_2(\vec x)$ are in $\nf$ $\agnf$,
then $\phi_1(\vec{x}) \vee \phi_2(\vec{x})$ are in $\nf$ $\agnf$.
\item  If $\phi(\vec x)$ is in $\nf$ $\agnf$ and
$A(\vec{x})$ is a $\sigmab$-atom or the $\sigmab$-guardedness predicate,
then $A(\vec{x}) \wedge \neg \phi(\vec{x})$ is in $\nf$ $\agnf$.
\item If $\delta$ is a CQ over signature $\sigma \cup \{Y_1,\ldots,Y_n\}$,
and $\phi_1, \ldots, \phi_n$ are in $\nf$ $\agnf$,
and for each $Y_i$ atom in $\delta$
there is some $\sigmab$-atom or $\sigmab$-guardedness predicate
in $\delta$ that contains its free variables,
then $\delta[Y_1 \mapsfrom \phi_1, \ldots, Y_n \mapsfrom \phi_n]$ is in $\nf$ $\agnf$.
We call $\delta[Y_1 \mapsfrom \phi_1, \ldots, Y_n \mapsfrom \phi_n]$ a \emph{CQ-shaped formula}.
\end{itemize}

Likewise, the \emph{$\nf$ for $\acgnf$}
over $\sigma$
consists of $\nf$ $\agnf$ formulas such that
for every CQ-shaped subformula $\delta$ that appears negatively
(in the scope of an odd number of negations),
and for every conjunct $\beta$ in $\delta$,
there must be some $\sigmab$-atom or $\sigmab$-guardedness predicate
in $\delta$ that contains the free variables of $\beta$.

\myparagraph{Width and CQ-rank}
For $\varphi$ in $\nf$ $\agnf$, we define the
\emph{width} of $\varphi$ to be the maximum number
of free variables in any subformula of $\varphi$.
The \emph{CQ rank} of $\varphi$ is
the maximum number of conjuncts in any CQ-shaped subformula
$\exists \vec{x} ( \bigwedge \gamma_i )$
where $\vec{x}$ is non-empty.
These will be important parameters in later proofs.

We write $\agnfk$ to denote \emph{$\nf$ $\agnf$ formulas
of width $k$}, and similarly for $\acgnfk$.

\myparagraph{Conversion into $\nf$}
Observe that formulas in $\afgtgd$ or $\acfgtgd$ are of the form
$\forall \vec{x} (\bigwedge \gamma_i \rightarrow \exists \vec{y} \bigwedge \rho_i)$
already
and thus can be naturally written in $\nf$ $\agnf$ or $\acgnf$ as
$\neg \exists \vec{x} ( \bigwedge \gamma_i  \wedge \neg \exists \vec{y} \bigwedge \rho_i )$,
with no blow-up in the size or width.

In general, $\agnf$ formulas can be converted into $\nf$,
but with an exponential blow-up in size.

\begin{proposition}\label{prop:nf}
Let $\varphi$ be a formula in $\agnf$.
We can construct an equivalent $\varphi'$ in $\nf$
in $\exptime$ such that
\begin{itemize}
\item $\mysize{\varphi'}$ is at most exponential in $\mysize{\varphi}$;
\item the width of $\varphi'$ is at most $\mysize{\varphi}$;
\item the CQ-rank of $\varphi'$ is at most $\mysize{\varphi}$;
\item if $\varphi$ is in $\acgnf$, then $\varphi'$ is in $\nf$ $\acgnf$.
\end{itemize}
\end{proposition}

\begin{proof}[Proof sketch]
The conversion works by using the same rewrite rules as in \cite{gnficalp}:
\begin{align*}
\exists x (\theta \vee \psi) &\to (\exists x \theta) \vee (\exists x \psi) \\
\theta \wedge (\psi \vee \chi) &\to (\theta \wedge \psi) \vee (\theta \wedge \chi) \\
\exists x (\theta) \wedge \psi &\to \exists x' (\theta[x'/x] \wedge \psi ) \text{ where $x'$ is fresh}
\end{align*}

The size, width, and CQ-rank bounds after performing this rewriting
are straightforward to check.

The rewrite rules preserve the polarity of subformulas,
which helps ensure that coveredness is preserved during this conversion.
\end{proof}

\clearpage

\section{Transitive-closure friendly tree decompositions for $\agnf$
(Proof of Proposition~\ref{prop:transdecomp})}\label{app:tctreelike}

\newcommand{\gnk}{\gn^k}
\newcommand{\mydom}[1]{\dom(#1)}

Recall the statement of Proposition~\ref{prop:transdecomp}:

\begin{quote}
{\bf Proposition~\ref{prop:transdecomp}.} 
  Every sentence $\phi$ in $\agnf$ has  transitive-closure friendly $k$-tree-like witnesses,
    where $k \leq~\mysize{\phi}$.
\end{quote}

That is,
for every $\varphi$ in $\agnf$
and for every finite
set of facts $\instance_0$,
  if there is any
  $\instance$ extending $\instance_0$ with $\sigmab$-facts and
  satisfying $\phi$ when each relation $R^\trans$ is interpreted as the transitive closure
of $R$, then there  is
  some $\instance$ like this
  that has a $\instance_0$-rooted $(k-1)$-width
   tree decomposition.

If $\mysize{\phi} < 3$, then $\phi$ is
necessarily a single 0-ary relation or its negation, in which case the
result is trivial, with $k = 1$.
Hence, in the rest of this section,
we will assume that $\mysize{\varphi} \geq 3$
and $k$ will be chosen such that
$3 \leq k \leq \mysize{\varphi}$
($k$ will be an upper bound on the maximum number of free variables
in any subformula of $\varphi$).

The proof uses a standard technique,
involving an unravelling related to a variant of guarded negation bisimulation
\cite{gnficalp}.
A related result and proof also appears in \cite{lics16-gnfpup}.

\myparagraph{Bisimulation game}
The
\emph{$\gn^k$ bisimulation game}
between sets of facts $\fA$ and $\fB$
is an infinite game 
played by two players,
Spoiler and Duplicator.
The game has two types of positions: 
\begin{itemize}
 \item[i)] partial isomorphisms $f:\fA\restrict{}{X} \to \fB\restrict{}{Y}$ or $g:\fB\restrict{}{Y} \to \fA\restrict{}{X}$, 
     where $X \subset \elems{\fA}$ and $Y \subset \elems{\fB}$ are both finite and are $\sigmab$-guarded;
 \item[ii)] partial rigid homomorphisms $f:\fA\restrict{}{X} \to \fB\restrict{}{Y}$ or $g:\fB\restrict{}{Y} \to \fA\restrict{}{X}$, 
     where $X \subset \elems{\fA}$ and $Y \subset \elems{\fB}$ are both finite and are of size at most $k$.  
\end{itemize}
A \emph{partial rigid homomorphism} is a partial homomorphism with respect to all $\sigma$-facts,
such that
the restriction to any $\sigmab$-guarded set of elements is a partial isomorphism.

From a type~(i) position $h$,
Spoiler must choose a finite subset $X \subset \elems{\fA}$ or 
a finite subset $Y \subset \elems{\fB}$, in either case of size at most $k$, 
upon which Duplicator 
must respond by a partial rigid homomorphism with domain $X$ or $Y$ accordingly,
mapping it into the 
other set of facts in a manner consistent with~$h$. 

From a type~(ii) position $h : X \to Y$ (respectively, $h : Y \to X$),
Spoiler must choose a finite subset $X \subset \elems{\fA}$
(respectively, $Y \subset \elems{\fB}$)
of size at most $k$,
upon which Duplicator 
must respond by a partial rigid homomorphism with domain $X$
(respectively, domain $Y$),
mapping it into the 
other set of facts in a manner consistent with $h$.

Notice that a type~(i) position is a special kind of type~(ii) position
where Spoiler has the option to \emph{switch the domain} to the other set of facts,
rather than just continuing to play in the current domain.

Spoiler wins if he can force the play into a position from which Duplicator cannot 
respond, and Duplicator wins if she can continue to play indefinitely.

A winning strategy for Duplicator in the $\gn^k$ bisimulation game
implies agreement between $\fA$ and $\fB$ on certain
$\agnf$ formulas.

\begin{proposition}\label{prop:bisim-game}
Let $\varphi(\vec{x})$ be a formula in $\agnf$,
and let $k \geq 3$ be greater than or equal to
the maximum number of free variables in any subformula of $\varphi$.

If Duplicator has a winning strategy
in the $\gn^k$ bisimulation game
between $\fA$ and $\fB$
starting from a type~(i) or (ii) position $\vec{a} \mapsto \vec{b}$
and $\fA$ satisfies $\varphi(\vec{a})$ when interpreting each $R^\trans \in \sigmad$ as the transitive closure of $R \in \sigmab$,
then $\fB$ satisfies $\varphi(\vec{b})$ when interpreting each $R^\trans \in \sigmad$ as the transitive closure of $R \in \sigmab$.
\end{proposition}

\begin{proof}
For this proof,
when we talk about sets of facts satisfying a formula,
we mean satisfaction when interpreting $R^\trans \in \sigmad$ as the transitive closure of $R \in \sigmab$.
We will abuse terminology slightly and say that $\varphi$ has width $k$
if the maximum number of free variables in any subformula of $\varphi$ is at most $k$
(this is abusing the terminology since we are not assuming in this proof that $\varphi$ is in $\nf$).

We proceed by induction on the number of $\sigmad$-atoms in $\varphi$
and the size of $\varphi$.

Suppose Duplicator has a winning strategy in the $\gn^k$ bisimulation game
between $\fA$ and $\fB$.

If $\varphi$ is a $\sigmab$-atom $A(\vec{x})$,
the result follows from the fact that the position
$\vec{a} \mapsto \vec{b}$ is a partial homomorphism.

Suppose $\varphi$ is a $\sigmad$-atom $R^\trans(x_1,x_2)$,
and $\vec{a} = a_1 a_2$ and $\vec{b} = b_1 b_2$.
If $\fA, \vec{a}$ satisfies $R^\trans(x_1,x_2)$,
there is some $n \in \mathbb{N}$ such that $n > 0$ and
there is an $R$-path of length $n$
between $a_1$ and $a_2$ in $\fA$.
We can write a formula $\psi_n(x_1,x_2)$ in $\agnf$
(without any $\sigmad$-atoms)
that is satisfied exactly when there is an $R$-path of length $n$.
Since we do not need to write this in $\nf$,
we can express $\psi_n$ in $\agnf$ with width $3$
(maximum of 3 free variables in any subformula).
Since
$\fA,\vec{a}$ satisfies $\psi_n$
and $\psi_n$ does not have any $\sigmad$-atoms and $k \geq\nolinebreak 3$,
we can apply the inductive hypothesis from the type~(ii) position
$\vec{a} \mapsto \vec{b}$
to ensure that
$\fB,\vec{b}$ satisfies $\psi_n$,
and hence $\fB,\vec{b}$ satisfies~$\varphi$.

If $\varphi$ is a disjunction,
the result follows easily from the inductive hypothesis.

Suppose $\varphi$ is a base-guarded negation
$A(\vec{x}) \wedge \neg \varphi'(\vec{x}')$.
By definition of $\agnf$, it must be the case that $A \in \sigmab$
and $\vec{x}'$ is a sub-tuple of $\vec{x}$.
Since $\fA,\vec{a}$ satisfies $\varphi$,
we know that $\fA,\vec{a}$ satisfies $A(\vec{x})$,
and hence $\vec{a}$ is $\sigmab$-guarded.
This means that $\vec{a} \mapsto \vec{b}$ is actually a partial isomorphism,
so we can view it as a position of type~(i).
This ensures that
$\fB,\vec{b}$ also satisfies $A(\vec{x})$.
It remains to show that it satisfies $\neg \varphi'(\vec{x}')$.
Assume for the sake of contradiction that it satisfies $\varphi'(\vec{x}')$.
Because $\vec{a} \mapsto \vec{b}$ is a type~(i) position, we can
consider the move in the game where Spoiler switches the domain to the other set of facts,
and then restricts to the elements in the subtuple $\vec{b}'$ of $\vec{b}$
corresponding to $\vec{x}'$ in $\vec{x}$.
Let $\vec{a}'$ be the corresponding subtuple of $\vec{a}$.
Duplicator must have a winning strategy
from the type~(i) position $\vec{b}' \mapsto \vec{a}'$,
so the inductive hypothesis ensures that
$\fA, \vec{a}'$ satisfies $\varphi'(\vec{x}')$,
a contradiction.

Finally, suppose $\varphi$
is an existentially quantified formula
$\exists y ( \varphi'(\vec{x},y) )$.
We are assuming that $\fA, \vec{a}$ satisfies $\varphi$.
Hence,
there is some $c \in \elems{\fA}$ such that
$\fA,\vec{a}, c$ satisfies $\varphi'$.
Because the width of $\varphi$ is at most $k$,
we know that the combined number of elements in $\vec{a}$ and $c$ is at most $k$.
Hence, we can consider the move in the game where Spoiler
selects the elements in $\vec{a}$ and $c$.
Duplicator must respond with $\vec{b}$ for $\vec{a}$,
and some $d$ for $c$.
This is a valid move in the game,
so Duplicator must still have a winning strategy from this position $\vec{a}c \mapsto \vec{b}d$,
and the inductive hypothesis implies that $\fB, \vec{b}, d$ satisfies $\varphi'$.
Consequently, $\fB, \vec{b}$ satisfies $\varphi$.
\end{proof}

\myparagraph{Unravelling}
The tree-like witnesses
can be obtained using an unravelling construction
related to the $\gn^k$ bisimulation game.
This unravelling construction is adapted from \cite{lics16-gnfpup}.

Fix a set of facts $\instance \supseteq \instance_0$.
Consider the set $\Pi$ of sequences of the form $X_0 X_1 \dots X_n$,
where $X_0 = \elems{\instance_0}$, and for all $i \geq 1$,
$X_i \subseteq \elems{\fA}$ with $\mysize{X_i} \leq k$.

We can arrange these sequences in a tree based on the prefix order.
Each sequence $\pi = X_0 X_1 \dots X_n$ identifies a unique node in the tree;
we say $a$ is \emph{represented} at node $\pi$
if $a \in X_n$.
For $a \in \elems{\fA}$, we say $\pi$ and $\pi'$ are \emph{$a$-equivalent}
if $a$ is represented at every node on the
unique minimal path between $\pi$ and $\pi'$
in this tree.
For $a$ represented at $\pi$, we write $[\pi, a]$ for the $a$-equivalence class.

The \emph{$\gn^k$-unravelling of $\fA$} is a set of facts $\structureunravelk{\fA}$
over elements $\set{ [\pi, a] : \pi \in \Pi \text{ and } a \in \elems{\fA}}$.
The fact
$R([\pi_1,a_1],\dots,[\pi_j,a_j]) \in \structureunravelk{\fA}$
iff $R(a_1, \dots, a_j) \in \fA$ and
there is some $\pi \in \Pi$ such that for all $i$, $[\pi,a_i] = [\pi_i,a_i]$.
We can identify $[\epsilon,a]$ with the element $a \in \elems{\instance_0}$,
so there is a natural $\instance_0$-rooted tree decomposition of width $k-1$
for $\structureunravelk{\fA}$
induced by the tree of sequences from $\Pi$.

Because this unravelling is related so closely to the $\gn^k$-bisimulation game,
it is straightforward to show that Duplicator
has a winning strategy in the bisimulation game
between $\fA$ and its unravelling.

\begin{proposition}\label{prop:unravel}
Duplicator has a winning strategy in the $\gn^k$ bisimulation game
between $\fA$ and $\structureunravelk{\fA}$.
\end{proposition}

\begin{proof}
Given a position $f$ in the $\gnk$-bisimulation game,
we say the \emph{active set} is the set of facts
containing the elements in the domain of $f$.
In other words, the active set is either $\fA$ or $\structureunravelk{\fA}$,
depending on which set Spoiler is currently playing in.
The \emph{safe positions} $f$ in the $\gnk$-bisimulation game
between $\fA$ and $\structureunravelk{\fA}$ are defined as follows:
if the active set is $\structureunravelk{\fA}$,
then $f$ is safe if for all $[\pi,a] \in \mydom{f}$, $f([\pi,a]) = a$;
if the active set is $\fA$,
then $f$ is safe if
there is some $\pi$ such that $f(a) = [\pi,a]$ for all $a \in \mydom{f}$.

We now argue that starting from a safe position $f$,
Duplicator has a strategy to move to a new safe position $f'$.
This is enough to conclude that Duplicator has a winning strategy
in the $\gnk$-bisimulation game
between $\fA$ and $\structureunravelk{\fA}$
starting from any safe position.

First, assume that the active set is $\structureunravelk{\fA}$.
\begin{itemize}
\item If $f$ is a type~(ii) position,
then Spoiler can select some new set $X'$ of elements from the active set.
Each element in $X'$ is of the form $[\pi',a']$.
Duplicator must choose $f'$ such that $[\pi',a']$ is mapped to $a'$ in $\fA$,
in order to maintain safety.
This new position $f'$ is consistent with $f$ on any elements in $X' \cap \mydom{f}$
since $f$ is safe.
This $f'$ is still a partial homomorphism since any relation holding for a tuple of
elements $[\pi_1,a_1], \dots, [\pi_n,a_n]$ from $\mydom{f'}$ must hold
for the tuple of elements $a_1, \dots, a_n$ in $\fA$ by definition of $\structureunravelk{\fA}$.
Consider some element $[\pi',a']$ in $\mydom{f'}$.
It is possible that there is some $[\pi,a']$ in $\mydom{f'}$ with $[\pi,a'] \neq [\pi',a']$;
however, $[\pi,a']$ and $[\pi',a']$ are not base-guarded in $\structureunravelk{\fA}$.
Hence, any restriction $f''$ of $f'$ to
a base-guarded set of elements is a bijection.
Moreover, such an $f''$ is a
partial isomorphism:
consider some $a_1,\dots,a_n$ in the range of $f''$ for which some relation $S$ holds in $\fA$;
since $(f'')^{-1}(a_1),\dots,(f'')^{-1}(a_n)$ must be base-guarded,
we know that there is some $\pi$ such that
$[\pi,a_1] = (f'')^{-1}(a_1)$,$\dots$, $[\pi,a_n] = (f'')^{-1}(a_n)$,
so by definition of $\structureunravelk{\fA}$,
$S$ holds of $(f'')^{-1}(a_1),\dots,(f'')^{-1}(a_n)$ as desired.
Hence, $f'$ is a safe partial rigid homomorphism.
\item If $f$ is a type~(i) position,
then Spoiler can either choose elements in the active set
and we can reason as we did for the type~(ii) case,
or Spoiler can select elements from the other set of facts.

We first argue that if Spoiler changes the active set
and chooses no new elements,
then the game is still in a safe position.
Since $f$ is a type~(i) position, 
we know that $\mydom{f}$ is guarded
by some base relation $S$, so there is some $\pi$ with $f(a) = [\pi,a]$ for all $a \in \mydom{f}$
by construction of $\structureunravelk{\fA}$.
Hence, the new position $f'=f^{-1}$ is still safe.

If Spoiler switches active sets and chooses new elements,
then we can view this as two separate moves:
in the first move, Spoiler switches active sets from $\structureunravelk{\fA}$ to $\fA$
but chooses no new elements,
and in the second move, Spoiler selects the desired new elements from $\fA$.
Because switching active sets leads to a safe position
(by the argument in the previous paragraph),
it remains to define Duplicator's safe strategy when the active set is $\fA$,
which we explain below.
\end{itemize}
Now assume that the active set is $\fA$.
Since $f$ is safe, there is some $\pi$ such that $f(a) = [\pi,a]$ for all $a \in \mydom{f}$.
\begin{itemize}
\item If $f$ is a type~(ii) position, then Spoiler can select some new set $X'$
of elements from the active set.
We define the new position $f'$ chosen by Duplicator
to map each element $a' \in X'$
to $[\pi',a']$ where $\pi' = \pi \cdot X'$.
By construction of the unravelling, $\pi' \in \Pi$
and the resulting partial mapping $f'$ still satisfies the safety property with $\pi'$
as witness.
Note that $f'$ is consistent with $f$ for elements of~$X'$ that are also
in~$\mydom{f}$, as we have
$[\pi \cdot X',a'] = [\pi,a']$ for $a' \in X' \cap \mydom{f}$.
Now consider some tuple $\vec{a} = a_1 \dots a_n$
of elements from $\mydom{f'}$ that are in some relation $S$.
We know that $f'(a_i) = [\pi',a_i]$,
hence $S$ must hold for $f'(\vec{a})$ in $\structureunravelk{\fA}$.
Moreover,
for any base-guarded set $\vec{a} = \set{ a_1, \dots, a_n }$ of distinct elements from
$\mydom{f'}$, $f'(\vec{a})$ must yield a set of distinct elements $\set{ f'(a_1), \dots, f'(a_n) }$,
and these elements can only participate in some fact in $\structureunravelk{\fA}$
if the underlying elements from $\vec{a}$ participate in the same fact in $\fA$.
Hence, $f'$ is a safe partial rigid homomorphism.
\item If $f$ is a type~(i) position,
then Spoiler can either choose elements in the active set
and we can reason as we did for the type~(ii) case,
or Spoiler can select elements from the other set of facts.
It suffices to argue that if Spoiler changes the active set like this,
and chooses no new elements,
then the game is still in a safe position.
But in this case $f' = f^{-1}$ is easily seen to still be safe.
\end{itemize}
This concludes the proof of Proposition~\ref{prop:unravel}.
\end{proof}

We can now conclude the proof of Proposition~\ref{prop:transdecomp}.
Assume that $\fA \supseteq \fA_0$ is a set of facts that satisfies $\varphi$
when interpreting $R^\trans$ as the transitive closure of $R$.
Let $3 \leq k \leq \mysize{\varphi}$ be an upper bound on the
maximum number of free variables in any subformula of $\varphi$.
Since $\fA$ satisfies $\varphi$,
Propositions~\ref{prop:unravel}~and~\ref{prop:bisim-game} imply that
$\structureunravelk{\fA}$ also satisfies $\varphi$ when properly interpreting $R^\trans$.
Hence, we can conclude that the unravelling $\structureunravelk{\fA}$
is the transitive closure friendly $k$-tree-like witness
for $\varphi$.

\clearpage

\section{Automata for $\agnf$ (Proof of Theorem~\ref{thm:automata})}

\newcommand{\perm}[1]{\langle #1 \rangle}
\newcommand{\cA}{\calA}
\newcommand{\cB}{\calB}
\newcommand{\powerset}[1]{\mathcal{P}(#1)}
\newcommand{\Dir}{\text{Dir}}
\newcommand{\dleft}{0}
\newcommand{\dright}{1}
\newcommand{\dup}{-1}
\newcommand{\N}{\mathbb{N}}
\newcommand{\sset}{\set}
\newcommand{\paramsk}{U}
\newcommand{\sigk}{\tilde{\sigma}_k}
\newcommand{\bagnames}[1]{\text{names}(#1)}
\newcommand{\mydecode}[1]{\decode(#1)}
\newcommand{\autsig}{\Gamma}

In this section, we prove Theorem~\ref{thm:automata},
about constructing automata for sentences in $\agnf$
and initial sets of facts $\instance_0$:

\begin{quote}
{\bf Theorem~\ref{thm:automata}.} 
Let $\phi$ be a sentence in $\agnf$,
and let $\instance_0$ be a finite set of facts.
We can construct in $\twoexp$
a 2-way alternating parity tree automaton $\calA_{\phi,\instance_0}$
such that\\[.3em]
\null\hfill$ \text{$\instance_0 \wedge \phi$ is satisfiable
\quad iff \quad
$L(\calA_{\phi,\instance_0}) \neq \emptyset$}$\hfill\null\\[.3em]
when $R^\trans \in \sigmad$ are interpreted as the transitive closure
of $R \in \sigmab$.
The number of states of $\calA_{\phi,\instance_0}$ is exponential in $\mysize{\phi} \cdot \mysize{\instance_0}$
and the number of priorities is linear in $\mysize{\phi}$.
\end{quote}

Before we prove the result,
we need to specify the
tree encodings/decodings and tree automata that we are using.
For the remainder of the section,
we fix some $\varphi \in \agnf$ and some finite set of facts $\instance_0$.

\subsection{Tree encodings/decodings}

\myparagraph{Tree encodings}
By Proposition~\ref{prop:transdecomp},
we know that if $\phi \in \agnf$,
then $\phi$ has transitive-closure friendly $k$-tree-like witnesses,
for $k \leq \mysize{\phi}$.
A $\instance_0$-rooted tree decomposition like this
can be encoded as a tree
with only a finite signature.
Let $\paramsk$ be a set of names of size $2k + l$
where $l$ is the size of $\elems{\instance_0}$.
The signature $\sigk$ for the encodings
is defined as follows.
\begin{itemize}
\item For all $a \in \paramsk$, there is a unary relation $D_a \in \sigk$
which indicates that $a$ is a name for an element represented in the bag.
\item For every relation $R \in \sigma$ of arity $n$
and every $n$-tuple $\vec{a} \in \paramsk^n$,
there is a unary relation $R_{\vec{a}} \in \sigk$,
which indicates that $R$ holds for the tuple of elements indexed by~$\vec{a}$.
\item For every $z \in \elems{\instance_0}$ and $c \in \paramsk$,
there is a unary relation $V_{c / z}$
which indicates the valuation for this element.
\end{itemize}

Tree decompositions and the corresponding encodings
can generally have unbounded (possibly infinite) degree.
We modify the standard encoding slightly so that we can use full binary trees:
we apply the first-child, next-sibling transformation
to the usual encoding, based on an arbitrary ordering of the children,
and make it a full binary tree by adding dummy nodes if necessary.

Each node in a binary tree can be identified with a finite string over $\{ 0 , 1 \}$,
with the root identified with $\epsilon$.
The \emph{biological children} of a node $u$ are the nodes
$u 0 1^\trans$ (these are the nodes that would have been children
of $u$ in the tree decomposition
before the first-child next-sibling translation).
The \emph{biological parent} of $v \neq \epsilon$ is the unique $u$
such that $v \in u 0 1^\trans$.
A \emph{biological neighbor} is a biological child or biological parent.
For these binary tree encodings,
we add to $\sigk$ unary predicates $P_i$ for $i \in \{0,1 \}$
which indicate the node is the $i$-th child of its~parent.

From now on, we use the term \emph{$\sigk$-tree}
to refer to an infinite full binary tree over the signature $\sigk$.

\myparagraph{Tree decodings}
If a $\sigk$-tree satisfies certain consistency properties,
then it can be decoded into a  set of $\sigma$-facts
with an $\instance_0$-rooted tree decomposition of width $k - 1$.
Let $\bagnames{v} := \set{ a \in \paramsk : D_a (v)}$ be the
set of \emph{names} used for elements in bag~$v$ in some tree.
We will abuse notation and write $\vec{a} \subseteq \bagnames{v}$ to
mean that $\vec{a}$ is a tuple over names from $\bagnames{v}$.
A \emph{consistent tree} $\tree$ (with respect to $\sigk$ and $\instance_0$)
is a $\sigk$-tree such that
every node $v$ satisfies
\begin{itemize}
  \item $\mysize{\bagnames{v}} \leq k$, except for the root (which has size $l$);
\item for all $R_{\vec{a}} \in \sigk$, if $R_{\vec{a}} (v)$ then $\vec{a} \subseteq \bagnames{v}$;
\item $P_i(v)$ holds iff $v$ is the $i$-th child of its parent;
\item for all $z \in \elems{\instance_0}$,
there is exactly one $c \in \bagnames{\epsilon}$ for the root $\epsilon$ such that
$V_{c/z} (\epsilon)$ holds,
and there is no $v \neq \epsilon$ with some $c \in \bagnames{v}$ such that $V_{c/z}(\epsilon)$ holds;
\item for every $c \in \bagnames{\epsilon}$, there is some $z \in \elems{\instance_0}$ such that $V_{c/z}(\epsilon)$ holds;
\item for each fact $R(z_1\dots z_n) \in \instance_0$,
$R_{c_1 \dots c_n}(\epsilon)$ holds, where each $c_i \in \bagnames{\epsilon}$ is the unique name such that
$V_{c_i / z_i}(\epsilon)$ holds;
\item for every $R_{c_1 \dots c_n}(\epsilon)$,
the fact $R(z_1 \dots z_n)$ is in $\instance_0$, where
each $z_i \in \elems{\instance_0}$ is the unique element such that
$V_{c_i / z_i}(\epsilon)$ holds.
\end{itemize}

The last four conditions ensure that there is a bijection
between the elements and facts represented at the root node
and the elements and facts in $\instance_0$.

Given a consistent tree $\tree$,
we say nodes $u$ and $v$ are \emph{$a$-connected}
if
there is a sequence of nodes $u = w_0, w_1, \dots, w_j = v$
such that $w_{i+1}$ is a biological neighbor of $w_i$,
and $a \in \bagnames{w_{i}}$ for all $i \in \set{0,\dots,j}$.
We write $[v,a]$ for the equivalence class of $a$-connected nodes of $v$.
For $\vec{a} = a_1 \dots a_n$,
we often abuse notation and write $[v,\vec{a}]$ for the tuple 
$[v,a_1],\dots,[v,a_n]$.

The \emph{decoding} of $\tree$ is the
set of $\sigma$-facts $\mydecode{\tree}$
using elements $\set{ [v,a] : \text{$v \in \tree$ and $a \in \bagnames{v}$}}$,
where we identify $z \in \elems{\instance_0}$ with the unique $[\epsilon,c]$ such that $V_{c/z}(\epsilon)$ holds.
For each relation $R$, we have $R([v_1,a_1],\dots,[v_j,a_j]) \in \mydecode{\tree}$ iff
there is some $w \in \tree$ such that
$R_{\vec{a}}(w)$ holds and $[w,a_i] = [v_i,a_i]$ for all~$i$.

\myparagraph{Free variables}
The automaton construction will be an induction on the structure of the formula,
so we will need to deal with formulas with free variables.

For this purpose, the tree encodings
can be extended with additional information
about valuations for free variables.
Such trees use an extended signature.

Namely, for each free first-order variable $z$ and each $c \in \paramsk$,
we introduce a predicate $V_{c / z}$;
if $V_{c / z}(v)$ holds, then this indicates that the valuation for $z$
is the element named by $c$ at~$v$
(we use notation similar to the valuations for $z \in \elems{\instance_0}$,
since these valuations all behave in a similar way).

At one point in what follows (specifically, in one case of the proof of
Lemma~\ref{lemma:aut}),
we will also use \emph{second-order variables}
to represent information about some additional relation $Y$.
For each second-order variable $Y$ of arity $n$ and each $\vec{a} \in \paramsk^n$,
the extended signature has a predicate $Y_{\vec{a}}$.
If $Y_{\vec{a}}$ holds at some node $v$,
then this indicates that the tuple of elements indexed by $\vec{a}$ at $v$
is in the relation $Y$.

We refer to these additional predicates that give a valuation for the free variables
as \emph{free variable markers}.
In a consistent tree,
the free variables markers
for a first-order variable $z$
must satisfy the condition that there is a unique $v$ and unique $c \in \bagnames{v}$ such that
$V_{c / z}(v)$ holds
(i.e.~for each $z$ there is exactly one $V_{c/z}$-fact in the tree).
The markers for a second-order variable $Y$
must satisfy the condition that if
$Y_{\vec{a}}(v)$ then $\vec{a} \subseteq \bagnames{v}$.

\subsection{Automata tools}

We will make use of automata running on infinite binary trees.
We briefly recall some definitions and key properties.
We will need to use 2-way automata
that can move both up and down as they process
the tree, so we highlight some
less familiar properties
about the relationship between 2-way and 1-way versions
of these automata.

\myparagraph{Trees}
The input to the automata will be infinite full binary trees $\tree$
over some finite set of propositions $\autsig$.
In other words, these are structures over a signature
with binary relations for the left and right child relation,
and unary relations for the propositions.
We also assume there are propositions
indicating whether each node is a left child, right child, or the root.
We write $\tree(v)$ for the set of propositions
that hold at node $v$.

\myparagraph{Tree automata}
An \emph{alternating parity tree automaton} $\cA$
is a tuple $\perm{\autsig,Q,q_0,\delta,\Omega}$
where
$\autsig$ is a finite set of propositions,
$Q$ is a finite set of states,
$q_0 \in Q$ is the initial state,
$\delta : Q \times \powerset{\autsig} \to \cB^{+}(\Dir \times Q)$
is the transition function with directions $\Dir \subseteq \set{ \dleft,\dright,\dup}$,
and $\Omega : Q \to P$
is the priority function with a finite set of \emph{priorities} $P \subseteq \N$.

The transition function $\delta$ maps a state and input letter
to a positive boolean formula
over $\Dir \times Q$ (denoted $\cB^{+}(\Dir \times Q)$)
that indicates possible next moves for the automaton.

Running the automaton $\cA$ on some input tree $\tree$
is best thought of in terms of an \emph{acceptance game}.
Positions in the game are of the form $(q,v) \in Q \times \tree$.
In position $(q,v)$, 
Eve chooses a disjunct $\theta$ in $\delta'(q,\tree(v))$,
where $\delta'$ is the result of
writing each of the transition function formulas
in disjunctive normal form.
Then Adam chooses a conjunct $(d,q')$ in $\theta$
and the game continues from position $(q',v')$,
where $v'$ is the node in direction $d$ from $v$
(Adam loses if there is no such node $v'$).

A play $(q_0,v_0) (q_1,v_1) \dots$
in the game is winning for Eve if it
satisfies the \emph{parity condition}: the maximum priority occurring infinitely often
in $\Omega(q_0) \Omega(q_1) \dots$ is even.
A \emph{strategy} for Eve is a function that, given the history of the play
and the current position in the game, determines Eve's choice in the game.
Note that we allow the automaton to be started from arbitrary positions in the tree,
rather than just the root.
We say that $\cA$ \emph{accepts $\tree$ starting from $v_0$} if Eve has a strategy
such that all plays consistent with the strategy starting from $(q_0,v_0)$ are winning.
$L(\cA)$ denotes the \emph{language} of trees accepted by $\cA$ starting from the root.

A 1-way alternating automaton is an automaton that uses only directions $\dleft$ and $\dright$.
A (1-way) nondeterministic automaton is a 1-way alternating automaton
such that every transition function formula is of the form $\bigvee_j \, (\dleft,q^j_\dleft) \wedge (\dright,q^j_\dright)$.

\myparagraph{Closure properties}
We recall some closure properties of these automata,
omitting the standard proofs; see \cite{Thomas97,Loding-unpublished} for more information.
Note that we state only the size of the automata for each property,
but the running time of the procedures constructing these automata
is polynomial in the output size.

First, the automata that we are using
are closed under union and intersection
(of their languages).

\begin{proposition}\label{prop:closure-union-intersection}
2-way alternating parity tree automata
and 1-way nondeterministic parity tree automata
are closed under union and intersection,
with only a polynomial blow-up in the number of states, priorities, and overall size.
\end{proposition}

For example, this means that
if we are given 2-way alternating parity tree automata $\cA_1$ and $\cA_2$,
then we can construct in $\ptime$ a 2-way alternating parity tree automaton $\cA$ such that
$L(\cA) = L(\cA_1) \cap L(\cA_2)$.

Another important language operation is projection.
Let $L'$ be a language of trees over propositions $\Gamma \cup \set{P}$.
The \emph{projection} of $L'$ with respect to $P$
is the language of trees $\tree$ over $\Gamma$ such that
there is some $\tree' \in L'$ such that $\tree$ and $\tree'$
agree on all propositions in $\Gamma$.
Projection is easy for nondeterministic automata
since the valuation for the projected proposition can be guessed by Eve.

\begin{proposition}\label{prop:closure-projection}
1-way nondeterministic parity tree automata
are closed under projection,
with no change in the number of states, priorities, and overall size.
\end{proposition}

Finally, complementation is easy for alternating automata
by taking the \emph{dual} automaton,
obtained by switching conjunctions and disjunctions in the transition function,
and incrementing all of the priorities by one.

\begin{proposition}\label{prop:dual}
2-way alternating parity tree automata are closed under complementation,
with no change in the number of states, priorities, and overall size.
\end{proposition}

\myparagraph{Connections between 2-way and 1-way automata}
It was shown by Vardi~\shortcite{Vardi98} that 2-way alternating parity tree automata can be converted to equivalent
1-way nondeterministic automata,
with an exponential blow-up.

\begin{theorem}[\cite{Vardi98}]\label{thm:2way-to-nd}
Let $\cA$ be a 2-way alternating parity tree automaton.
We can construct a 1-way nondeterministic parity tree automaton $\cA'$
such that $L(\cA) = L(\cA')$.
The number of states of $\cA'$ is exponential in the number of states of $\cA$,
but the number of priorities of $\cA'$ is linear in the number of priorities of $\cA$.
\end{theorem}

1-way nondeterministic tree automata can be seen as a special case of 2-way alternating automata,
so the previous theorem shows that 1-way nondeterministic and 2-way alternating parity automata are
equivalent, in terms of their ability to recognize trees starting from the root.

We need another conversion from 1-way nondeterministic to 2-way alternating automata
that we call \emph{localization}.
This is the process by which a 1-way nondeterministic automaton that is running on trees
with extra information about some predicate annotated on the tree
is converted to an equivalent 2-way alternating automaton
that operates on trees without these annotations
under the assumption that these predicates hold only locally at the
position the 2-way automaton is launched from.
A similar localization theorem is present in prior work \cite{BourhisKR15,lics16-gnfpup}.

\begin{theorem}\label{thm:localization}
Let $\autsig' := \autsig \cup \sset{P_1,\dots,P_j}$.
Let $\cA'$ be a 1-way nondeterministic parity automaton on $\autsig'$-trees.
We can construct a 2-way alternating parity automaton $\cA$ on $\autsig$-trees
such that
for all $\autsig$-trees $\tree$ and nodes $v \in \dom(\tree)$,
\begin{align*}
&\text{$\cA'$ accepts $\tree'$ from the root}
\ \text{ iff } \
\text{$\cA$ accepts $\tree$ from $v$}
\end{align*}
where $\tree'$ is the $\autsig'$-tree obtained from $\tree$ by setting $P_1^{\tree'} =  \dots = P_j^{\tree'} = \set{v}$.
The number of states of $\cA$ is linear in the number of states of $\cA'$,
and the overall size of $\cA$ is linear in the size of $\cA'$.
The number of priorities of $\cA$ is linear in the number of priorities of $\cA'$.
\end{theorem}

\begin{proof}[Proof sketch]
$\cA$ simulates $\cA'$ by guessing in a backwards fashion an initial part of a run of $\cA'$
on the path from $v$ to the root and then processing the rest of the tree in a normal downwards fashion.
The subtlety is that the automaton $\cA$
is reading a tree without valuation for $P_1,\dots,P_j$
so once the automaton leaves node $v$, if it were to cross this position again,
it would be unable to correctly simulate $\cA'$.
To avoid this issue, we only send downwards copies of the automaton in directions
that are not on the path from the root to $v$.
\end{proof}

\myparagraph{Emptiness testing}
Finally, we make use of the well-known fact that
language emptiness of tree automata is decidable.

\begin{theorem}[\cite{EmersonJ88},\cite{Vardi98}]
For 1-way nondeterministic parity tree automata,
emptiness is decidable in time polynomial in the number of states
and exponential in the number of priorities.
For 2-way alternating parity automata,
it is decidable in time exponential in the 
number of states and~priorities.
\end{theorem}

\subsection{Construction}
We are almost ready to construct an automaton
for $\varphi \in \agnf$ and $\instance_0$
to prove Theorem~\ref{thm:automata}.
It is convenient to work with $\nf$ formulas,
so let $\varphi'$ be the $\nf$ $\agnf$ sentence
that is equivalent to $\varphi$.

We build up this automaton inductively,
so we must construct an automaton $\cA_\psi$ for each
subformula $\psi(\vec{x})$ of $\varphi'$.
The automaton $\cA_\psi$ will not specify a single initial state.
Instead,
there will be a designated initial state for each
possible ``local assignment''
for the free variables $\vec{x}$.
A \emph{local assignment} $\vec{a}/\vec{x}$
for $\vec{a} = a_1 \dots a_n \in \paramsk^n$ and $\vec{x} = x_1 \dots x_n$ is a mapping
such that $x_i \mapsto a_i$.
A node $v$ in a consistent tree $\tree$ with $\vec{a} \subseteq \bagnames{v}$
and a local assignment $\vec{a}/\vec{x}$,
specifies a valuation for $\vec{x}$.
We say it is local since the free variable markers for $\vec{x}$
would all appear locally in $v$.

We will write $\cA_\psi$ for the automaton for $\psi$
(without specifying the initial state),
and will write $\cA_\psi^{\vec{a}/\vec{x}}$ for $\cA_\psi$
with the designated initial state for $\vec{a}/\vec{x}$.
We call $\cA_\psi^{\vec{a}/\vec{x}}$ a \emph{localized automaton},
since it is testing whether some tuple that is represented
locally in the tree satisfies $\psi$.
Localized automata are useful
because they can be launched to test
that a tuple of elements that appear together in a node satisfy some
property~--- without having the markers for this tuple explicitly written on the
tree.

The construction is described in the following lemma.

\begin{lemma}\label{lemma:aut}
Let $\psi(\vec{x})$ be a subformula of $\varphi'$ (the $\nf$ version of $\varphi$).
Let $k$ be the width of $\varphi'$,
let $l$ be the size of $\elems{\instance_0}$,
and let $K := 2k + l$.

We can construct a
2-way alternating parity tree automaton $\calA_\psi$
such that for all consistent trees~$\tree$,
for all local assignments $\vec{a}/\vec{x}$,
and for all nodes $v$ in $\tree$ with $\vec{a} \subseteq \bagnames{v}$,
\begin{align*}
&\text{$\calA_\psi^{\vec{a}/\vec{x}}$ accepts $\tree$
starting from $v$} \\
\text{iff} \quad &\text{$\decode(T), [v,\vec{a}]$ satisfies $\psi$}
\end{align*}
when each $R^\trans \in \sigmad$ is interpreted as the transitive closure
of $R \in \sigmab$.

Further, there is a polynomial function $f$ independent of~$\psi$
such that the number of states of $\calA_\psi$
is at most $N_\psi := f(m_\psi) \cdot 2^{f(K r_\psi)}$
where $m_\psi = \mysize{\psi}$ and 
$r_\psi$ is the CQ-rank of~$\psi$.
The overall size of the automaton
and the running time of the construction
is at most exponential in $\mysize{\sigma} \cdot N_\psi$.
The number of priorities is linear in $\psi$.
\end{lemma}

\begin{proof}
We proceed by induction on $\nf$ $\psi(\vec{x})$ in $\agnf$.
We will write $m_\psi$ for $\mysize{\psi}$,
$r_\psi$ for the CQ-rank of $\psi$,
and $N_\psi$ for $f(m_\psi) \cdot 2^{f(Kr_\psi)}$
  for some suitably chosen (in particular, non-constant) polynomial $f$
independent of $\psi$
(we will not define $f$ explicitly).

During each case of the inductive construction,
we will describe informally how to build the desired automaton,
and we will analyze the number of priorities
and the number of states required.
We defer the analysis of the overall size
of the automaton until the end of this proof.

\myparagraph{Base cases}
For each of the base cases $\psi(\vec{x})$,
we first describe a
2-way alternating parity tree automaton $\cB_\psi$ that runs
on trees \emph{with the free variable markers for $\vec{x}$
written on the tree}:
\begin{itemize}
\item Suppose $\psi$ is a $\sigmab$-atom $\alpha(\vec{x})$.
Eve tries to
navigate to a node $v$ whose label includes fact $\alpha(\vec{b})$.
If she is able to do this,
Adam can then challenge Eve to show that $\vec{x}$ corresponds to $\vec{b}$.
Say he challenges her on $b_i \in \vec{b}$.
Then Eve must navigate from $v$ to
the node carrying the marker $b_i/x_i$.
However, she must do this
by passing through a series of biological neighbors
that also contain $b_i$
(the intermediate nodes in between biological neighbors
might not contain $b_i$).
If she is able to do this,
$\cB_\psi$ enters an accepting sink state
(with priority 0).
The other states are non-accepting (with priority~1)
to ensure that Eve actually witnesses $\alpha(\vec{x})$.
The number of states of $\cB_\psi$ is linear in $K$,
since the automaton must remember the name $b_i$ that
Adam is challenging.
There are two priorities.

\item The case when $\psi$ is the $\sigmab$-guardedness predicate
$\guardedb(\vec{x})$ is similar,
except Eve can choose any atom $\alpha$ over $\sigmab$
that uses all of the variables $\vec{x}$,
and then proceed as in the previous case.

\item Suppose $\psi$ is an equality $x_1 = x_2$.
Eve navigates to the node $v$ with the marker $a/x_1$.
She is then required to navigate from $v$ to
the node carrying the marker for $x_2$.
She must do so by passing through a series of biological neighbors
that also contain $a$ (again, the intermediate nodes
in between biological neighbors might not contain $a$).
If she is able to reach the marker $a/x_2$ in this way
then $x_1$ and $x_2$ are marking the same element
in the underlying set of facts,
so $\cB_\psi$ moves to a sink state with priority 0 and she wins.
The other states have priority 1, so if Eve is not able to do this,
then Adam wins.
The state set is of size linear in $K$, in order to remember the name $a$.
There are two priorities.

\item Suppose $\psi$ is a $\sigmad$-atom $R^\trans(x_1,x_2)$.
Eve first tries to navigate to the node $v_0$ carrying the marker $a_1/x_1$ for~$x_1$.
The automaton $\cB_\psi$ then simulates the following game.
The initial position in the game is $(v_0,a_1)$.
In general, positions in the game are of the form $(v,a)$ for a node $v$ and a name $a$,
and one round of the game consists of the following: Eve can either
\begin{itemize}
\item choose $a'$ in $v$
such that label at $v$ includes fact $R(a,a')$;
she immediately wins if $v$ includes marker $a' / x_2$,
otherwise she proceeds to the next round in position $(v,a')$,
or
\item choose some biological neighbor $v'$ which includes the name $a$,
and the game proceeds to the next round in position $(v',a)$.
\end{itemize}
This game can be implemented
using a 2-way automaton.
Winning corresponds to moving to a sink state with priority 0.
All of the other states are assigned priority~1.
This ensures that eventually Eve witnesses a path of $R$-facts from $x_1$ to $x_2$.
The number of states in $\cB_\psi$ is again linear in $K$,
since it must remember the name $a$
that is currently being processed along this path.
There are only two priorities.
\end{itemize}
For each base case $\psi(\vec{x})$,
we have constructed an automaton $\cB_\psi$ with two priorities and
a state set of size linear in $K$.
However, this automaton runs on trees with the free variable markers for $\vec{x}$,
so it remains to show that we can construct the automaton $\cA_\psi$ required by the lemma
that runs on trees without these markers.

First, we can convert $\cB_\psi$ into an equivalent nondeterministic parity tree automaton
with an exponential blow-up in the number of states
and a linear blow-up in the number of priorities
(using Theorem~\ref{thm:2way-to-nd}).
After this step, the number of states is exponential in $K$.

For each local assignment $\vec{a}/\vec{x}$,
we can then apply the localization theorem (Theorem~\ref{thm:localization})
to the set of predicates of the form $V_{a_i/x_i}$,
and eliminate the dependence on any other
$V_{c/x_i}$ for $c \neq a_i$ by always assuming these predicates do not hold.
This results in a localized automaton $\cA_{\psi}^{\vec{a}/\vec{x}}$ that
no longer relies on free variable markers for $\vec{x}$.
By Theorem~\ref{thm:localization},
there is only a linear blow-up in the number of states and number of priorities,
so after this step the number of states in each $\cA_\psi^{\vec{a}/\vec{x}}$ is exponential in $K$.

Finally, we take $\cA_\psi$ to be the disjoint union of~$\cA_{\psi}^{\vec{a}/\vec{x}}$
over all local assignments $\vec{a}/\vec{x}$;
the designated initial state for each localization
is the initial state for $\cA_\psi^{\vec{a}/\vec{x}}$.
Since there are at most $K^k$ localizations,
the number of states in $\cA_\psi$ is still exponential in $K$,
which can be assumed to be less than $N_\psi$ by the choice of $f$.
The number of priorities is a constant independent of $\psi$.

\myparagraph{Inductive cases}
We now proceed with the inductive cases.
We build $\cA_\psi$
with the help of inductively defined automata for its subformulas.

\begin{itemize}
\item Suppose $\psi$ is a guarded negation of the form $\alpha(\vec{x}) \wedge \neg \psi'(\vec{x})$.
Construct $\cA_\psi$ by taking the disjoint union of~$\cA_{\alpha}$,
the dual of $\cA_{\psi'}$
(obtained by switching conjunctions and disjunctions in the transition function formulas in $\cA_{\psi'}$,
and incrementing each priority by one),
and fresh states $q_{\vec{a}/\vec{x}}$ with priority 1
for each local assignment $\vec{a}/\vec{x}$.
For each local assignment $\vec{a}/\vec{x}$,
the designated initial state is $q_{\vec{a}/\vec{x}}$.
From state $q_{\vec{a}/\vec{x}}$, Adam is given a choice
whether to move to the initial state of $\cA_{\alpha}^{\vec{a}/\vec{x}}$
or to the initial state of the dual of $\cA_{\psi'}^{\vec{a}/\vec{x}}$.
The idea is that Adam selects which of the conjuncts
to challenge Eve on.

The state set of $\cA_\psi$
is of size at most
\begin{align*}
&f(m_{\alpha} ) \cdot 2^{f(K r_{\alpha})} + f(m_{\psi'}) \cdot 2^{f(K r_{\psi'})} + K^k \\
\leq \
&2^{f(K r_\psi)} (f(m_{\alpha}) + f(m_{\psi'}) + 1 ) \\
\leq \
&2^{f(K r_\psi)} f(m_{\alpha} + m_{\psi'} + 1) \leq N_\psi .
\end{align*}
The number of priorities
is linear in the size of $\psi$,
since it is at most the sum of the 
number the priorities in the subautomata for $\alpha$ and $\psi'$
(which by the inductive hypothesis were linear in the size of these subformulas).

\item Suppose $\psi$ is a disjunction $\psi_1 \vee \dots \vee \psi_s$.
Construct $\cA_\psi$ by taking the disjoint union of the $\cA_{\psi_i}$
and fresh states $q_{\vec{a}/\vec{x}}$ with priority 1
for each local assignment $\vec{a}/\vec{x}$.
For each local assignment $\vec{a}/\vec{x}$,
the designated initial state is $q_{\vec{a}/\vec{x}}$.
In state $q_{\vec{a}/\vec{x}}$, Eve chooses which $\cA_{\psi_i}^{\vec{a}/\vec{x}}$
to simulate.

The number of states of $\cA_\psi$ is at most
\begin{align*}
&f(m_{\psi_1} ) \cdot 2^{f(K r_{\psi_1})} + \dots + f(m_{\psi_s}) \cdot 2^{f(K r_{\psi_s})} + K^k \\
\leq \
&2^{f(K r_\psi)} (f(m_{\psi_1}) + \dots + f(m_{\psi_s}) + 1 ) \\
\leq \
&2^{f(K r_\psi)} f(m_{\psi_1} + \dots+ m_{\psi_s} + 1) \leq N_\psi .
\end{align*}
The number of priorities
is linear in the size of $\psi$,
since it is at most the sum of the 
number of priorities in the subautomata for $\psi_1$ to $\psi_s$
(which by the inductive hypothesis were linear in the size of these subformulas).

\item Suppose $\psi(\vec{x})$ is a CQ
\[
\exists y_1 \dots y_t ( \alpha_1(\vec{z}_1) \wedge \dots \wedge \alpha_s(\vec{z}_s) )
\]
where each $\vec{z}_i$ is a tuple of variables coming from $\vec{x}$ and $y_1,\dots,y_t$,
and each $\alpha_i$ is an atom over $\sigmab \cup \sigmad$.
This is a specific case, but it is helpful for handling the general CQ-shaped
formulas in the next point.

We start by defining an automaton that runs on trees
with free variable markers for $\vec{x}$ and $y_1 \dots y_t$.
For $1 \leq i \leq s$, let $\cB_{\alpha_i}$ be the automaton for $\alpha_i$
described in the base cases above that runs on trees
with the free variable markers for $\vec{x}$ and $y_1 \dots y_t$.
Let $\calC$ be the automaton
obtained by taking the disjoint union of~$\cB_{\alpha_1}, \dots, \cB_{\alpha_s}$,
and an automaton checking that there is precisely one free variable marker
for $y_1 \dots y_t$,
and adding a new initial state with priority 1 from which
Adam can choose which of these subautomata to simulate.
Thus, $\calC$ is a 2-way alternating automaton
with number of states linear in $Ks \leq Kr_\psi$,
and two priorities;
it checks that the body of the CQ holds
in a tree with all of the free variable markers present.

We can then convert $\calC$ to an equivalent nondeterministic parity tree automaton $\calC'$
using Theorem~\ref{thm:2way-to-nd},
with an exponential blow-up in the number of states,
and a linear blow-up in the number of priorities.
After this step, the number of states is exponential in $Kr_\psi$.

Next, we take the projection of
$\calC'$ on the free variable markers for $y_1 \dots y_t$
to obtain $\cB_\psi$:
that is, $\cB_\psi$
simulates $\calC'$ while guessing
the markers for the variables $y_1 \dots y_t$.
This is an automaton for $\psi$,
but it runs on trees with markers for the free variables $\vec{x}$.

For each local assignment $\vec{a}/\vec{x}$,
we can then apply the localization theorem (Theorem~\ref{thm:localization})
to the set of predicates of the form $V_{a_i/x_i}$,
and eliminate the dependence on any other
$V_{c/x_i}$ for $c \neq a_i$ by always assuming these predicates do not hold.
This results in a localized automaton $\cA_{\psi}^{\vec{a}/\vec{x}}$ that
no longer relies on free variable markers for~$\vec{x}$.
By Theorem~\ref{thm:localization},
there is only a linear blow-up in the number of states and number of priorities,
so after this step the number of states is exponential in $Kr_\psi$.

Finally, we take $\cA_\psi$ to be the disjoint union of the $\cA_{\psi}^{\vec{a}/\vec{x}}$
over all local assignments $\vec{a}/\vec{x}$;
the designated initial state for each localization
is the initial state for $\cA_\psi^{\vec{a}/\vec{x}}$.
Since there are at most $K^k$ localizations,
the number of states in $\cA_\psi$ is still exponential in $K r_\psi$,
which can be assumed to be less than $N_\psi$ by the choice of $f$.
The number of priorities is a constant independent of $\psi$.

\item Suppose $\psi$ is a CQ-shaped formula of the form
\[\delta[Y_1 \mapsfrom \phi_1, \ldots, Y_n \mapsfrom \phi_n]
\]
where $\delta$ is a CQ over $\sigma \cup \{Y_1,\dots,Y_n\}$ and $\phi_i \in \agnf$.
The inductive hypothesis yields $\cA_{\phi_i}$ for each of the $\phi_i$.
Let $\calN$ be the automaton for the CQ $\delta$ obtained using a similar approach as the previous case.
Note that this automaton runs on trees with a valuation for the free second-order variables $Y_i$
marked on the tree.
These free variables represent base-guarded relations
(i.e.~relations in which each tuple in the relation is base-guarded),
since it is guaranteed that for each $Y_i$ atom, there is a $\sigmab$-atom or $\sigmab$-guardedness predicate in $\delta$ that contains its free variables.

To construct $\cA_\psi$,
take the disjoint union of~$\calN, \allowbreak \cA_{\phi_1}, \dots \cA_{\phi_n}$.
For each localization $\vec{a}/\vec{x}$, the designated initial state is
the initial state for $\vec{a}/\vec{x}$ coming from $\calN$.
The idea is that $\cA_\psi$ 
starts by simulating $\calN$, but with Eve guessing valuations for $Y_i$.
This is where it is important that the $Y_i$ are $\sigmab$-guarded relations:
since any $Y_i$-fact must be about a $\sigmab$-guarded set of elements,
these elements must appear together in some node of the tree,
so Eve can guess an annotation of the tree that indicates that $Y_i$ holds of these elements.
Adam can either accept Eve's guesses of the valuation and continue the simulation of $\calN$,
or can challenge one of Eve's assertions of~$Y_i$ by launching the appropriate localized version of~$\phi_i$.
That is, if Eve guesses that $Y_i(\vec{z}_i)$ holds of~$\vec{b}$ at~$v$,
then Adam could challenge this by launching $\cA_{\varphi_i}^{\vec{b}/\vec{z}_i}$
starting from~$v$.
This is where it is crucial that we have localized automata
for these subformulas and for all possible local assignments
that can be launched from internal nodes when Adam challenges
one of Eve's guesses: in particular,
note that the same $\cA_{\varphi_i}^{\vec{b}/\vec{z}_i}$
can be launched for different initial localizations
$\vec{a}/\vec{x}$.

By the inductive hypothesis, each $\cA_{\phi_i}$ automaton
has at most $f(m_{\phi_i}) \cdot 2^{f(K r_{\phi_i})}$ states,
and number of priorities linear in $m_{\phi_i}$.
Likewise, the automaton $\calN$ for $\delta$ has two priorities
and number of states exponential in $Kr_\delta$,
which we can assume to be at most
at most $2^{f(Kr_\delta)}$.

Hence, the number of priorities in $\cA_\psi$ is linear in $m_\psi$,
and the number of states in $\cA_\psi$ is at most
\begin{align*}
&2^{f(K r_\delta)} + f(m_{\phi_1}) \cdot 2^{f(K r_{\phi_1})} + \dots \\
&\phantom{2^{f(K r_\delta)} } + f(m_{\phi_n}) \cdot 2^{f(K r_{\phi_n})} \\
\leq \
&2^{f(K r_\psi)} (1 + f(m_{\phi_1}) + \dots + f(m_{\phi_n}) ) \leq N_\psi .
\end{align*}
\end{itemize}
This concludes the inductive cases.

\myparagraph{Overall size}
We have argued that each automaton has at most $N_\psi$ states
and the number of priorities at most linear in~$\psi$.
It remains to argue that the overall size of $\cA_\psi$ is at most
exponential in $\mysize{\sigma} \cdot N_\psi$.
The size of the priority mapping is at most polynomial in $N_\psi$.
The size of the alphabet is exponential in $\mysize{\sigma} \cdot K^k$,
which is at most exponential in $\mysize{\sigma} \cdot N_\psi$.
For each state and alphabet symbol,
the size of the corresponding transition function formula
can always be kept of size at most exponential in $N_\psi$.
Hence, the overall size of the transition function is at most exponential in
$\mysize{\sigma} \cdot N_\psi$.
Thus, the overall size of $\cA_\psi$ is
at most exponential in $\mysize{\sigma} \cdot N_\psi$.

It can be checked that
the running time of the construction is polynomial
in the size of the constructed automaton,
and hence is also exponential in $\mysize{\sigma} \cdot N_\psi$.
\end{proof}

We must also construct an automaton that checks that
the input tree is consistent,
and actually represents a set of facts $\instance$
such that $\instance \supseteq \instance_0$
and where every $R^\trans$-fact in $\instance_0$ is actually witnessed
by some path of $R$-facts in $\instance$.
For notational simplicity in the statement of the lemma,
we assume that the element names in $\instance_0$ are
used as the names in $\paramsk$ for the root of the consistent trees
(but this is only a technicality).

\begin{lemma}\label{lemma:consistency-distinguished-relations}
We can construct
a 2-way alternating parity tree automaton $\calA_{\instance_0}$
in time doubly exponential in $\mysize{\sigma} \cdot K$,
such that for all trees~$\tree$,
\begin{align*}
&\text{$\calA_{\instance_0}$ accepts $\tree$} \\
\text{iff} \quad 
&\text{$T$ is consistent and for all facts $S(\vec{c}) \in \instance_0$,} \\
&\text{$\decode(\tree), [\epsilon,\vec{c}]$ satisfies $S(\vec{x})$}.
\end{align*}
when $R^\trans \in \sigmad$ is interpreted as the transitive closure
of $R \in \sigmab$.
The number of states is at most exponential in $\mysize{\sigma} \cdot K$,
the number of priorities is two,
and
the overall size is at most doubly exponential in $\mysize{\sigma} \cdot K$.
\end{lemma}

\begin{proof}
The automaton is designed to allow
Adam to challenge some consistency condition
or a particular fact $S(\vec{c})$
in~$\instance_0$.

It is straightforward to design suitable automata checking each
consistency condition,
so suppose Adam challenges some fact in $\instance_0$.
Then the automaton simply launches $\cA_{S(\vec{x})}^{\vec{c}/\vec{x}}$
(obtained from Lemma~\ref{lemma:aut}) from the root.
Note that in case $S(\vec{c})$ is some $R^\trans(c_1,c_2)$,
this $R$-path  witnessing this fact may require elements outside of $\elems{\instance_0}$
even though $c_1$ and $c_2$ are names of elements in $\instance_0$.

The number of states is exponential in $\mysize{\sigma} \cdot K$,
and the overall size is at most doubly exponential in $\mysize{\sigma} \cdot K$.
Only two priorities are needed.
\end{proof}

We can now conclude the proof of Theorem~\ref{thm:automata}.
Recall that $\varphi$ is in $\agnf$ and $\instance_0$ is some finite set of facts.
Without loss of generality, we can assume that
$\mysize{\varphi} \cdot \mysize{\instance_0}
\geq \mysize{\sigma}$.
We construct the $\nf$ $\varphi'$ equivalent to $\varphi$ in exponential time
using Proposition~\ref{prop:nf}.
Although the size of $\varphi'$ can be exponentially larger than $\varphi$,
the width and CQ-rank is at most~$\mysize{\varphi}$,
so we can apply Lemma~\ref{lemma:aut} to
construct an automaton for~$\varphi'$ (and hence $\varphi$)
in time doubly exponential in $\mysize{\varphi} \cdot \mysize{\instance_0}$.
However, the number of states and priorities in this automaton is at most
singly exponential in $\mysize{\varphi} \cdot \mysize{\instance_0}$,
and the number of priorities is linear in $\mysize{\varphi}$.
By taking the intersection of
this automaton from Lemma~\ref{lemma:aut} with 
the automaton for $\instance_0$ and consistency from Lemma~\ref{lemma:consistency-distinguished-relations},
we have a 2-way alternating parity tree automaton $\cA_{\varphi,\instance_0}$
of the desired size
that has a non-empty language
iff
$\varphi$ is satisfiable.
This concludes the proof of Theorem~\ref{thm:automata}.

\clearpage

\section{Base-guarded-interface tree decompositions for $\agnf$}

We prove the following result:

\begin{proposition}\label{prop:guarded-interface-dec-appendix}
  Every sentence $\phi$ in $\agnf$ has base-guarded-interface $k$-tree-like witnesses for
   some $k \leq \mysize{\phi}$.
\end{proposition}

That is,
for every sentence $\phi$ in $\agnf$
and for every finite
set of facts $\instance_0$,
if there is some
$\instance \supseteq \instance_0$ satisfying $\phi$
then there is
such an $\instance$
that has a $\instance_0$-rooted $(k-1)$-width
base-guarded-interface tree decomposition.

The result and proof of Proposition~\ref{prop:guarded-interface-dec-appendix}
is very similar to Proposition~\ref{prop:transdecomp}.
However, unlike Proposition~\ref{prop:transdecomp},
we do not interpret the distinguished relations in a special way here.
This allows us to prove the stronger base-guarded-interface property
about the corresponding tree decompositions,
which will be important for later arguments
(e.g., Proposition~\ref{prop:rewritelin} and Theorem~\ref{thm:ptimetransdataupper}).

We first consider a variant of the $\gn^k$ bisimulation game defined earlier in Appendix~\ref{app:tctreelike}.
The positions in the game are the same as before:
\begin{itemize}
 \item[i)] partial isomorphisms $f:\fA\restrict{}{X} \to \fB\restrict{}{Y}$ or $g:\fB\restrict{}{Y} \to \fA\restrict{}{X}$, 
     where $X \subset \elems{\fA}$ and $Y \subset \elems{\fB}$ are both finite and are $\sigmab$-guarded;
 \item[ii)] partial rigid homomorphisms $f:\fA\restrict{}{X} \to \fB\restrict{}{Y}$ or $g:\fB\restrict{}{Y} \to \fA\restrict{}{X}$, 
     where $X \subset \elems{\fA}$ and $Y \subset \elems{\fB}$ are both finite and are of size at most $k$.  
\end{itemize}

However, the rules of the game are different.

From a type (i) position $h$,
Spoiler must choose a finite subset $X \subset \elems{\fA}$ or 
a finite subset $Y \subset \elems{\fB}$, in either case of size at most $k$, 
upon which Duplicator 
must respond by a partial rigid homomorphism with domain $X$ or $Y$ accordingly,
mapping it into the 
other set of facts in a manner consistent with~$h$.
(This is the same as before).

In a type (ii) position $h$,
Spoiler is only allowed to select some base-guarded subset $X'$ of $\mydom{h}$,
and then the game proceeds from the type (i) position
obtained by restricting $h$ to this base-guarded subset.

Thus, the game strictly alternates between type (ii) positions and base-guarded positions of type (i).
We call this a \emph{base-guarded-interface $\gn^k$ bisimulation game},
since the interfaces (i.e.~shared elements) between the domains of consecutive positions
must be base-guarded.
We can then show:

\begin{proposition}\label{prop:bisim-game-guarded}
Let $\varphi \in \agnfk$ in $\nf$.

If Duplicator has a winning strategy in the
base-guarded-interface $\gn^k$ bisimulation game
between $\fA$ and $\fB$
starting from a type (i) position $\vec{a} \mapsto \vec{b}$
and $\fA$ satisfies $\varphi(\vec{a})$,
then $\fB$ satisfies $\varphi(\vec{b})$.
\end{proposition}

\begin{proof}
Suppose Duplicator has a winning strategy in the base-guarded-interface
$\gn^k$ bisimulation game
between $\fA$ and $\fB$.

If $\varphi$ is a $\sigma$-atom $A(\vec{x})$,
the result follows from the fact that the position
$\vec{a} \mapsto \vec{b}$ is a partial homomorphism.

If $\varphi$ is a disjunction,
the result follows easily from the inductive hypothesis.

Suppose $\varphi$ is a base-guarded negation
$A(\vec{x}) \wedge \neg \varphi'(\vec{x}')$.
By definition of $\agnf$, it must be the case that $A \in \sigmab$
and $\vec{x}'$ is a sub-tuple of $\vec{x}$.
Since $\fA,\vec{a}$ satisfies $\varphi$,
we know that $\fA,\vec{a}$ satisfies $A(\vec{x})$,
which implies (by induction) that
$\fB,\vec{b}$ also satisfies $A(\vec{x})$.
It remains to show that $\fB$ satisfies $\neg \varphi'(\vec{x}')$.
Assume for the sake of contradiction that it satisfies $\varphi'(\vec{x}')$.
Because $\vec{a} \mapsto \vec{b}$ is a type (i) position, we can
consider the move in the game where Spoiler switches the domain to the other set of facts,
keeps the same set of elements,
and then collapses to the base-guarded elements in the subtuple $\vec{b}'$ of $\vec{b}$
corresponding to $\vec{x}'$ in $\vec{x}$.
Let $\vec{a}'$ be the corresponding subtuple of $\vec{a}$.
Duplicator must still have a winning strategy
from this new type (i) position $\vec{b}' \mapsto \vec{a}'$,
so the inductive hypothesis ensures that
$\fA, \vec{a}'$ satisfies $\varphi'(\vec{x}')$,
a contradiction.

Finally, suppose $\varphi$
is a CQ-shaped formula
\[\delta[Y_1 \mapsfrom \phi_1, \ldots, Y_n \mapsfrom \phi_n]
\]
where $\delta$ is a CQ $\exists \vec{y} ( \alpha_1 \wedge \dots \wedge \alpha_j )$
over $\sigma \cup \{Y_1,\dots,Y_n\}$ and $\phi_i$ is in $\nf$ $\agnfk$.
We are assuming that $\fA, \vec{a}$ satisfies $\varphi$.
Hence,
there is some $\vec{c} \in \elems{\fA}$ such that
$\fA,\vec{a},\vec{c}$ satisfies $(\alpha_1 \wedge \dots \wedge \alpha_j)[Y_1 \mapsfrom \phi_1, \ldots, Y_n \mapsfrom \phi_n]$.
Because the width of $\varphi$ is at most $k$,
we know that the combined number of elements in $\vec{a}$ and $\vec{c}$ is at most $k$.
Hence, we can consider the move in the game where Spoiler
selects the elements in $\vec{a}$ and $\vec{c}$.
Duplicator must respond with some
$\vec{d} \in \elems{\fB}$
such that $\vec{a}\vec{c} \mapsto \vec{b}\vec{d}$
is a partial rigid homomorphism, a type (ii) position.
Now consider the possible conjuncts in this CQ-shaped formula.
Conjuncts that are $\sigma$-atoms must be satisfied in $\fB, \vec{b}\vec{d}$
since $\vec{a}\vec{c} \mapsto \vec{b}\vec{d}$
is a partial homomorphism with respect to $\sigma$.
For the conjuncts $\varphi_i$ corresponding to $Y_i$,
we can consider Spoiler's restriction of $\vec{a}\vec{c}$ to the elements used by this conjunct,
and the corresponding restriction of $\vec{b}\vec{d}$.
This is a valid move to a type (i) position,
since the definition of $\agnf$ requires that these non-atomic conjuncts are base-guarded.
Moreover, this new position witnesses the satisfaction of that conjunct in $\fA$.
Since Duplicator must still have a winning strategy from this new type (i) position,
the inductive hypothesis implies that this conjunct is also satisfied in $\fB$.
Since this is true for all conjuncts in the CQ-shaped formula,
$\fB,\vec{b}$ satisfies $\varphi$ as desired.
\end{proof}

We then use a variant of the unravelling based on this game.
The \emph{base-guarded-interface $\gn^k$-unravelling} $\structureunravelkint{\fA}$
is defined in a similar fashion to the $\gn^k$-unravelling,
except it uses only sequences
$\Pi \cap \set{ X_0 \dots X_n : \text{for all $i \geq 1$, $X_i \cap X_{i+1}$ is $\sigmab$-guarded}}$.
This unravelling has an $\instance_0$-rooted base-guarded-interface tree decomposition of width $k-1$.
Moreover:

\begin{proposition}\label{prop:unravelguarded}
Duplicator has a winning strategy in the base-guarded-interface $\gn^k$ bisimulation game
between $\fA$ and $\structureunravelkint{\fA}$.
\end{proposition}

\begin{proof}
The proof is similar to Proposition~\ref{prop:unravel}.
The delicate part of the argument is when
Spoiler selects some new elements $X'$ in $\fA$
starting from a safe position $f$
(for which there is some $\pi$ such that $f(a) = [\pi,a]$ for all $a \in \mydom{f}$).
We need to show that
$[\pi',a']$ for $a' \in X'$ and $\pi' = \pi \cdot X'$ is well-defined in $\structureunravelkint{\fA}$.
This is well-defined only 
if the overlap
between the elements in $\pi$ and $\pi'$ is 
base-guarded.
But because the base-guarded-interface $\gnk$ bisimulation game
strictly alternates between type (i) and (ii) positions,
Spoiler can only select new elements $X'$ in a type (i) position,
so the overlap satisfies this requirement.
The remainder of the proof is the same as in Proposition~\ref{prop:unravel}.
\end{proof}

We can conclude the proof of Proposition~\ref{prop:guarded-interface-dec-appendix} as follows.
Assume that $\fA$ is a set of facts that satisfies $\varphi$.
By Proposition~\ref{prop:nf},
we can convert to an equivalent $\varphi' \in \agnfk$ in $\nf$
with width $k \leq \mysize{\varphi}$.
Since $\fA$ satisfies $\varphi'$,
Propositions~\ref{prop:unravelguarded}~and~\ref{prop:bisim-game-guarded} imply that
$\structureunravelkint{\fA}$ also satisfies $\varphi' \in \agnfk$.
Hence, we can conclude that the unravelling $\structureunravelkint{\fA}$
is a base-guarded-interface $k$-tree-like witness
for $\varphi$.

\clearpage

\section{Reduction of $\owqalin$ to $\owqa$ \\ (Proof of Lemmas~\ref{lemma:guarded-interface-dec}~and~\ref{lemma:reducelin}
for Proposition~\ref{prop:rewritelin})}
\label{app:owqalinred}
\newcommand{\instanceg}{\instance'}

Recall the statement of Proposition~\ref{prop:rewritelin},
which describes the reduction from $\owqalin$ to $\owqa$:

\begin{quote}
{\bf Proposition~\ref{prop:rewritelin}.} 
For any finite set of facts $\instance_0$,
  constraints $\Sigma \in \acgnf$,
  and
  base-covered UCQ $Q$,
  we
  can compute $\instance_0'$ and $\Sigma' \in \agnf$ in $\ptime$
  such that
  $\owqalin(\instance_0, \Sigma, Q) \text{ iff } 
  \owqa(\instance_0', \Sigma', Q)$.
\end{quote}

Specifically,
$\instance_0'$ is $\instance_0$ together with facts $G(a,b)$ for every pair $a,b \in \elems{\instance_0}$,
where $G$ is some fresh binary base relation.
$\Sigma'$ consists of $\Sigma$ together with the $k$-guardedly linear axioms
  for each distinguished relation, where $k$ is  $\max(\mysize{\Sigma \wedge
  \neg Q}, \arity{\sigma \cup \{G\}})$.

Recall that the \emph{$k$-guardedly linear axioms}
require that each binary relation $\drel$ is:
\glinaxioms{itemize}
The idea is that these axioms are strong enough to enforce conditions about
transitivity and irreflexivity within ``small'' sets of elements~--- intuitively,
within sets of at most $k$ elements
that appear together in some bag of a $(k-1)$-width tree decomposition.

The proof of the correctness of the reduction is described in the body of the paper,
but relies on Lemmas~\ref{lemma:guarded-interface-dec}~and~\ref{lemma:reducelin},
which we prove now.

\subsection{Proof of Lemma~\ref{lemma:guarded-interface-dec}}

Recall the statement:
\begin{quote}
{\bf Lemma~\ref{lemma:guarded-interface-dec}.} 
The sentence $\Sigma' \wedge \neg Q$ has base-guarded-interface $k$-tree-like
  witnesses for $k = \max(\mysize{\Sigma \wedge \neg Q}, \arity{\sigma \cup
  \{G\}})$.
\end{quote}

By Proposition~\ref{prop:nf} and Proposition~\ref{prop:guarded-interface-dec-appendix},
$\Sigma \wedge \neg Q$ has a base-guarded-interface $k$-tree-like witness
for $k = \mysize{\Sigma \wedge \neg Q}$.

To prove this lemma, then, it suffices to argue that
the \mbox{$k$-guardedly linear} axioms can also be written
in $\nf$ $\agnf$ with width at most
$k$.

The guardedly total axiom is written in $\nf$ $\agnf$
as
\[
\neg \exists x y ( \guardedbg(x,y) \wedge \neg ( x = y \vee  x \drel y \vee y \drel x ) )
\]
with width at most $k$. The irreflexive axiom is already written in normal form
$\agnf$ with width at most $k$.
For the $k$-guardedly transitive axioms,
note that $\psi_l(x,y)$ has width $l+1$
and $\psi_l(x,x)$ has width $l$, so that 
each of the $k$-guardedly transitive axioms has width at most $k$: this uses the
fact that the width of the $\guardedbg$-atoms have arity at most
$\arity{\sigma \cup \{G\}}$, and we know that $k \geq \arity{\sigma \cup \{G\}}$

Therefore, unlike the property of being a linear order,
the $k$-guardedly linear restriction can be expressed in $\agnf$,
and can even be written in $\nf$ $\agnf$ of width at most $k$.
Overall, this means that if $\Sigma \wedge \neg Q$ has width at most $k$
when converted into $\nf$
then $\Sigma' \wedge \neg Q$
also has width at most $k$.
Hence,
the sentence $\Sigma' \wedge \neg Q$ has base-guarded-interface $k$-tree-like witnesses for
   $k = \mysize{\Sigma \wedge \neg Q}$,
   by Proposition~\ref{prop:guarded-interface-dec-appendix}.
 
\subsection{Proof of Lemma~\ref{lemma:reducelin}}

Recall the statement:
\begin{quote}
{\bf Lemma~\ref{lemma:reducelin}.} 
If there is $\instance' \supseteq \instance_0'$
that satisfies $\Sigma' \wedge \neg Q$
and has a $\instance_0'$-rooted base-guarded-interface $(k-1)$-width tree decomposition,
then there is $\instance'' \supseteq \instance'$
that satisfies $\Sigma' \wedge \neg Q$
where each distinguished relation is a strict linear order.
\end{quote}

We start with some auxiliary lemmas
about base-guarded-interface tree decompositions.

\myparagraph{Transitivity lemma}
We first prove a result about transitivity
for sets of facts with
base-guarded-interface tree decompositions.

\begin{lemma}\label{lem:guarded-transitivity}
Suppose $\instanceg$ is a set of facts with a
$\instanceg_0$-rooted $(k-1)$-width base-guarded-interface tree decomposition $(T, \child ,\lambda)$.
If $\instanceg$ is $k$-guardedly transitive with respect to binary relation $\drel$,
and there is a $\drel$-path $a_1 \dots a_n$
where the pair $\{a_1, a_n\}$ is base-guarded,
then $a_1 \drel a_n \in \instanceg$.
\end{lemma}

\begin{proof}
Suppose there is an $\drel$-path $a_1 \dots a_n$
and that the pair $\{a_1, a_n\}$ is base-guarded, with $v$ a node
where $a_1, a_n$ appear together.
We can assume that $a_1 \dots a_n$ is a minimal $\drel$-path
between $a_1$ and $a_n$,
so there are no repeated intermediate elements.
Consider a minimal subtree $T'$ of $T$ containing
$v$ and containing all of the elements $a_1 \dots a_n$.
We proceed by induction on the length of the path
and on the number of nodes of $T'$ (with the lexicographic order on this pair)
to show that $a_1 \drel a_n$ is in~$\instanceg$.

If all elements $a_1 \dots a_n$ are represented at $v$,
then either (i) all elements are in the root 
or (ii) the elements are in some internal node.
For (i), by construction of $\instanceg_0$,
every pair of elements in $a_1 \dots a_n$ is guarded (by $G$).
Hence, repeated application of the axiom
\[
\forall x y z ( (x \drel z \wedge z \drel y \wedge \guardedbg(x,y)) \rightarrow x \drel y )
\]
(which is part of the $k$-guardedly transitive axioms)
is enough to ensure that
$a_1 \drel a_n$ holds.
For (ii), since the bag size of an internal node is at most $k$,
we must have $n \leq k$, in which case an application of the
  $k$-guardedly transitive axiom to the guarded pair $\{a_1, a_n\}$ ensures that $a_1 \drel a_n$ holds.
This covers the base case of the induction.

Otherwise,
there must be some $1 \leq i < j \leq n$ such that
$a_i$ and $a_j$ are represented at $v$,
but $a_{i'}$ is not represented at $v$ for $i < i' < j$ (in particular $a_{i+1}$
is not represented at~$v$).
We claim that $a_i$ and $a_j$ must be in an interface together.

We say $a_{i+1}$ is \emph{represented in the direction of $v'$}
if $v'$ is a child of $v$ and $a_{i+1}$ is represented in the subtree
rooted at~$v'$,
or $v'$ is the parent of $v$ and $a_{i+1}$ is represented in the tree obtained
from~$T'$
by removing the subtree rooted at $v$.
Note that by definition of a tree decomposition,
since $a_{i+1}$ is not represented at $v$,
it can only be represented in at most one direction.

Let $v_{i+1}$ be the neighbor (child or parent) of $v$ such that
$a_{i+1}$ is represented in the direction of $v_{i+1}$.
It is straightforward to show that
$a_i$ and $a_{j}$ must both be represented in the subtree in the direction of~$v_{i+1}$
in order to witness the facts $a_i \drel a_{i+1}$ and $a_{j-1} \drel a_{j}$.
But $a_i$ and $a_j$ are both in $v$,
so they must both be in $v_{i+1}$.
Hence, $a_i$ and $a_j$ are in the interface between $v$ and $v_{i+1}$.

If this is an interface with the root node, then the pair $a_i, a_j$
is base-guarded (by definition of $\instanceg_0$).
Otherwise, the definition of base-guarded-interface tree decompositions
ensures that they are base-guarded.

Hence, we can apply the inductive hypothesis to the path $a_i \dots a_j$
and the subtree $T''$ of $T'$ in the direction of $v_{i+1}$
to conclude that $a_i \drel a_j$ holds
(we can apply the inductive hypothesis
because $T''$ is smaller than $T'$ as we removed $v$, and $a_i \ldots a_j$ is no
longer than $a_1 \ldots a_n$).
If $i = 1$ and $j = n$, then we are done.
If not,
then we can apply the inductive hypothesis to the new, strictly shorter path
$a_1 \dots a_i a_j \dots a_n$ in $T'$
and conclude that $a_1 \drel a_n$ is in $\instanceg$ as desired.
\end{proof}

\myparagraph{Cycles lemma}
We next show that within base-guarded-interface tree decompositions,
$k$-guarded transitivity and irreflexivity imply
cycle-freeness.

\begin{lemma} \label{lem:nobadcycle}
Suppose $\instanceg$ is a set of facts with a
$\instanceg_0$-rooted $(k-1)$-width base-guarded-interface tree decomposition $(T, \child ,\lambda)$.
If $\instanceg$ is $k$-guardedly transitive and irreflexive with respect to $\drel$,
then $\drel$ in $\instanceg$ cannot have a cycle.
\end{lemma}

\begin{proof}
Suppose for the sake of contradiction that there is a cycle
$a_1 \dots a_n a_1$ in $\instanceg$ using relation $\drel$.
Take a minimal length cycle.

If elements $a_1 \dots a_n$ are all represented in a single node in $T$,
then either (i) all elements are in the root 
or (ii) the elements are in some internal node.
For (i), by construction of $\instanceg_0$,
every pair of elements in $a_1 \dots a_n$ is guarded (by $G$).
Hence, repeated application of the axiom
\[
\forall x y z ( (x \drel z \wedge z \drel y \wedge \guardedbg(x,y)) \rightarrow x \drel y )
\]
(which is part of the $k$-guardedly transitive axioms)
would force $a_1 < a_1$ to be in $\instanceg$,
which would contradict irreflexivity.
  Likewise, for (ii), since the bag size of an internal node is at most $k$,
we must have $n \leq k$, so we can apply the $k$-guardedly transitive axioms
to deduce $a_1 \drel a_1$, which contradicts irreflexivity.

Even if this is not the case, then since $a_n \drel a_1$ holds,
there must be some node $v$ in which both $a_1$ and $a_n$ are represented.
Since not all elements are represented at $v$, however,
there is $1 \leq i < j \leq n$ such that
$a_i$ and $a_j$ are represented at $v$,
but $a_{i'}$ is not represented at $v$ for $i < i' < j$.
We claim that $a_i$ and $a_j$ must be in an interface together.
Observe that $a_{i+1}$ is not represented at $v$.
Let $v_{i+1}$ be the neighbor of $v$ such that
$a_{i+1}$ is represented in the subtree in the direction of $v_{i+1}$.
It is straightforward to show that
$a_i$ and $a_{j}$ must both be represented in the subtree of $T'$ in the direction of $v_{i+1}$
in order to witness the facts $a_i \drel a_{i+1}$ and $a_{j-1} \drel a_{j}$.
But $a_i$ and $a_j$ are both in $v$,
so they must both be in $v_{i+1}$.
Hence, $a_i$ and $a_j$ are in the interface between $v$ and $v_{i+1}$.
If this is an interface with the root node, then the pair $a_i, a_j$
is base-guarded (by definition of $\instanceg_0$);
otherwise, the definition of base-guarded-interface tree decomposition
ensures that they are base-guarded.
By Lemma~\ref{lem:guarded-transitivity} this means that $a_i \drel a_j$ holds.
Hence, there is a strictly shorter cycle $a_1 \dots a_i a_j \dots a_n a_1$,
contradicting the minimality of the original cycle.
\end{proof}

\myparagraph{Base-coveredness lemma}
Lastly, we note that adding only facts about unguarded sets of elements cannot impact $\acgnf$ constraints.
This is where we are utilizing the base-coveredness assumption.

\begin{lemma}\label{lemma:acov}
Let $\instanceg' \supseteq \instanceg$
with additional facts about distinguished relations, 
but no new facts about base-guarded tuples of elements.
Let $\varphi(\vec{x}) \in \acgnf$.
If $\instanceg, \vec{a}$ satisfies $\varphi(\vec{x})$
then $\instanceg', \vec{a}$ satisfies $\varphi(\vec{x})$.
\end{lemma}

\begin{proof}
We assume without loss of generality that
$\varphi$ is in $\nf$ $\acgnf$.

Let $\acgnfplus$ (respectively, $\acgnfminus$)
denote the $\nf$ $\agnf$ formulas
where the covering requirements
(distinguished atoms in CQ-shaped subformulas are appropriately base-guarded)
are required for
positively occurring (respectively, negatively occurring)
CQ-shaped formulas.
Observe that $\acgnf = \acgnfminus$.

We prove a slightly stronger result:
\begin{quote}
For $\varphi(\vec{x}) \in \acgnfminus$: \\
$\instanceg, \vec{a}$ satisfies $\varphi(\vec{x})$ implies $\instanceg', \vec{a}$ satisfies $\varphi(\vec{x})$. \\
For $\varphi(\vec{x}) \in \acgnfplus$: \\
$\instanceg', \vec{a}$ satisfies  $\varphi(\vec{x})$ implies $\instanceg, \vec{a}$ satisfies  $\varphi(\vec{x})$.
\end{quote}

We proceed by induction on the structure of~$\varphi$.
The base case for a $\sigmab$ atom is immediate.
The inductive case for disjunction is also immediate.

Suppose $\varphi := A(\vec{x}) \wedge \neg \varphi'(\vec{x})$,
and $\varphi \in \acgnfminus$.
If $\instanceg,\vec{a}$ satisfies $\varphi(\vec{x})$,
then $\instanceg',\vec{a}$ satisfies $A(\vec{x})$ by the inductive hypothesis.
We must also have $\instanceg',\vec{a}$ satisfies $\neg \varphi'(\vec{x})$,
for if not,
then $\instanceg',\vec{a}$ satisfies $\varphi'(\vec{x})$ (for $\varphi' \in \acgnfplus$),
so the inductive hypothesis implies that $\instanceg,\vec{a}$ satisfies $\varphi'(\vec{x})$, a contradiction.
Hence, $\instanceg', \vec{a}$ satisfies  $\varphi(\vec{x})$ as desired.
The proof is similar starting from $\varphi \in \acgnfplus$.

That leaves only the general CQ-shaped formula case.
Suppose $\varphi := \exists \vec{y} \, ( \beta_1(\vec{x}_1\vec{y}_1) \wedge \dots \wedge \beta_j(\vec{x}_j\vec{y}_j) )$,
where $\vec{x}_i$ and $\vec{y}_i$ denote the tuple of variables from $\vec{x}$ and $\vec{y}$ used by $\beta_i$.

If $\varphi$ is in $\acgnfminus$,
then there are no covering restrictions for this CQ
since it appears positively.
If $\instanceg, \vec{a}$ satisfies $\varphi(\vec{x})$,
then there exists $\vec{b}$,
such that $\instanceg,\vec{a}_i\vec{b}_i$ satisfies $\beta_i$ for all $1 \leq i \leq j$.
But $\instanceg' \supseteq \instanceg$, so this witness $\vec{b}$ and
the corresponding facts also appear in $\instanceg'$,
and $\instanceg',\vec{a}$ satisfies $\varphi$.

If $\varphi$ is in $\acgnfplus$
and $\instanceg', \vec{a}$ satisfies  $\varphi$,
then there is some $\vec{b}$ such
that $\instanceg', \vec{a}_i\vec{b}_i$ satisfies $\beta_i$ for all $1 \leq i \leq j$.
It suffices to show that $\instanceg, \vec{a}_i\vec{b}_i$ satisfies  $\beta_i$ for all $1 \leq i \leq j$.
Consider the possible $\beta_i$.
If $\beta_i$ is a $\sigmab$-atom, then $\instanceg, \vec{a}_i\vec{b}_i$ satisfies $\beta_i$, since $\instanceg$
has the same $\sigmab$-facts as $\instanceg'$.
If $\beta_i$ is a $\sigmad$-atom, then the covering requirements ensure that
there is some $\sigmab$-atom $\beta_j$ in $\varphi$
including at least the free variables $\vec{x}_i\vec{y}_i$ of $\beta_i$.
This means $\vec{a}_i\vec{b}_i$ is base-guarded.
Since $\instanceg$ and $\instanceg'$ agree on facts about base-guarded tuples like this, $\instanceg,\vec{a}_i\vec{b}_i$ satisfies $\beta_i$.
Finally, if $\beta_i$ is some structurally simpler $\acgnf$ formula,
then the inductive hypothesis ensures that
$\instanceg,\vec{a}_i\vec{b}_i$ satisfies $\beta_i$.
\end{proof}

\myparagraph{Final proof of Lemma~\ref{lemma:reducelin}}
\newcommand{\instancegtrans}{\mathcal{G}}
\newcommand{\instancegext}{\instanceg'}
We are now ready to prove Lemma~\ref{lemma:reducelin}:

We start with some $\instanceg \subseteq \instanceg_0$ satisfying $\Sigma' \wedge \neg Q$
with a 
$\instanceg_0$-rooted $(k-1)$-width base-guarded-interface tree decomposition.
We prove that there is an extension $\instancegext$ of $\instanceg$ satisfying $\Sigma' \wedge \neg Q$
in which each distinguished relation is a strict linear order.
Note that because $\instanceg$ satisfies $\Sigma'$,
we know that $\instanceg$ is $k$-guardedly linear.

We present the argument when
there is one $\drel$ in $\sigmad$ that is not a strict linear order in $\instanceg$,
but the argument is similar if there are multiple distinguished relations like
this, as we can handle each distinguished relation independently
with the method
that we will present.
Let $\instancegtrans$ be the extension of $\instanceg$ obtained by taking $\drel$ in $\instancegtrans$
to be the
transitive closure of $\drel$ in $\instanceg$.
Suppose for the sake of contradiction that there is a $\drel$-cycle in $\instancegtrans$.
We proceed by induction on the number of facts from $\instancegtrans \setminus \instanceg$ used in this cycle.
If there are no facts from $\instancegtrans \setminus \instanceg$ in the cycle,
Lemma~\ref{lem:nobadcycle} yields the contradiction.
Otherwise,
suppose that there is a cycle involving $(a_1,a_n)$,
where $(a_1,a_n)$ is a $\drel$-fact in $\instancegtrans \setminus \instanceg$ coming from
facts $(a_1,a_2), \dots, (a_{n-1},a_n)$ in $\instanceg$.
By replacing $(a_1,a_n)$ in this cycle with $(a_1,a_2), \dots, (a_{n-1},a_n)$,
we get a (longer) cycle with fewer facts from $\instancegtrans \setminus \instanceg$,
which is a contradiction by the inductive hypothesis.

Since $\drel$ is transitive in $\instancegtrans$,
the relation $\drel$ in $\instancegtrans$ must be a strict partial order.
We now apply the ``order extension principle'' or
``Szpilrajn extension theorem'' \cite{Szpilrajn30}:
any strict partial order can be extended to a strict total order.
From this, we deduce that $\instancegtrans$ can be further extended by additional $\drel$-facts to obtain some $\instancegext$
where $\drel$ is a strict total order.

We must prove that $\instancegext \supseteq \instancegtrans \supseteq \instanceg \supseteq \instanceg_0$
does not include any new $\drel$-facts
about base-guarded tuples.
Suppose for the sake of contradiction that there is a new fact $a \drel b$
in $\instancegext \setminus \instanceg$,
where $\{a,b\}$ is base-guarded in $\instanceg$.
By the guardedly total axiom, it must be the case that there was already
$b \drel a$ in $\instanceg$, and hence also in $\instancegext$.
But $a \drel b$ and $b \drel a$ in $\instancegext$ would together imply
$a \drel a$ in $\instancegext$,
contradicting the fact that $\instancegext$ is a strict linear order.

Hence, $\instanceg$ and $\instancegext$ agree on all facts about base-guarded tuples.
Since $Q$ is base-covered and $\Sigma \in \acgnf$,
$\Sigma \wedge \neg Q \in \acgnf$.
Thus, Lemma~\ref{lemma:acov} guarantees
that $\Sigma \wedge \neg Q$ is still satisfied in $\instancegext$.
Since $\instancegext$ also trivially satisfies all of the $k$-guardedly linear axioms,
it satisfies $\Sigma' \wedge \neg Q$ as required.

This concludes the proof of Lemma~\ref{lemma:reducelin},
and hence the proof of Proposition~\ref{prop:rewritelin}.

\clearpage

\section{Data complexity upper bounds for transitivity}

\subsection{Proof of Theorem~\ref{thm:conptransdataupper}}

We begin with the proof of Theorem   \ref{thm:conptransdataupper}.
Recall the statement:

\medskip

\begin{quote}
{\bf Theorem~\ref{thm:conptransdataupper}.} 
For any $\agnf$ constraints $\Sigma$ and CQ $Q$, given a finite
  set of facts $\instance_0$, we can decide $\owqatc(\instance_0, \Sigma, Q)$ in
  $\conp$ data complexity.
\end{quote}

\medskip

Fix the signature~$\sigma$.

For a set of $\sigma$-facts $\instance$, an \emph{$\instance,k$-rooted structure}
is one  that consists of $\instance$ unioned with sets of facts $T_{\vec c}$ for
$\vec c \in \dom(\instance)^k$ where the domain of $T_{\vec c}$ overlaps with the domain of $\instance$ only in~$\vec c$,
the  facts of $T_{\vec c}$ involving only elements of $\vec c$ are all present in $\instance$, and
for two $k$-tuples $\vec c$ and $\vec{c}'$,  the domain of $T_{\vec c}$ overlaps with the domain of $T_{\vec c'}$ only
within $\vec c \cap \vec{c}'$.

The following proposition follows from Proposition~\ref{prop:transdecomp}.

\begin{proposition} \label{prop:treelike}
For any set of $\sigma$-facts $\instance$, if a $\agnf$ sentence $\Sigma$ over~$\sigma$ is satisfiable by some set of facts containing $\instance$ with
relations $R_i^\trans$ interpreted as the transitive closure of~$R_i$,
then $\Sigma$ is satisfied (with the same restriction)
in an $\instance',k$-rooted structure, where $k$ is at most $\mysize{\sigma}$
 and $\instance'$ is a superset
of $\instance$ that has the same domain.
\end{proposition}

Let $\fo(\sigma)$ denote first-order logic over the signature $\sigma$.
Let $\fo(\sigma \cup \{d_1 \ldots d_k\})$ denote first-order logic over the signature $\sigma$
extended with $k$ new constants, which will be used to represent the overlap elements.
Note that formulas in both $\fo(\sigma)$ and $\fo(\sigma \cup \{d_1 \ldots d_k\})$
can make use of distinguished $R^\trans_i$ relations that are part of $\sigma$.

Given an $\instance,k$-rooted structure
$\frakA$, and  number $j$, the \emph{$j$-abstraction of $\frakA$}
is the expansion of $\instance$ with relations $P_\tau(x_1 \ldots x_k)$ for each 
 $\fo(\sigma \cup \{d_1 \ldots d_k\})$ sentence $\tau$
 of quantifier-rank $j$, up to logical equivalence (so there are finitely many
such relations). $P_\tau(x_1 \ldots x_k)$ is interpreted by the set of $k$-tuples   $\vec c$ such that $T_{\vec c}$ satisfies
$\tau$ when interpreting the constants in $\tau$ by $\vec{c}$.
We let $\sigma_{j,k}$ be the signature of the $j$-abstraction of such structures.

\begin{lemma} \label{lem:decomp}
For any sentence $\phi$ of $\fo(\sigma)$
and any $k$, there is 
$j$ having the following property:

Let $\frakA_1$ be an $\instance_1,k$-rooted structure for some set of $\sigma$-facts~$\fA_1$,
and
let $\frakA_2$ be an $\instance_2,k$-rooted structure for some set of $\sigma$-facts~$\fA_2$,
where the interpretations of the $R^\trans_i$ relations in each structure are the
transitive closure of the corresponding $R_i$ relations.
If the $j$-abstractions of $\frakA_1$ and $\frakA_2$ agree on all $\fo(\sigma_{j,k})$ sentences of quantifier-rank
at most $j$, then
$\frakA_1$ and $\frakA_2$ agree on~$\phi$.
\end{lemma}
\begin{proof}
Let $j_\phi$ be the quantifier-rank of $\phi$. We choose $j \defeq j_\phi\nolinebreak\cdot\nolinebreak k$.
We give a strategy for Duplicator in the $j_\phi$-round standard pebble game for $\fo(\sigma)$ over
$\frakA_1$ and $\frakA_2$.
With $i$ moves left to play, we will ensure the following invariants on a game position  consisting of a sequence $\vec p_1 \in \frakA_1$ and $\vec p_2 \in \frakA_2$:

\begin{itemize}
  \item Let $\vec{p_1}'$ be the subsequence of $\vec p_1$ that comes from
    $\instance_1$ and let $\vec{p_2}'$ be defined similarly
    for $\vec p_2$ and $\instance_2$. Then $\vec{p_1}'$ and $\vec {p_2}'$ should 
form a  winning position for Duplicator in the
  $i \cdot k$ round $\fo(\sigma_{j,k})$ game on the $j$-abstractions.

\item Fix any $k$-tuple $\vec c_1 \in \instance_1$ and let $P^1_{\vec c_1}$ be the subsequence of $\vec p_1$
that lies in $T_{\vec c_1}$ within $\frakA_1$. Then if $P^1_{\vec c}$ is non-empty, $\vec c_1$ also lies in $\vec p_1$.
 Let $\vec c_2$ be
the corresponding $k$-tuple to $\vec c_1$ in $\vec p_2$, and let $P^2_{\vec c_2}$ be the subsequence of $\vec p_2$
that lies in $T_{\vec c_2}$ within $\frakA_2$. Then  $P^1_{\vec c_1}$ and $P^2_{\vec c_2}$ form a winning position
in the $i$-round pebble game on  $T_{\vec c_1}$ and $T_{\vec c_2}$.

The analogous property holds fixing any $k$-tuple $\vec c_2 \in \instance_2$.
\end{itemize}

We now explain the strategy of the Duplicator, focusing for simplicity on moves of Spoiler within $\frakA_1$, with the strategy
on~$\frakA_2$ being similar. If Spoiler plays within $\instance_1$, Duplicator responds using
her strategy for the games on the $j$-abstractions of $\instance_1$ and $\instance_2$. It is easy to see that the invariant is preserved.

If Spoiler plays an element  within a substructure $T_{\vec c_1}$  within $\frakA_1$ that is already inhabited, then by the invariant
$\vec c_1$ is pebbled and there is a corresponding $\vec c_2$ in $\frakA_2$ with substructure $T_{\vec c_2}$ of $\frakA_2$ such that
the pebbles within $T_{\vec c_2}$  are winning positions in the game on $T_{\vec c_1}$ and $T_{\vec c_2}$ with $i$ moves left to play. Thus Duplicator can respond
using the strategy in this game from those positions.

Now suppose Spoiler plays an element  $e_1$ within a substructure $T_{\vec c_1}$  within $\frakA_1$ that is not already inhabited.
We first use  $\vec c_1$ as  a sequence of plays for Spoiler in the game on the $j$-abstractions of $\frakA_1$ and $\frakA_2$, extending the 
positions given by
$\vec{p_1}'$ and $\vec {p_2}'$.
By the inductive invariant, responses of Duplicator exist,
and we collect them  to get a tuple $\vec c_2$. 
Since a winning strategy in a game preserves atoms, and
we have a fact in the $j$-abstraction corresponding to the $j$-type of 
$\vec c_1$ in $T_{\vec c_1}$, we know that 
$\vec c_2$ must satisfy the same $j$-type in~$T_{\vec c_2}$
  that $\vec c_1$ does in $T_{\vec c_1}$.
Therefore $\vec c_1$ must satisfy the same $\fo(\sigma \cup \{d_1 \ldots d_k\})$ sentences of quantifier-rank at most 
$j$ in $T_{\vec c_1}$ as $\vec c_2$ does in $T_{\vec c_2}$.
Thus Duplicator can use the corresponding
strategy to respond to $e_1$ with an $e_2$ in $T_{\vec c_2}$
such that $\{e_1\}$ and $\{e_2\}$ are a winning position in the $i-1$ round
pebble game on $T_{\vec c_1}$ and $T_{\vec c_2}$.

Since the
response of Duplicator corresponds to $k$ moves in the  game within
the $j$-abstractions, one can verify that the invariant is preserved.

We must verify that this strategy gives a partial isomorphism. Consider a fact $F$ that holds of a tuple 
$\vec t_1$ within $\frakA_1$, and let $\vec t_2$ be  the  tuple obtained  using this strategy in $\frakA_2$.
We first consider the case where $F$ is a $\sigmab$-fact.
\begin{itemize}
\item If $\vec t_1$ lies completely within some $T_{\vec c_1}$, then the last invariant guarantees that $\vec t_2$
lies in some $T_{\vec c_2}$. The last invariant also guarantees that $\sigmab$-facts of $\frakA_1$ are preserved since such facts
must lie in $T_{\vec c_1}$, and the corresponding positions are winning 
in the game between   $T_{\vec c_1}$  and $T_{\vec c_2}$.

\item If $\vec t_1$  lies completely within $\instance_1$, then
the first invariant guarantees that the fact is preserved.
\end{itemize}
By the definition of a rooted structure, the above two cases are exhaustive.

We now consider the case where $F$ is  of the form
 $R^\trans_i(t_1, t_2)$. 
\begin{itemize}
\item If $t_1$ and $t_2$ both lie in some $T_{\vec c_1}$, then
we reason as in the first case above.
\item If $t_1$ and $t_2$ are both in $\instance_1$, we reason
as in the second case above.
\item If  $t_1$ lies in $T_{\vec c_1}$, $t_2$ lies in $T_{\vec c_2}$,
then $t_1$ reaches some $c_i$, $c_i \in \vec c_1$,
$c_i$ reaches some  $c_j \in \vec c_2$, and $c_j$ reaches
$t_2$ within $T_{\vec c_2}$. Then we use a combination of the first two cases
above to  conclude that $F$ is preserved.\qedhere
\end{itemize}

\end{proof}

From Lemma \ref{lem:decomp} we easily obtain:

\begin{corollary} \label{cor:composition} Given $\phi$ and $k$ there is a number $j$
and a sentence $\phi'$ in the language of $j$-abstractions over $\sigma$ such that
for all sets of facts $\instance$, an   $\instance,k$-rooted structure satisfies $\phi$ iff its $j$-abstraction satisfies $\phi'$.
\end{corollary}

We can now put these results together to prove Theorem  \ref{thm:conptransdataupper}:

\begin{proof}
Fixing $Q$ and $\Sigma$, we give an $\np$ algorithm for the complement. Let $\phi=\Sigma \wedge \neg Q$, and
$k = \mysize{\varphi}$.
Let $j$ and $\phi'$ be the number and formula guaranteed for $\phi$ by Corollary \ref{cor:composition}.

Let $\fo(\sigma \cup \{d_1 \ldots d_k\})$ denote first-order logic over the signature $\sigma$ of $\Sigma \wedge \neg Q$, together
with $k$ constants.

Let $\types_j$ be the collection of assignments of truth values to all $\fo(\sigma \cup \{d_1 \ldots d_k\})$ sentences with quantifier-rank at most $j$ such
that the conjunction of the corresponding sentences is consistent.
Note that the set is finite since $j$ and the signature are fixed. 

Given $\instance$, guess an extension $\instance'$ with additional facts but the same domain. 
Guess a function $f$ mapping each  $k$-tuple over $\instance$ to a  $\rho \in \types_j$, 
and then for each $\tau \in \fo(\sigma \cup \{d_1 \ldots d_k\})$ of quantifier rank at most $j$,
interpret $P_\tau$ by the set of tuples   $\vec c$ such
that $\tau \in f(\vec c)$.
Check whether $\instance'$ satisfies $\phi'$ with these interpretations,
and if so return true.

We argue for correctness. If the algorithm returns true with $\instance'$ the witness, then create an
$\instance',k$-rooted structure $\frakA$ by picking for each $\vec c$ a structure satisfying
the sentences in $f(\vec c)$ with distinguished
elements interpreted by $\vec c$ (such a structure exists by consistency of $f(\vec c)$), and letting the remaining domain elements be disjoint from the domain
of $\instance'$. Note that by construction, $\frakA$ has $\instance'$ as its $j$-abstraction.
By the choice of $j$ and $\phi'$, and the observation above, $\frakA$ satisfies $\Sigma \wedge \neg Q$.
Thus this structure
witnesses that $\owqatc(\instance, \Sigma, Q)$  is false.

On the other hand, if $\owqatc(\instance, \Sigma, Q)$  is false, then by Proposition \ref{prop:treelike}
we have an extension $\instance'$ without adding values to the domain, and an $\instance',k$-rooted structure $\frakA$
that satisfies $\Sigma \wedge \neg Q$. By the choice of $j$ and $\phi'$,  the $j$-abstraction of $\frakA$ satisfies
$\phi'$. For each $\vec c$   in the $j$-abstraction of  $\frakA$,
the type of $\vec c$ must be in $\types_j$. Hence we can guess collections such that
the algorithm returns true.
\end{proof}

\subsection{Proof of Theorem~\ref{thm:ptimetransdataupper}: $\ptime$ data complexity bound for $\owqatrans$}
We now turn to the case where our constraints are restricted to $\acfgtgd$s and deal with
$\owqatrans$, not $\owqatc$. Recall that Theorem \ref{thm:ptimetransdataupper} states
a $\ptime$ data complexity bound for this case:

\medskip

\begin{quote}
{\bf Theorem~\ref{thm:ptimetransdataupper}.} \ptimetransdataupper
\end{quote}

\medskip

The proof will follow from a reduction to traditional $\owqa$,
similar to the proof of Proposition~\ref{prop:rewritelin}:

\begin{proposition}
  \label{prop:reducetr}
  For any finite set of facts $\instance_0$,
  constraints $\Sigma \in \acgnf$,
  and
  base-covered UCQ $Q$,
  we
  can compute $\instance_0'$ and $\Sigma' \in \gnf$ in $\ptime$
  such that
  $\owqatr(\instance_0, \Sigma, Q) \text{ iff } 
  \owqa(\instance_0', \Sigma', Q)$.
  Furthermore,
  if $\Sigma$ is in $\acfgtgd$ then $\Sigma'$ is in $\fgtgd$.
\end{proposition}

\begin{proof}
\renewcommand{\drel}{R^+}
We define $\instance_0'$ and $\Sigma'$ as follows:
  \begin{itemize}
  \item $\instance_0'$ is $\instance_0$ together with facts $G(a,b)$ for every pair $a,b \in \elems{\instance_0}$
  for some fresh binary base relation $G$, and
  \item $\Sigma'$ is $\Sigma$ together with the $k$-guardedly-transitive axioms
  for each distinguished relation, where $k$ is  $\mysize{\Sigma \wedge \neg Q}$.
  \end{itemize}
These can be constructed in time polynomial in the size of the input.

As discussed in the proof of Lemma~\ref{lemma:guarded-interface-dec},
  the $k$-guardedly transitive axioms (see Appendix~\ref{app:owqalinred})
can be written in $\nf$ $\agnf$ with width at most $k$,
and hence in $\gnf$.

Now we prove the correctness of the reduction.
Suppose $\owqa(\instance_0',\Sigma',Q)$ holds,
so any $\instance' \supseteq \instance_0'$ satisfying $\Sigma'$ must satisfy $Q$.
Now consider $\instance \supseteq \instance_0$ that satisfies $\Sigma$
and where all $\drel$ in $\sigmad$
are transitive.
We must show that $\instance$ satisfies $Q$.
First, observe that $\instance$ satisfies $\Sigma'$
since the $k$-guardedly-transitive axioms for $\drel$
are clearly satisfied for all $k$
when $\drel$ is transitively closed.
Now consider the extension of $\instance$ to $\instance'$
with additional facts $G(a,b)$ for all $a,b \in \elems{\instance_0}$.
This must still satisfy $\Sigma'$:
adding these guards
means there are additional $k$-guardedly-transitive requirements
on the elements from $\instance_0$,
but these requirements already hold
since $\drel$ is transitively closed on all elements.
Hence, by our initial assumption, $\instance'$ must satisfy $Q$.
Since $Q$ does not mention $G$,
the restriction of $\instance'$ back to $\instance$ still satisfies $Q$ as well.
Therefore, $\owqa(\instance_0,\Sigma,Q)$ holds.

On the other hand, suppose for the sake of contradiction
that $\owqa(\instance_0',\Sigma',Q)$ does not hold,
but $\owqatr(\instance_0,\Sigma,Q)$ does.
Then there is some $\instance' \supseteq \instance_0'$ such that
$\instance'$ satisfies $\Sigma' \wedge \neg Q$, and hence also satisfies $\Sigma \wedge \neg Q$.
Since $\Sigma \wedge \neg Q$ is in $\agnf$,
Proposition~\ref{prop:guarded-interface-dec-appendix}
implies that we can take $\instance'$ to be a set of facts that has an
$\instance_0'$-rooted $(k-1)$-width base-guarded-interface tree decomposition.
Let $\instance''$ be the result of taking the transitive closure
of the distinguished relations in $\instance'$.
By Lemma~\ref{lem:guarded-transitivity},
transitively closing like this can only add $\drel$-facts
about pairs of elements that are not base-guarded.
Moreover, Lemma~\ref{lemma:acov} ensures that adding
$\drel$-facts about these non-base-guarded pairs of elements
does not affect satisfaction of $\acgnf$ sentences,
so $\instance''$ must still satisfy $\Sigma \wedge \neg Q$.
Restricting $\instance''$ to its $\sigma$-facts results in an $\instance$
where every distinguished relation is transitively closed
and where $\Sigma \wedge \neg Q$ is still satisfied,
since $\Sigma$ and $Q$ do not mention relation $G$.
But this contradicts the assumption that $\owqatr(\instance_0,\Sigma,Q)$ holds.

This concludes the proof of correctness.

Finally,
observe that the $k$-guardedly-transitive axioms can be written as $\fgtgd$s (in fact, $\afgtgd$s):
they are equivalent to the conjunction of $\fgtgd$s of the form
\begin{align*}
&\forall x\, y\, x_1 \dots x_{l+1}\, \big [ \big ( x = x_1 \wedge x_{l+1} = y \ \wedge \\
&\ \ \drel(x_1,x_2) \wedge \dots \wedge \drel(x_l,x_{l+1}) \wedge S(x,y) \big )
 \rightarrow
\drel(x,y) \big ]
\end{align*}
for all $S \in \sigmab \cup \set{ G }$,
$1 \leq l \leq k$, and $\drel \in \sigmad$.
Therefore, if $\Sigma$ is in $\acfgtgd$ then $\Sigma'$ is in $\fgtgd$
as claimed.
\end{proof}

Theorem~\ref{thm:ptimetransdataupper} easily follows from this.

\begin{proof}[Proof of Theorem~\ref{thm:ptimetransdataupper}]
  Recall that we have fixed constraints $\Sigma$ in $\acfgtgd$ and a base-covered UCQ $Q$.
  We must show $\ptime$ data complexity of $\owqatr(\instance_0,\Sigma,Q)$
  for any finite initial set of facts $\instance_0$.
  Use Proposition~\ref{prop:reducetr} to construct $\Sigma'$ from $\Sigma$ (in constant time,
  since $\Sigma$ is fixed) and $\instance_0'$ from $\instance_0$ (in time polynomial in $\mysize{\instance_0}$).
  Since $\Sigma$ is in $\acfgtgd$, $\Sigma'$ is in $\fgtgd$.
  Therefore, the $\ptime$ data complexity upper bound for $\owqatr$ with $\acfgtgd$s follows from
  the $\ptime$ data complexity upper bound for $\owqa$ with $\fgtgd$s \cite{bagetcomplexityfg}.
\end{proof}

\clearpage

\section{Hardness results}

\myparagraph{Chase}
In the proofs of this section and of subsequent sections, we will need the
standard database construction of the \emph{chase} \cite{ahv} by TGDs:

\begin{definition}
  \label{def:chase}
  The chase
applies to a set of facts $\instance$ and to a set $\Sigma$ of TGDs, and
constructs a set of facts $\instance' \supseteq \instance$, possibly infinite,
which satisfies $\Sigma$, in the following manner.

We first define a \emph{chase
round} as follows: for each TGD $\tau: \forall \vec{x} ~ \phi(\vec{x})
\rightarrow \exists \vec{y} ~ \psi(\vec{x}, \vec{y})$, for each
homomorphism $h$ from~$\vec{x}$ to the elements of $\instance$ such that
the facts of $\phi(h(\vec{x}))$ are in~$\instance$, if $h$ does not extend to
a homomorphism from $\vec{x} \cup \vec{y}$ such that the facts of
$\psi(h(\vec{x}), h(\vec{y}))$ hold in~$\instance$, then we call
$\phi(h(\vec{x}))$ a \emph{violation} of $\tau$ in~$\instance$: we repair it by
creating fresh
elements (called \emph{existential witnesses}) $\vec{b}$ for each variable of $\vec{y}$, and add to
$\instance$ the facts $\psi(h(\vec{x}), \vec{b})$.

Applying a chase round
means performing 
this process in parallel for all TGDs and violations,
creating fresh existential witnesses for each TGD and violation. The \emph{chase}
of~$\instance$ by~$\Sigma$ is the (potentially infinite) set of facts obtained
by repeated applications of chase rounds.
\end{definition}

When we use the chase, we will often use the fact that the result satisfies
$\Sigma$, and that all existentially quantified variables when applying rules
are instantiated by fresh existential witnesses (so no new facts are created on
an element unless it occurs on a fact which is part of a violation).

\subsection{Proof of Theorems~\ref{thm:tcdisj} and~\ref{thm:lindisj}}

The hardness results for $\owqatc$ and $\owqalin$
mentioned in the body
depend on the reductions described in Theorems~\ref{thm:tcdisj}~and~\ref{thm:lindisj}.
We start by proving Theorem~\ref{thm:tcdisj}, and we will adapt the proof
afterwards to show Theorem~\ref{thm:lindisj}.

Recall the result statement:

\begin{quote}
{\bf Theorem~\ref{thm:tcdisj}.} \tcdisj
\end{quote}
  
We start by creating a UCQ $Q'$, and then modify the
proof to make $Q'$ a CQ. Throughout the proof, whenever we talk of the $\did$s in~$\Sigma$,
we mean all dependencies of~$\Sigma$, including those where the head is a
trivial disjunction with only one disjunct.

\medskip

\myparagraph{Definition of the reduction}
Create the signature $\sigma'$ from $\sigma$ and $\Sigma$ by:
\begin{itemize}
  \item 
 adding a fresh binary base predicate $E$ and taking the transitive closure
$E^\trans$ of~$E$ as the one
distinguished relation of~$\sigma'$;
  \item replacing each predicate $R$ in~$\sigma$
with a base predicate $R'$ in~$\sigma'$ of arity $\arity{R}+2$;
\item adding to~$\sigma'$, for each $\did$ of the form \[\tau: \forall \vec{x} ~ R(\vec{x}) \rightarrow \bigvee_{1 \leq i \leq n}
  \exists \vec{y_i} ~ R_i(\vec{x}, \vec{y_i}),\]
a base predicate
$\witness_\tau(\vec{x}, \vec y_1, e_1, f_1, \ldots, \vec y_n, e_n, f_n)$.
\end{itemize}
We create $\Sigma'$ from $\Sigma$ by replacing each $\did$
$\tau: \forall \vec{x} ~ R(\vec{x}) \rightarrow \bigvee_{1 \leq i \leq n}
\exists \vec{y_i} ~ R_i(\vec{x}, \vec{y_i})$
by $\aincd$s equivalent to:
\begin{align*}
\forall \vec{x} \, e \, f \, R'(\vec{x}, e, f) &\rightarrow 
\exists \vec{y_1} e_1 f_1, \ldots, \vec{y_n} e_n f_n \\
	&\quad\,\,\, \witness_\tau(\vec{x}, \vec y_1, e_1, f_1, \ldots, \vec y_n, e_n, f_n) \wedge \\
	&\quad\,\,\, \bigwedge_i R'_i(\vec{x}, \vec{y_i},e_i, f_i) \wedge E^+(e_i,f_i) 
\end{align*}
Note that we have written the $\aincd$s for a given $\tau$ as a single TGD above with  multiple
conjuncts in the head, but we can easily rewrite them as multiple $\aincd$s for
$\tau$: the
first one has the same body and the $\witness_\tau$-fact as head atom, and the
others have the $\witness_\tau$-fact as body atom and each one of the other facts
as head atom.

The intuition for the proof is that a fact $R(\vec c)$  over the original schema will correspond to 
facts $R'(\vec c, e, f)$ in the new schema with fresh elements $e$ and $f$. The
fresh elements will always be connected by an $E$-path (as required by the
$E^\trans$-fact), which will be imposed
(via failure of the query) to have length $1$ or~$2$. Facts of this type with
a path of length~$1$ will be called \emph{genuine facts}, which intuitively
hold, and those with a path of length~$2$ will be called \emph{pseudo-facts}
and will be ignored by the query.

This mechanism allows us to eliminate disjunction from $\did$s as follows:
we require that, when the body atom holds, there 
are witness facts $R'_i(\vec c, \vec d_i, e_i, f_i)$
for \emph{all} of  the disjuncts. However, we will use the query to require
that, when the match of the body atom is a genuine fact, not all disjuncts can
be pseudo-facts, so one of them must be a genuine fact; the others can be made
pseudo-facts.
Note that $\Sigma'$ still requires matches for all of the disjuncts even when
the body is matched to a pseudo-fact; however, the query will only require that
one of the head atoms is matched to a genuine fact when the body is itself
matched to a genuine fact.
To this end, we will call the sequence $\vec c, \vec d_1, e_1, f_1, \ldots, \vec d_n, e_n, f_n$ will be called a \emph{witness vector} for
$R(\vec c)$ and $\tau$, and we capture such witness vectors in the predicate
$\witness_\tau$.

The UCQ $Q'$ contains the following disjuncts:
\begin{itemize}
\item \emph{$Q$-generated disjuncts}: One disjunct for each disjunct of the original UCQ $Q$, where
each atom $R(\vec{x})$ is replaced by the conjunction $R'(\vec{x}, e, f)
\wedge E(e, f)$, where $e$ and $f$ are fresh. That is, we have a witness for $Q$
consisting of genuine facts.

  \item \emph{$E$-path length restriction disjuncts}:
    For each predicate $R$ in~$\sigma$, we have a disjunct that succeeds if
    the $E$-path for an $R'$-fact has length~$\geq 3$, i.e., 
    $R'(\mathbf{x}, e, f) \wedge E(e, y_1) \wedge E(y_1, y_2) \wedge E(y_2,
y_3)$. Intuitively, for every $R'$-fact, the $E^\trans$-fact on its two last
    elements must make it either a genuine fact or a pseudo-fact.
\item  \emph{$\did$ satisfaction disjuncts}:
For every $\did$ $\tau: \forall \vec{x} ~ R(\vec{x}) \rightarrow
\bigvee_i \exists \vec{y_i} R_i(\vec{x}, \vec{y_i})$ in
$\Sigma$, 
we have a disjunct
\begin{align*}
  Q_\tau: & R'(\vec{x}, e,f) \wedge E(e,f) \\
  & \wedge  \witness_\tau(\vec{x}, \vec y_1, e_1, f_1, \ldots, \vec y_n, e_n, f_n) \\
  & \wedge \bigwedge_{1 \leq i \leq n} R'_i(\vec x, \vec y_i, e_i,f_i) \wedge \left(E(e_i, w_i) \wedge E(w_i, f_i)\right)
\end{align*}
Informally, the failure of $Q_\tau$ enforces that we cannot have the body of
    $\tau$ holding as a genuine fact and  each
of the components of the witness vector realized by a pseudo-fact. 
\end{itemize}
Observe that all of these disjuncts are trivially base-covered
(since they do not use $E^\trans$).

We now explain how to rewrite the facts of an initial fact set $\instance_0$ on
$\sigma$ to a fact set $\instance_0'$ on $\sigma'$.
Create $\instance_0'$ by replacing each fact $F = R(\vec{a})$ of $\instance_0$
by the facts $R'(\vec{a}, b_F, b'_F)$, and $E(b_F, b'_F)$, where $b_F$ and
$b'_F$ are fresh, so that they are genuine facts.

\medskip

\myparagraph{Correctness proof for the reduction}
We now show that the claimed equivalence holds: $\owqa(\instance_0, \Sigma,
Q)$ holds iff $\owqatc(\instance_0', \Sigma', Q')$ holds.

First, let $\instance \supseteq \instance_0$ satisfy $\Sigma$ and
violate $Q$. We must construct $\instance'$
that satisfies $\Sigma'$ and violates $Q'$ (when interpreting $E^+$ as the
transitive closure of~$E$).

We construct $\instance'$ using the following steps:
\begin{itemize}
  \item Modify $\instance$ in the same way that we used to build $\instance_0'$
    from $\instance_0$ (i.e., expand each fact with two fresh elements with
an $E$-edge between them), yielding $\instance_1$;
  \item We now need to ensure that witnesses exist as required by $\Sigma'$,
    which we will create as pseudo-facts. 

For every $\did$ $\tau$ of~$\Sigma$ and fact $F=R(\vec c)$ of~$\instance$ that
    matches the body of~$\tau$, as $\instance$ satisfies $\Sigma$, there is at least
    one $i_0$ such that some fact $R_{i_0}(\vec c, \vec d_{i_0})$ 
    witnesses that $\tau$ is not violated in~$\instance$.
    Call $I$ the set of such indices for which a witness exists in~$\instance$.
    We then know by construction that, for all $i \in I$, the set of facts $\instance_1$
    contains $F_{i} \defeq R'_{i}(\vec c, \vec d_{i}, e_{i}, f_{i})$
    and $F_{i}' \defeq E(e_{i}, f_{i})$ for some
    $e_{i}$ and $f_{i}$. For every $i \in \{1, \ldots, n\} \backslash I$,
    create a fresh $\vec d_i, e_i, f_i, w_i$
    and add a fact $R'_i(\vec c, \vec d_i, e_i,f_i)$ along with the facts
    $E(e_i, w_i)$ and $E(w_i, g_i)$ indicating that this is a pseudo-fact. We also add a 
fact $\witness_\tau(\vec c, \vec d_1, e_1, f_1 \ldots \vec d_n, e_n, f_n)$
    containing the witness vector consisting of the elements of the $F_i$ above
    (those that we created, for $i \in \{1, \ldots, n\} \backslash I$, and
    those that
    already existed, for $i \in I$).

We call $\instance_2$ the
    result of performing this process simultaneously in all places where it is
    applicable. Observe that, in~$\instance_2$, we have ensured that no rule of
    $\Sigma'$ has a violation whose body matches a genuine fact.
  \item The above process creates new pseudo-facts, and we also have to satisfy
    the rules of~$\Sigma'$ for these.
We create $\instance_3$ from $\instance_2$ by simply chasing with $\Sigma'$ wherever applicable (see Definition~\ref{def:chase}), always creating fresh elements; whenever
    we need a witness for some $E^{\trans}$ requirement, we always create an $E$-path of length $2$ with a fresh element
    in the middle, that is, we always create pseudo-facts.
\end{itemize}

Let $\instance' \defeq \instance_3$.
It is clear that
$\instance' \supseteq \instance_0'$ and that $E^\trans$ is indeed the transitive
closure of~$E$, and it is immediate by definition of the
chase that $\instance'$ satisfies $\Sigma'$, so we must check that $\instance'$
violates $Q'$, which we do by considering each kind of disjunct.

For the \emph{$E$-path length restriction disjuncts}
 observe that we only create paths of length $1$ or $2$ of~$E$ (of length
$1$ when creating $\instance_1$, and of length 2 when creating $\instance_2$ and $\instance_3$).
We always create these paths on fresh elements, so these paths of length 1 and 2 are never connected;
hence, there
is no $E$-path of length $3$ at all in~$\instance'$.

For the \emph{$\did$ satisfaction disjuncts}, assume by
contradiction that there is a match for a disjunct $Q_\tau$ of~$Q'$ in~$\instance'$.
Fix $\vec{c},e,f$ such that $R'(\vec{c}, e,f) \wedge E(e,f)$ holds,  a fact
$\witness_\tau(\vec{c}, \vec d_1, e_1, f_1 \ldots \vec d_n, e_n, f_n)$, and
  pseudo-facts  $R'_i(\vec {c}, \vec d_i, e_i, f_i)$ with $e_i,f_i$ connected by paths of length $2$.
The genuine fact $R'(\vec{c}, e,f)$  could not have been generated  within either of the second
or third steps in the creation of~$\instance'$ above, since all the $R'$-facts generated there have paths only of length $2$ between
the last two components (that is, they are pseudo-facts). Thus $R'(\vec{c}, e, f)$ must have been generated in the first step,
coming from  fact $R(\vec{c})$ in~$\instance$.
But then, as $\instance$ satisfies $\tau$, there must be $i_0 \leq n$ such that
$\instance$ contains $R_{i_0}(\vec c, \vec d_{i_0})$ for some $\vec d_{i_0}$, and
thus $\instance_1$ must contain~$R'_{i_0}(\vec c, \vec d_{i_0}, e_{i_0},
f_{i_0})$ and $E(e_{i_0}, f_{i_0})$ for some $e_{i_0}$ and $f_{i_0}$. 
Further, the $\witness_\tau$-fact must have been created in the second step
above, as the other $\witness_\tau$-facts are created during the third step,
where they only cover pseudo-facts rather than genuine facts. Hence,
in
creating the witness vector corresponding to $R(\vec{c}, e, f)$ in the second
step, we would not have generated a path of length $2$
for $e_{i_0}, f_{i_0}$ (as we would have had $i_0 \in I$), a contradiction.

Finally, for the \emph{$Q$-generated disjuncts},
observe that any match of them must be on
facts of~$\instance'$ created for facts of~$\instance$ (as they are annotated
by $E$-paths of length $1$), so we can conclude because $\instance$ violates
$Q$.

Hence, $\instance'$ satisfies $\Sigma'$ and violates $Q'$, which concludes the
first direction.

\medskip

In the other direction,
let $\instance' \supseteq \instance_0'$ be a counterexample
to $\owqatc(\instance_0', \Sigma', Q')$.
Consider the set of~$R'$-facts from $\instance'$
such that
$R' \in \sigma'$ corresponds to some $R \in \sigma$
and the elements in the last two positions of this $R'$-fact are connected by an
$E$-fact, i.e., the genuine facts.
Construct a set of facts $\instance$ on $\sigma$ by
projecting away the last two positions from these $R'$-facts,
and discarding all of the other facts.

It is clear by
construction of~$\instance_0'$ that $\instance \supseteq \instance_0$ and that
$E^\trans$ is indeed the transitive closure of~$E$ in~$\instance$. Further, as
$\instance'$ violates $Q'$, it is clear that $\instance$ violates $Q$, as any
match of a disjunct of~$Q$ on~$\instance$ implies a match of the corresponding
$Q$-generated disjunct $Q'$ in~$\instance'$. So it suffices to show
that $\instance$ satisfies $\Sigma$.

Assume by contradiction that some
$\did$ $\tau$ of $\Sigma$ is violated on a fact $F = R(\vec{x})$ of $\instance$
(that is $F$ matches the body of $\tau$).
Let $F' = R'(\vec{c},e,f)$
be the fact in~$\instance'$ from which we created $F$;
we know that the last two elements of
$F'$ are connected by an $E$-fact.
Since $\instance'$ satisfies $\Sigma'$,
we know that there are $\vec{d}_1, e_1, f_1, \dots, \vec{d}_n, e_n, f_n$
such that
$\witness_\tau(\vec{c}, \vec d_1, e_1, f_1, \ldots, \vec d_n, e_n, f_n)$
and
$\bigwedge_i R'_i(\vec{c}, \vec{d_i},e_i, f_i) \wedge E^+(e_i,f_i)$.
Moreover, since $E^{+}$ is the transitive closure of $E$
in~$\instance'$,
we know that for each $i$,
there is some $E$-path connecting $e_i$ and $f_i$.
By the $E$-path length-restriction disjuncts and $\did$ satisfaction disjuncts,
it must be the case that
there is an $E$-path of length at most 2 between each $e_i$ and $f_i$,
and for some $j$ there cannot be a path of length 2 between $e_j$ and $f_j$
(otherwise, $\instance'$ would satisfy the corresponding $\did$ satisfaction
disjunct in~$Q'$), so then the path must have length~$1$.
But this means that
$\instance'$ contains $R'_j(\vec{c},\vec{d}_j,e_j,f_j)$ and $E(e_j,f_j)$,
so $R_j(\vec{c},\vec{d}_j)$ is a fact in~$\instance$
witnessing the satisfaction of $\did$ $\tau$,
a contradiction.
Hence, $\instance$ satisfies $\Sigma$, which concludes the proof.

\medskip

\myparagraph{From UCQ to CQ}
Last, we explain how to replace the UCQ $Q'$ by a CQ. We do this by a general
process that we will reuse in several upcoming proofs: intuitively, we increase the arity to annotate
facts with an additional Boolean value carried over in dependencies and add an
$\Or$-relation to combine such values.

Formally,
define a signature $\sigma_\Or$ with
a ternary relation $\mathrm{Or}$ and a unary relation
$\mathrm{True}$.
Define a set of facts $\instance_{\mathrm{Or}}$
with two domain elements $\true$ and
  $\false$ that contains
  the fact $\mathrm{True}(\true)$ and
  the facts $\mathrm{Or}(b, b', b'')$ for all $\{(b, b', b \vee b')
    \mid b, b' \in \{\false, \true\}\}$.

Define $\sigma''$ from $\sigma'$ by increasing the arity of each relation
in~$\sigma'$ \emph{except $E$ and $E^\trans$} and adding the relations from $\sigma_\Or$.

Define $\Sigma''$ from $\Sigma'$ by adding a new variable $b$ which is
universally quantified and is put in the head and body facts.

Define $Q''$ from $Q'$ as follows:
\begin{itemize}
  \item Add to the atoms of each disjunct of $Q'$ (except $E$-atoms) one common
    variable which is shared between all atoms and left free: we call each
    resulting CQ  $Q_i(w_i)$, where $w_i$ is the new variable.
  \item Define the Boolean CQ $Q''$ as the following (existentially closed), where $m$ is the number of
    disjuncts of~$Q'$:
    \begin{align*}
      \Or(w_1, w_2, w'_1) \wedge \Or(w'_1, w_3, w'_2) \wedge \cdots \\
      \wedge\, \Or(w'_{m-2}, w_m, w'_{m-1})
      \wedge \mathrm{True}(w'_{m-1})\\
      \wedge \bigwedge_{1 \leq i \leq m} Q_i(w_i)
    \end{align*} 
\end{itemize}
Note that $Q''$ is trivially covered.

Define the set of facts $\instance_0''$ from $\instance_0'$ by:

\begin{itemize}
  \item Adding the facts of $\instance_\Or$;
  \item Putting $\true$ as the last element of all other facts except $E$-facts;
  \item Adding \emph{vacuous matches}: for each $Q_i(w_i)$, we add a set
    of facts that satisfy $Q_i(w_i)$, with $\false$ as the common
    last element of all facts, but the domains being otherwise disjoint.

    Intuitively, the purpose of the vacuous matches is to ensure that the
    $Q_i(w_i)$ always have a match but with $w_i$ set to false, and otherwise
    they have no purpose and they simply do not interact with the other facts.

    We accordingly call an element \emph{vacuous} in a set of facts if it occurs in
    no fact with $\true$ as the last element, and call a fact \emph{vacuous} if
    it is an $E$- or $E^\trans$-fact on vacuous elements, or it is a fact for
    another relation than $E$ but its last element is not $\true$.
\end{itemize}
We will now show the following equivalence: $\owqatc(\instance_0', \Sigma', Q')$
iff $\owqatc(\instance_0'', \Sigma'', Q'')$, which concludes the proof.

\medskip

In one  direction, we assume we have a counterexample $\instance'$ to $\owqatc(\instance_0', \Sigma', Q')$.
We construct $\instance''$ from  $\instance'$ by extending it according to the
process above (to define $\instance_0''$ from $\instance_0'$),
and chasing by
$\Sigma''$ on all facts from the vacuous matches (see
Definition~\ref{def:chase}), with $\false$ being propagated
as the last element, so the resulting elements and facts are all vacuous.
We claim that $\instance''$ witnesses the failure of $\owqatc(\instance_0'', \Sigma'', Q'')$.

It is clear that $\instance''$ is a superset of $\instance_0''$
and that $E^\trans$ is indeed interpreted as the transitive closure of~$E$. We argue
that $\Sigma''$ is satisfied by $\instance''$, by looking whether the required
witness facts exist for each type of fact.  
Vacuous facts have the required witnesses because we chased them in constructing $\instance''$.
No constraints of $\Sigma''$ hold about the facts from $\instance_\Or$. Finally, for the
facts of $\instance''$ created from 
facts of~$\instance'$, they have the required witnesses because $\Sigma'$ was satisfied by
$\instance'$ and the last position of such a fact is always $\true$ so the last
variable was correctly exported.

We now explain why
$\instance''$ violates $Q''$. Assuming by contradiction that $\instance''$
satisfies $Q''$, by definition of the $\Or$- and $\mathrm{True}$-facts that
$\instance''$ contains by construction, it must be the case that $\instance''$ satisfies
$Q_i(\true)$ for some $Q_i$. But it is then clear that $\instance'$ satisfies
the corresponding disjunct of $Q'$, as this match cannot involve any vacuous
facts. This proves one direction.

\medskip

For the other direction, we assume we have a counterexample $\instance''$ for
$\owqatc(\instance_0'', \Sigma'', Q'')$. We
construct $\instance'$ from $\instance''$ by
keeping only the facts in the base signature with last element $\true$ and
keeping precisely the $E$- and
$E^\trans$-facts that are connected to them.
It is clear that, as
$\instance'' \supseteq \instance''_0$, we have $\instance' \supseteq
\instance'_0$. To see that $\instance'$ satisfies $\Sigma'$, assume by
contradiction that $F'$ witnesses a violation of an $\incd$
$\tau'$ of $\Sigma'$ in
$\instance'$, and let $F''$ be the corresponding fact in~$\instance''$. By
definition of $\Sigma''$, there is a corresponding $\incd$ $\tau''$ in~$\Sigma''$ that asserts the
existence of a fact $F_2''$. So the only way $F'$ can violate $\tau'$ is that
$F_2''$ is an $E$-fact or that the last element of $F_2''$ is not $\true$, but
as the last element of $F''$ is $\true$, this is impossible.
Hence, we have a contradiction, and $F'$ satisfies $\Sigma'$.

The only thing left to show is that $\instance'$ violates $Q'$. Assuming to the
contrary that $\instance'$ satisfies some disjunct of~$Q'$, we know that,
considering the corresponding $Q_i(w_i)$, $\instance''$ satisfies $Q_i(\true)$.
Now, from the facts in~$\instance_\Or \subseteq \instance''$, and from the
vacuous matches and their connected $E$ -and $E^+$-facts, we know that we can construct a match of the entire CQ $Q''$ in
$\instance''$, a contradiction as $\instance''$ violates $Q'$. This concludes
the correctness proof, and concludes the proof of Theorem~\ref{thm:tcdisj}.

\bigskip

We now prove Theorem~\ref{thm:lindisj}, which states:

\begin{quote}
{\bf Theorem~\ref{thm:lindisj}.} \lindisj
\end{quote}

The entire proof is shown by adapting the proof of Theorem~\ref{thm:tcdisj}. We start by showing the claim with a UCQ.
Intuitively, instead of using $E^\trans$ to emulate a disjunction on the length of
the path to encode genuine facts and pseudo facts, we will use the order
relation to emulate disjunction on the same elements: 
$e < f$ will indicate a genuine fact,
whereas $f < e$ will indicate a pseudo-fact, and $e = f$ will be
prohibited by the query.

\myparagraph{Definition of the reduction}
We define $\sigma'$ as in the proof of Theorem~\ref{thm:tcdisj},
except that we do not add the predicates $E$
and $E^\trans$, but add a predicate ${<}$ as a distinguished relation instead.
We also define $\Sigma'$ as before except that we drop all mention of $E^+$.

The UCQ $Q'$ contains the following disjuncts (existentially closed):
\begin{itemize}
\item \emph{Order restriction disjuncts}: For each $R \in \sigma$, we have a disjunct
    $R'(\vec{x},e,e)$ to enforce disjunction between genuine facts and
    pseudo-facts.
\item \emph{$Q$-generated disjuncts}:
Each disjunct of the original UCQ $Q$, where
each atom $R(\vec{x})$ is replaced by the conjunction $R'(\vec{x}, z, z')
\wedge z < z'$, where $z$ and $z'$ are fresh.
That is, we have a witness for $Q$
consisting of genuine facts.

\item \emph{$\did$ satisfaction disjuncts}:
For every $\did$ $\tau: \forall \vec{x} ~ R(\vec{x}) \rightarrow
\bigvee_i \exists \vec{y_i} R_i(\vec{x}, \vec{y_i})$ in
$\Sigma$, 
we have a disjunct
\begin{align*}
Q_\tau: R'(\vec{x}, e,f) \wedge e < f \\
\wedge\, \witness_\tau(\vec{x}, \vec y_1, e_1, f_1 \ldots \vec y_n, e_n, f_n) \\ 
\wedge\, \bigwedge_{1 \leq i \leq n} R'_i(\vec x, \vec y_i, e_i,f_i) \wedge f_i < e_i
\end{align*}
Intuitively, $Q_\tau$ is satisfied if the body of  $\tau$ is matched to a genuine
    fact but each
of the components of the witness vector is matched to a pseudo-fact. 
\end{itemize}
Observe that all CQs of the resulting UCQ are base-covered, as required.

The process to define $\instance_0'$ from $\instance_0$ is
defined like in the proof of Theorem~\ref{thm:tcdisj}
except that we remove the
$E$-facts and replace them by $b_F < b'_F$.

\medskip

\myparagraph{Correctness proof for the reduction}
The proof that $\owqa(\instance_0, \Sigma,
Q)$ holds iff $\owqalin(\instance_0', \Sigma', Q')$ holds
is similar to the proof for Theorem~\ref{thm:tcdisj},
so we sketch the proof and highlight the main differences.

For one direction, let $\instance \supseteq \instance_0$ satisfy $\Sigma$ and
violate $Q$.
We construct $\instance'$ from $\instance$ as follows:
\begin{itemize}
  \item Construct   $\instance_1$ from $\instance$ as we constructed
    $\instance_0'$ from $\instance_0$ above.
  \item The construction of $\instance_2$ and $\instance_3$ is as before,
  except that we create a $<$-fact to indicate a pseudo-fact.
  \item The new step is that $\instance'$ is constructed from $\instance_3$
    by completing $<$ to
    be a total order. To do so, however, we must ensure that our
    definition of $<$ in~$\instance_3$ does not contain any cycles. This is easy to see,
    however: we only imposed an order relation between \emph{disjoint pairs} of
    elements. Hence, it is clear that $<$ cannot contain any loop, so we can
    simply complete this partial order to a total order using the order
    extension principle \cite{Szpilrajn30}.
\end{itemize}

As before it is clear that $\instance' \supseteq \instance'_0$ and that
$\instance'$ satisfies $\Sigma'$, and we have made sure that $<$ is a total
order. To see why $Q'$ is not satisfied in~$\instance'$, we proceed exactly as
before for the $\did$ satisfaction disjuncts and $Q$-generated disjuncts,
but replacing ``having an $E$-fact between $e$ and $f$'' by ``having $e
< f$'', and replacing ``having an $E$-path of length $2$ between $e$ and $f$'' by
``having $e > f$'', and likewise for $e_i$ and $f_i$. By
construction, we never have $e = f$ or $e_i = f_i$ in any fact within~$\instance'$,
so we also do not match the order-restriction disjuncts in~$Q'$.

For the other direction, suppose we have some counterexample $\instance'$
to $\owqatc(\instance_0',\Sigma',Q')$.
We construct $\instance$ from $\instance'$ by
keeping all facts whose last two elements $e$ and $f$ are such that $e < f$. The
result still clearly satisfies $\instance \supseteq \instance_0$, and the proof
of why it violates
$Q$ is unchanged. To show that $\instance$ satisfies $\Sigma$, we adapt the
argument of the proof of Theorem~\ref{thm:tcdisj},
but instead of relying on the $E$-path length disjuncts
we rely on totality of the order and the order-restriction disjuncts.
Totality of the order ensures that for
fact $F_\tau$, and for all $i$, we have either $e_i < f_i$, $e_i = f_i$
and $f_i < e_i$. But the order-restriction disjuncts are violated,
so it must be either $e_i < f_i$ or $f_i < e_i$,
and the $\did$ satisfaction disjuncts of  $Q$ are violated,
so we must have $e_i < f_i$ for some~$i$.
Hence, we can argue as before that the satisfaction of $\Sigma'$ by $\instance'$
ensures that $\Sigma$ is satisfied in~$\instance$.

\myparagraph{From UCQ to CQ}
The proof from UCQ to CQ works exactly like before, except that we do not
increase the arity of~$<$ (recall that we did not increase the arity of~$E^\trans$
and~$E$), and we use $\owqalin$ instead of $\owqatc$.
When showing that we can construct a counterexample
to $\instance''$ to $\owqalin(\instance_0'',\Sigma'',Q'')$ from
a counterexample to $\owqalin(\instance_0',\Sigma',Q')$,
we make $<$ a total order in~$\instance''$ using again
the order extension principle (the order on vacuous matches, and on the domain
elements of $\instance_\Or$, is arbitrary). Observe that the resulting CQ is
clearly base-covered, as all disjuncts of the UCQ $Q'$ were base-covered.

\subsection{Proof of Propositions~\ref{prop:lindatacompltrans}~and~\ref{prop:lindatacompl}}

We now give data complexity lower bounds that show
$\conp$-hardness even in the absence of constraints.

We first prove Proposition~\ref{prop:lindatacompltrans}:

\begin{quote}
{\bf Proposition~\ref{prop:lindatacompltrans}.} \lindatacompltrans
\end{quote}

\begin{proof}
  We first prove the result for a UCQ $Q$, and then for a CQ $Q'$.
  
\myparagraph{Definition of the reduction}
  We define the signature $\sigma$ as containing:
  \begin{itemize}
    \item one binary predicate $E$ and its transitive closure $E^\trans$ (again
      playing a similar role as in the proof Theorem~\ref{thm:tcdisj});
    \item one binary relation $G$ to code the edges of a graph which will be
      provided as input to the reduction;
    \item one $7$-ary relation $V$ to code vertices and their color. 
The idea is that
one position is for the vertex and then for each of the $3$ colors we will have two positions that will
encode whether or not the vertex has that color. If the positions associated with a color $C$ are connected
by an $E$-edge, this will indicate coloring the vertex with color $C$, while if they are connecting by a path of length $2$ this will
indicate not being colored with color $C$.
  \end{itemize}

  We then define the UCQ $Q$ to contain the following disjuncts (existentially
  closed):
  \begin{itemize}
  \item \emph{$E$-path length restriction disjuncts}:
    For each predicate $R$ in~$\sigma$, we enforce that the $E$-path for the
      $R'$-fact has length $\geq 3$:
    $R'(\mathbf{x}, e, f) \wedge E(e, y_1) \wedge E(y_1, y_2) \wedge E(y_2,
y_3)$.
  \item \emph{Adjacency disjuncts:} For $i \in \{1, 2, 3\}$, the disjunct $Q_i$ that succeeds if two
      adjacent vertices were assigned the same color:
      \begin{align*}
        V(x, e_1, f_1, e_2, f_2, e_3, f_3) \wedge G(x, x') \\
        \wedge V(x', e_1', f_1', e_2', f_2',
        e_3', f_3') \wedge 
        E(e_i, f_i) \wedge E(e_i', f_i')
      \end{align*}
    \item \emph{Coloring disjunct:} A disjunct that succeeds if a vertex was not assigned any color: $V(x,
      e_1, f_1, e_2, f_2, e_3, f_3) \wedge \bigwedge_{i \in \{1, 2, 3\}} E(e_i,
      w_i) \wedge E(w_i, f_i)$
  \end{itemize}

  Given a directed graph $\calG$, we code it in~$\ptime$ as the instance $\instance_0$
  defined by having:
  \begin{itemize}
    \item One fact $G(x, y)$ for each edge $(x,y)$ in~$\calG$
    \item The facts $V(x, e_{x,1}, f_{x,1}, e_{x,2}, f_{x,2}, e_{x,3},
      f_{x,3})$ and $E^\trans(e_{x,i}, f_{x,i})$ for $i \in \{1, 2, 3\}$
      for each vertex $x$ in~$G(x, y)$, where all the $e_{x,i}$ and $f_{x,i}$
      are fresh.
  \end{itemize}

\myparagraph{Correctness proof for the reduction}
  We now show that $\calG$ is 3-colorable iff $\owqatc(\instance_0, \emptyset,
  Q)$ is false, completing the reduction.

  First, consider a $3$-coloring of $\calG$. Construct
  $\instance \supseteq \instance_0$ as follows.
  For each vertex $x$ of
  $\calG$
  (with corresponding $V$-fact $V(x, e_{x,1}, f_{x,1}, e_{x,2}, f_{x,2}, e_{x,3},f_{x,3})$ as defined above)
  create the facts
  $E(e_{x,i}, f_{x,i})$ where $i$ is the color assigned to $x$, and the facts $E(e_{x,j},
  w_{x,j})$ and $E(w_{x,j}, f_{x,j})$ for the other colors $j \in \{1, 2, 3\} \backslash
  \{i\}$ (with the two $w_{x,j}$ being fresh). It is clear that $\instance$ thus
  defined is such that $\instance \supseteq \instance_0$,
  and that $E^\trans$ is
  the transitive closure of~$E$ in~$\instance$.
  The $E$-path length restriction disjuncts of $Q$ do not match in~$\instance$
  (note that we only create $E$-paths whose endpoints are pairwise distinct),
  and the coloring disjunct does not match either.
  Finally, the fact that we have a $3$-coloring ensures that the adjacency disjuncts
  do not match either.
  Hence, we have a set of facts violating $Q$.

  \medskip

  For the other direction, consider some $\instance \supseteq \instance_0$ that
  violates $Q$. Since $\instance$ violates the first and last disjunct of $Q$
  and $E^\trans$ is the transitive closure of~$E$,
  any vertex $x$ of~$\calG$
  (with corresponding $V$-fact $V(x, e_{x,1}, f_{x,1}, e_{x,2}, f_{x,2}, e_{x,3},f_{x,3})$ defined above)
  there must be an $E$-path of
  length~$1$ or~$2$ from $e_{x,i}$ to $f_{x,i}$ for all $i \in \{1, 2, 3\}$. Further, as
  $\instance$ violates the last disjunct of~$Q$, at least one of these paths must
  have length~1. Define a coloring of $\calG$ by giving each vertex $x$ a
  color $i$ such that $E(e_{x,i}, f_{x,i})$ holds in the $V$-fact for~$x$.
  This indeed defines a 3-coloring, as any
  violation of the 3-coloring witnessed by two adjacent vertices of color $i$
  would imply a match of $Q_i$ in~$\instance$.

  \medskip

\myparagraph{From UCQ to CQ}
  We replace the UCQ $Q$ by a CQ $Q'$ in the same manner as in the proof of
  Theorem~\ref{thm:tcdisj}: we increase the arity of all predicates  and add the
  relations of $\sigma_\Or$, add to $\instance_0$ the facts of~$\instance_\Or$
  and the vacuous matches, and rewrite the query as in the proof of
  Theorem~\ref{thm:tcdisj}.
  We can then adapt the argument of that proof to show 
  that the resulting $\owqatc$ problem with the CQ is equivalent to the
  previously defined problem with a UCQ.
\end{proof}

We then modify the proof to show Proposition~\ref{prop:lindatacompl}:

\begin{quote}
{\bf Proposition~\ref{prop:lindatacompl}.} \lindatacompl
\end{quote}

\begin{proof}
  We define $\sigma$ as in the previous proof but with an order relation $<$
  and
  without $E$, $E^\trans$.
  We define $Q$ as in the proof
  of Proposition~\ref{prop:lindatacompltrans} but
  without its first disjunct, and replacing in the other disjuncts $E(e_i, f_i)$
  by $e_i < f_i$, and $E(e_i, w_i) \wedge E(w_i, f_i)$ by $f_i < e_i$. Unlike
  in the proof of Theorem~\ref{thm:lindisj}, we need not worry about equalities
  (and we need not add order restriction disjuncts),
  as all the elements of relevant $V$-facts are created already in
  $\instance_0$, where they are created as distinct elements.
  We define $\instance_0$ in the same fashion as in the proof of
  Proposition~\ref{prop:lindatacompltrans} but without the
  $E^\trans$-facts.

  We prove the same equivalence as in that proof but for $\owqalin$. We do it by
  replacing $E$-paths of length $1$ from an $e$-element to an $f$-element by $e
  < f$, and $E$-paths of length $2$ by $f < e$.

  We replace the UCQ by a CQ exactly as in the other proof. As in the proof of
  Theorem~\ref{thm:lindisj}, the order on the vacuous
  matches is arbitrary.
\end{proof}

\clearpage

\section{Undecidability results related to transitivity (from Section~\ref{sec:undecidtrans})}

We first prove the second result as it is simpler to understand.

\begin{quote}
{\bf Theorem~\ref{thm:undectransb}.} \undectransb
\end{quote}

\begin{proof}
  As in Proposition~\ref{prop:lindatacompl}, we first prove the result with a
  UCQ and then modify the proof to use a CQ.

  An \emph{infinite tiling problem} is specified by a set of colors $\mathbb{C} = C_1,
  \ldots, C_k$, a set of forbidden \emph{horizontal} patterns $\mathbb{H}
  \subseteq \mathbb{C}^2$ and a set of forbidden \emph{vertical} patterns
  $\mathbb{V} \subseteq \mathbb{C}^2$. It asks, given a sequence $c_0, \ldots,
  c_n$ of colors of $\mathbb{C}$, whether there exists a function $f :
  \mathbb{N}^2 \rightarrow \mathbb{C}$ such that $f((0, i)) = c_i$ for all $0
  \leq i \leq n$, and for all $i, j \in \mathbb{N}$, we have $(f(i, j), f(i+1,
  j)) \notin \mathbb{H}$ and $(f(i, j), f(i, j+1)) \notin \mathbb{V}$.

  It is well-known that we can take $\mathbb{C}$, $\mathbb{V}$, $\mathbb{H}$
  such that the corresponding tiling problem is undecidable; we fix such a
  problem.

  \myparagraph{Definition of the reduction}
  We define a binary relation $S'$ (for ``successor''), a transitive relation $S^\trans$,
  one binary relation $K_i$ for each color $C_i$, and one unary relation $K_i'$
  for each color $C_i$.

  We write the following $\did$s $\Sigma$ (note that they are not base-guarded),
  dropping universal quantification for brevity:
  \begin{align*}
    S'(x, y) & \rightarrow \exists z ~ S'(y, z)\\
    S'(x, y) & \rightarrow S^\trans(x, y)\\
    S^\trans(x, y) & \rightarrow \bigvee_i K_i(x, y)\\
    S^\trans(x, y) & \rightarrow \bigvee_i K_i(y, x)\\
    S^\trans(x, y) & \rightarrow \bigvee_i K_i'(x)
  \end{align*}
  Intuitively, $K_i'(x)$ stands for $K_i(x, x)$, but we need a different
  predicate because variable reuse is not allowed in inclusion dependencies.

  The UCQ $Q$ is a disjunction of the following disjuncts (existentially closed):
  \begin{itemize}
    \item For each forbidden horizontal pair $(C_i, C_j) \in \mathbb{H}$, with
      $1 \leq i, j \leq k$, the disjuncts:
      \begin{align*}
        K_i(x, y) \wedge S'(y, y') \wedge K_j(x, y')\\
        K_i'(y) \wedge S'(y, y') \wedge K_j(y, y')\\
        K_i(y', y) \wedge S'(y, y') \wedge K_j'(y')
      \end{align*}
    \item For each forbidden vertical pair $(C_i, C_j) \in \mathbb{V}$, the
      analogous disjuncts:
      \begin{align*}
        K_i(x, y) \wedge S'(x, x') \wedge K_j(x', y)\\
        K_i'(x) \wedge S'(x, x') \wedge K_j(x', x)\\
        K_i(x, x') \wedge S'(x, x') \wedge K_j'(x')
      \end{align*}
  \end{itemize}
  Given an initial instance of the tiling problem $c_0, \ldots, c_n$, we encode
  it in the initial set of facts $\instance_0$:
  \begin{itemize}
    \item $S'(a_i, a_{i+1})$ for $0 \leq i < n$;
    \item for $0 < i \leq n$,
      the fact $K_j(a_0, a_i)$ such that $C_j$ is the color of initial element
      $c_i$;
    \item the fact $K_j'(a_0)$ such that $C_j$ is the color of $c_0$.
  \end{itemize}

  \myparagraph{Correctness proof for the reduction}
  We claim that the tiling problem has a solution iff
  there is a (generally infinite) superset of $\instance_0$ that satisfies $\Sigma$ and
  violates $Q$ and where $S^\trans$ is transitive. From this we conclude the
  reduction and deduce the undecidability of $\owqatr$ as stated.

  \medskip

  For the forward direction, from a solution $f$ to the tiling problem for
  input $\vec{c}$, we construct the counterexample $\instance \supseteq \instance_0$ as follows.
  We first create 
  an infinite chain $S'(a_0, a_1), \ldots, S'(a_m, a_{m+1}), \ldots$ to complete
  the initial chain of $S'$-facts in~$\instance_0$, and fix $S^\trans$ to be the
  transitive closure of this $S'$-chain (so it is indeed transitive).
  For all $i, j \in \mathbb{N}$ such that $i \neq j$, we create the
  fact $K_l(a_i, a_j)$ where $l = f(i,j)$. For all $i \in \mathbb{N}$, we create
  the fact $K_l'(a_i)$ where $l = f(i,i)$. This clearly satisfies the
  constraints in $\Sigma$, and does not satisfy the query
  because $f$ is a tiling.

  \medskip

  For the backward direction, consider a  $\instance \supseteq \instance_0$ that
  satisfies $\Sigma$ and violates $Q$. Starting at the chain of $S'$-facts
  of~$\instance_0$, we can deduce, using the constraints,  the existence of an infinite chain $a_0, \ldots, a_n,
  \ldots$ of~$S'$-facts
  (whose elements may be distinct or not, this does not matter). Define a
  tiling $f$ matching the initial tiling problem instance as follows. For all
  $i < j$ in $\mathbb{N}$, as there is a path of $S'$-facts from $a_i$ to $a_j$,
  we infer that $S^\trans(a_i, a_j)$ holds, so that $K_l(a_i, a_j)$ holds for
  some $1 \leq l \leq k$; pick one such fact, taking the fact of $\instance_0$
  if $i = 0$ and $j \leq n$, 
  and fix $f(i, j) \defeq l$. For $i > j$ we can likewise see
  that $S^\trans(a_j, a_i)$ holds whence $K_l(a_i,a_j)$ holds for some~$l$, and
  we continue as before. For $i \in \mathbb{N}$, as $S'(a_i, a_{i+1})$ holds, we
  know that $K_l(a_i)$ holds for some $1 \leq l \leq k$ (again we take the
  fact of~$\instance_0$ if $i=0$), and fix accordingly
  $f(i, i) \defeq l$. The resulting $f$ clearly satisfies the initial tiling
  problem instance $c_0, \ldots, c_n$, and it is clearly a solution to the
  tiling problem, as any forbidden pattern in~$f$ would witness a match of a CQ
  of~$Q$ in~$\instance$. This shows that the reduction is correct, and
  concludes the proof with the UCQ $Q$.

  \myparagraph{From UCQ to CQ}
  We now adapt the proof to use a CQ, similarly to the proof of
Theorem~\ref{thm:tcdisj}: we
  use the signature $\sigma_{\mathrm{Or}}'$ constructed by extending $S'$ (and
  \emph{only} $S'$) to have a third position, and adding the relations of
  $\sigma_\Or$ (see the proof of Theorem~\ref{thm:tcdisj}).
  We modify $\Sigma$ by
  changing the first dependency to propagate also the third element of the
  $S'$-atoms, i.e., $\forall x \, y \, b ~ S'(x, y, b) \rightarrow \exists z ~ S'(y,
  z, b)$, and replace the second dependency similarly by $\forall x \, y \, b ~ S'(x,
  y, b) \rightarrow S^\trans(x, y)$.
  We construct the CQ $Q'$ from the UCQ $Q$
  by adding one variable $b$ to each $S'$-fact of each disjunct, 
  and connecting these disjuncts on their free variables with
  $\mathrm{Or}$-facts and adding a $\mathrm{True}$-fact as in
  the proof of Theorem~\ref{thm:tcdisj}. 
  
  We then define the initial set of facts $\instance_0'$ to 
  include
  the facts of $\calF_{\mathrm{Or}}$,
  the facts of~$\instance_0$ where each $S'$ is extended by adding $\true$ as
  its third element, and
  vacuous matches for each $Q_i$ (defined as in the proof of
  Theorem~\ref{thm:tcdisj},
  with the element $\false$ at the third position of each $S'$-fact, all the
  vacuous matches having pairwise disjoint domains except for~$\false$).

  \medskip

  To show the forward direction, we construct $\instance$ from $\instance_0$ as
  before but with $\true$ at the third position of all created $S'$-facts, plus
  a completion of the vacuous matches obtained by chasing (see
  Definition~\ref{def:chase}). Like in the proof of Theorem~\ref{thm:tcdisj},
  any match of the CQ $Q'$ must imply
  a match of one of the disjuncts of the UCQ $Q$, and as this match must be on
  an $S'$-fact with $\true$ as third position, it cannot involve a vacuous
  match, so we conclude as before.
  
  Conversely, for the backward direction, we
  define the tiling analogously to what we did before, and we observe that any
  violation of the tiling property would imply a match of one disjunct of the
  UCQ $Q$, which we can extend thanks to the vacuous matches to a match of $Q'$.
  This concludes the proof.
\end{proof}

We now prove the first statement, drawing inspiration from the previous proof, but using 
the transitive closure to emulate disjunction as in Theorem~\ref{thm:tcdisj}.
Recall the statement:

\begin{quote}
{\bf Theorem~\ref{thm:undectrans}.} \undectrans
\end{quote}

\begin{proof}
  We reuse the notations for tiling problems from the previous proof.
  We first prove the result with two distinguished relations $S^\trans$ and
  $C^\trans$ and with a UCQ, and then explain how the proof is modified to use only a single
  transitive relation $S^\trans$, and finally explain how to adapt the proof to
  use a CQ rather than a UCQ.

  \myparagraph{Definition of the reduction}
  We define a binary relation $S$ (for ``successor'') of which $S^\trans$ is
  interpreted as the transitive closure, one binary relation $S'$,
  one 3-ary relation $G$ (for ``grid''),
  one binary relation $G'$ (standing for the diagonal cells of the grid),
  one binary relation $T$ (a terminal for gadgets that we will define to
  indicate colors)
  and one binary relation $C$ of which
  $C^\trans$ is interpreted as the transitive closure. The distinction between
  $S$ and $S'$ is not important for now but will be important when we adapt the
  proof later to use a single distinguished relation.

  We write the following inclusion dependencies $\Sigma$ (with universal
  quantification dropped for brevity):
  \begin{align*}
    S'(x, y) & \rightarrow \exists z ~ S'(y, z)\\
    S'(x, y) & \rightarrow S(x, y)\\
    S^\trans(x, y) & \rightarrow \exists z ~ G(x, y, z)\\
    S^\trans(y, x) & \rightarrow \exists z ~ G(x, y, z)\\
    S^\trans(x, y) & \rightarrow \exists z ~ G'(x, z)\\
    G(x, y, z) & \rightarrow \exists w ~ T(z, w)\\
    G'(x, z) & \rightarrow \exists w ~ T(z, w)\\
    T(z, w) & \rightarrow C^\trans(z, w)
  \end{align*}

  In preparation for defining the query $Q$, we define $Q_i(z)$
  for all $i > 0$ to match the left endpoint of $T$-facts covered by a $C$-path
  of length $i$ (intuitively coding color~$i$):
  \[
    \exists z_1 \ldots z_i \, w ~ C(z, z_1) \wedge C(z_1, z_2) \wedge \ldots,
    C(z_{i-1}, z_i) \wedge
    T(z, z_i),
  \] 

The query $Q$ is a disjunction of the following disjuncts (existentially closed):
  \begin{itemize}
    \item \emph{$C$-path sanity disjuncts:} One disjunct written as follows,
      where $k$ is the number of colors
  \begin{align*}
    G(x, y, z) \wedge S'(x, w) \wedge\, T(z, z) \wedge C(z, z_1) \\
    \wedge\, C(z_1, z_2) \wedge \cdots \wedge C(z_{k-1}, z_k)
    \wedge
    C^\trans(z_k, z')
  \end{align*}
      and one disjunct defined similarly but with $G(x, y, z)$ replaced by
      $G'(x, z)$. Intuitively, these disjuncts impose that $C$-paths that cover
      $T$-facts must code colors between $1$ and $k$, and the distinction
      between $G$ and $G'$ is for reasons similar to the distinction between the
      $K_i$ and $K'_i$ in the proof of Theorem~\ref{thm:undectransb}.
    \item \emph{Horizontal adjacency disjuncts:}
      For each forbidden horizontal pair $(C_i, C_j) \in \mathbb{H}$, with
      $1 \leq i, j \leq k$, the disjuncts:
        \begin{align*}
G(x, y, z) \wedge G(x, y', z')
\wedge Q_i(z) \wedge Q_j(z') \wedge S'(y, y')\\
G'(y, z) \wedge G(y, y', z') \wedge Q_i(z) \wedge Q_j(z') \wedge S'(y, y')\\
G(y', y, z) \wedge G'(y', z') \wedge Q_i(z) \wedge Q_j(z') \wedge S'(y, y')
\end{align*}
    \item \emph{Vertical adjacency disjuncts:}
      For each $(C_i, C_j) \in \mathbb{V}$, the same queries but replacing
      atoms $S'(y, y')$ by $S'(x, x')$ and the two first atoms of the last two
      subqueries by:
      \begin{itemize}
        \item $G'(x, z) \wedge G(x, x', z')$
        \item $G(x, x', z) \wedge G'(x', z')$
      \end{itemize}
  \end{itemize}

  Given an initial instance of the tiling problem $c_0, \ldots, c_n$, we encode it in the initial
  set of facts $\instance_0$:
  \begin{itemize}
    \item $S'(a_i, a_{i+1})$ for $0 \leq i < n$;
    \item $G(a_0, a_i, b_{0,i})$ for $0 < i \leq n$;
    \item $G'(a_0, b_{0,0})$
    \item for all $0 \leq i \leq n$, letting $j$ be such that $c_i$ is the
      $j$-th color $C_j$, 
      we create the \emph{length-$j$ gadget on $b_{0,i}$}:
      we create a path $C(b_{0,i}, d_{0,i}^1), 
      C(d_{0,i}^1, d_{0,i}^2),
      \ldots
      C(d_{0,i}^{j-1}, d_{0,i}^j)$,
      and the fact $T(b_{0,i}, d_{0,i}^j)$;
  \end{itemize}

  \myparagraph{Correctness proof for the reduction}
  We claim that the tiling problem has a solution iff
  there is a (generally infinite) superset of $\instance_0$ that satisfies $\Sigma$ and
  violates $Q$, where the $S^\trans$ and $C^\trans$ predicates are interpreted as the
  transitive closure of $S$ and $C$, from which we conclude the reduction and
  deduce the undecidability of $\owqatc$ as stated.

  \medskip

  For the forward direction, from a solution $f$ to the tiling problem for input
  $\vec{c}$, we construct  $\instance \supseteq \instance_0$ as follows.
  We first create 
  an infinite chain $S'(a_0, a_1), \ldots, S'(a_m, a_{m+1}), \ldots$ to complete
  the initial chain of $S'$-facts in~$\instance_0$, we create the implied
  $S$-facts, and make $S^\trans$ the transitive closure.
  We then create one fact $G(a_i, a_j, b_{i,j})$ for all $i \neq j$
  in~$\mathbb{N}$ and one fact $G'(a_i, b_{i,i})$ for all $i \in
  \mathbb{N}$.
  Last, for all $i, j \in \mathbb{N}$, letting $l \defeq f(i, j)$, 
  we create the length-$l$ gadget on $b_{i,j}$ with fresh elements.

  It is clear that $\instance$ contains the facts of~$\instance_0$. It is easy to verify
  that it satisfies $\Sigma$. To see that we do not satisfy the query, observe
  that:
  \begin{itemize}
    \item The $C$-path sanity disjuncts have no match 
      because all $C$-paths created have length $\leq k$ and are on disjoint
      sets of elements;
    \item For the horizontal adjacency disjuncts, it is clear that, in any
      match, $z$ must be of the form $b_{i,j}$ and $z'$ of the form $b_{i,j+1}$;
      the reason for the three different forms is that the case where $i = j$
      and $i \neq j$ are managed differently. Then, as $f$ respects
      $\mathbb{H}$, we know that the $Q_i$ and $Q_j$ subqueries cannot be
      satisfied, because for any $l \in \mathbb{N}$ and $i', j' \in \mathbb{N}$,
      we have $Q_l(b_{i',j'})$ iff $f(i', j') = l$ by construction;
    \item The reasoning for the vertical adjacency disjuncts is analogous.
  \end{itemize}

  Hence, $\instance \supseteq \instance_0$, satisfies $\Sigma$, and violates
  $Q$, which concludes the proof of the forward direction of the implication.

  \medskip

  For the backward direction, consider a  $\instance \supseteq \instance_0$ that
  satisfies $\Sigma$ and violates $Q$. Starting at the chain of $S'$-facts
  of~$\instance_0$, we can see that there is  an infinite chain $a_0, \ldots, a_n,
  \ldots$ of $S'$-facts
  (whose elements may be distinct or not, this does not matter), and hence
  we infer the existence of the corresponding $S$-facts.  We can also infer the
  existence of elements $b_{i,j}$ for all $i, j \in \mathbb{N}$ (again, these
  elements may be distinct or not) such that $G'(a_i, b_{i,i})$ holds and $G(a_i, a_j,
  b_{i,j})$ holds if $i \neq j$. From this we conclude that there is a fact $T(b_{i,j}, c_{i,
  j})$ for all $i, j \in \mathbb{N}$, with a $C$-path from $b_{i,j}$ to
  $c_{i,j}$. As the $C$-path sanity disjuncts are violated, there cannot be such a $C$-path
  of length $\geq k$, so we can define a function $f$ from $\mathbb{N} \times
  \mathbb{N}$ to $\mathbb{C}$ by setting $f(i, j)$ to be $c_l$ where $l$ is the
  length of one such path, for all $i, j\in \mathbb{N}$; this can be performed in a
  way that matches~$\instance_0$ (by choosing the path that appears
  in~$\instance_0$ if there is one).

  Now, assume by contradiction that $f$ is not a valid tiling. If there are $i,
  j \in \mathbb{N}$ such that $(f(i, j), f(i, j+1)) \in \mathbb{H}$, then
  consider the match $x \defeq a_i$, $y \defeq a_j$, $y' \defeq a_{j+1}$, $z \defeq
  b_{i,j}$, and $z' \defeq
  b_{i,j+1}$. If $i
  \neq j$ and $i \neq j+1$, we know that $G(a_i, a_j, b_{i,j})$ and $G(a_i,
  a_{j+1}, b_{i,j+1})$ hold, and
  taking the witnessing paths used to define $f(i,j)$ and $f(i,j+1)$, we obtain
  matches of $Q_{f(i,j)}(b_{i,j})$ and $Q_{f(i,j+1)}(b_{i,j+1})$, so that we
  obtain a match of one of the disjuncts of~$Q$ (one of the first horizontal
  adjacency disjuncts), a contradiction. The cases where $i = j$ and where $i
  = j+1$ are similar and correspond to the second and third kinds of horizontal
  adjacency disjuncts. The case
  of $\mathbb{V}$ is handled similarly with the vertical adjacency disjuncts.
  Hence, $f$ is a valid tiling, which
  concludes the proof of the backward direction of the implication, shows the
  equivalence, and concludes the reduction and the undecidability proof.

  \myparagraph{Adapting to a single distinguished relation}
  To prove the result with a single distinguished relation $S^\trans$, simply
  replace all occurrences of $C$ and $C^\trans$ in the query and constraints by
  $S$ and $S^\trans$. The rest of the construction is unchanged.
  The proof of the backwards direction is unchanged, using $S$ in place of $C$;
  what must be changed is the proof of the forward direction.
  
  Let $f$ be the
  solution to the tiling problem. We start by
  constructing a set of facts $\instance_1$ as before from~$f$ to complete
  $\instance_0$, replacing the $C$-facts
  in the gadgets by
  $S$-facts. Now, we complete $S^\trans$ to add the transitive closure of these
  paths (note that they are disjoint from any other $S$-fact), and 
  complete this to a set of facts to satisfy $\Sigma$: create $G$- and $G'$-facts, and create gadgets,
  this time taking all of them to have length $k+1$: this yields $\instance_2$.
  We repeat this last process indefinitely
  on the path of $S$-facts created in the gadgets of the previous iteration, and
  let $\instance$ be the result of this infinite process, which
  satisfies~$\Sigma$.
 
  We
  justify as before that $Q$ has no matches: as we create no
  $S'$-facts in $\instance_i$ for all $i > 1$, it suffices to observe that no
  new matches of $Q$ can include any of the new facts, because each disjunct
  includes an $S'$-fact. Hence, we can conclude as before.

  \myparagraph{From UCQ to CQ}
  To prove the result with a CQ rather than a UCQ, we proceed as for the proof
  of Theorem~\ref{thm:tcdisj}: we extend $S'$ to be a ternary relation with
  a propagated value, add the relations of $\sigma_\Or$ (see the definition in
  the proof of Theorem~\ref{thm:tcdisj}), 
  modify the $S'$-atoms in all disjuncts of the UCQ $Q$ to add the variable,
  connect them as before yielding the CQ $Q'$, and modify the initial instance
  to add dummy matches, to add the element $\true$ to the $S'$-facts that we
  create, and to add the facts of $\instance_{\mathrm{Or}}$. As before, the
  proof of the forward direction
  is unchanged except that we add the value $\true$ to all $S'$-facts, and chase
  on the vacuous matches to satisfy the constraints (recall
  Definition~\ref{def:chase}). The query is violated because,
  thanks to the $S'$ contained in each disjunct of $Q$, 
  any match of the query ensures that we have a
  match of a disjunct of $Q$ on the part that corresponds to $\instance_0$ (not
  on the vacuous matches). For the backwards direction, we extract the tiling as
  before, and argue thanks to the vacuous matches that any violation of the
  tiling property would violate a UCQ of $Q$, and hence violate $Q'$. This
  concludes the proof.
\end{proof}

\clearpage

\section{Undecidability results related to linear orders (from Section~\ref{sec:undecidlin})} 

We first prove Theorem  \ref{thm:undeccq}.
Recall the statement:

\begin{quote}
{\bf Theorem~\ref{thm:undeccq}.} \undeccq
\end{quote}

\begin{proof}
  We first show the claim for a UCQ rather than a CQ.
  As in the proof of Theorem~\ref{thm:undectrans}, we fix an undecidable
  infinite tiling problem $\mathbb{C}$, $\mathbb{V}$, $\mathbb{H}$,
  and will reduce that problem to the $\owqalin$ problem.

  \myparagraph{Definition of the reduction}
  We consider the signature consisting of two binary relations $R$ and $D$ (for
  ``right'' and ``down''), $k-1$ unary relations $K_1, \ldots, K_{k-1}$
  (representing the colors), and one unary relation $S$ (representing the fact
  of being a vertex of the grid -- this is just to simplify things).
  
  We put the following inclusion dependencies in $\Sigma$:

  \begin{itemize}
    \item $\forall x \, S(x) \rightarrow \exists y \, R(x, y)$
    \item $\forall x \, S(x) \rightarrow \exists y \, D(x, y)$
    \item $\forall x y \, R(x, y) \rightarrow S(y)$
    \item $\forall x y \, D(x, y) \rightarrow S(y)$
  \end{itemize}
  We will use the following abbreviations:
  \begin{itemize}
    \item $K_1'(x)$ stands for $\exists y ~ x < y \wedge K_1(y)$
    \item $K_k'(x)$ stands for $\exists y ~ x > y \wedge K_{k-1}(y)$
    \item for all $1 < i < k$, $K_i'(x)$ stands for $\exists y y' ~
      K_{i-1}(y) \wedge y < x \wedge x < y' \wedge K_i(y')$.
  \end{itemize}
  Intuitively, the $K_i'$ describe the color of elements, which is encoded in
  their order relation to elements labeled with the $K_i$.

  We consider a UCQ formed of the following disjuncts (existentially closed):

  \begin{itemize}
    \item $R(x, y) \wedge D(x, z) \wedge R(z, w) \wedge D(y, w') \wedge w < w'$
    \item $R(x, y) \wedge D(x, z) \wedge R(z, w) \wedge D(y, w') \wedge w' < w$
    \item for each $(c, c') \in \mathbb{H}$, $R(x, y) \wedge K'_c(x) \wedge
      K'_{c'}(y)$
    \item for each $(c, c') \in \mathbb{V}$, $D(x, y) \wedge K'_c(x) \wedge
      K'_{c'}(y)$
  \end{itemize}
  Intuitively, the first two disjuncts enforce a grid structure, by saying that
  going right and then down must be the same as going down and then right. The
  two other disjuncts enforce that there are no bad horizontal or vertical
  patterns.

  Let us now present the reduction. Consider an instance $c_0, \ldots, c_n$ of
  the tiling problem. We construct a set of facts $\instance_0$ as follows:
  \begin{itemize}
    \item $S(a_0), \ldots, S(a_n)$
    \item $R(a_{i-1}, a_i)$ for $1 \leq i \leq n$
    \item $K_i(b_i)$ for $1 \leq i \leq k$
    \item for each $i$ such that $c_i$ is the color $C_1$, $a_i < b_1$
    \item for each $i$ such that $c_i$ is the color $C_k$, $a_i > b_{k-1}$
    \item for each $1 < j < k$ and $i$ such that $c_i$ is $C_j$, $b_{j-1} < a_i$
      and $a_i < b_j$
  \end{itemize}

  \myparagraph{Correctness proof for the reduction}
  Let us show that the reduction is sound. Let us first assume that the tiling
  problem has a solution $f$. We construct a counterexample $\instance \supseteq \instance_0$ as a grid
  of the $R$ and $D$ relations, with the first elements of the first row being
  the $a_0, \ldots, a_n$, and with the color of elements being coded as their
  order relations to the $b_j$ like when constructing $I$ above. Complete the
  interpretation of $<$ to a total order by choosing one arbitrary total order
  among the elements labeled with the same color, for each color. The resulting
  interpretation is indeed a total order relation, formed of the following: some total order on
  the elements of color $1$, the element $b_1$, some total order on the elements
  of color $2$, the element $b_2$, \ldots, the element $b_{k-2}$, some total order
  on the elements of color $k-1$, the element $b_{k-1}$, some total order on the
  elements of color $k$.

  It is immediate that the result satisfies $\Sigma$. To see why it does not
  satisfy the first two disjuncts of the UCQ, observe that any match of
  $R(x,y) \wedge D(x,z) \wedge R(z,w) \wedge D(y,w')$
  must have $w = w'$, by
  construction of the grid in $\instance$.
  To see why it does not satisfy the other disjuncts,
  notice that any such match must be a pair of two vertical or two horizontal
  elements; since the elements can match only one $K'_c$ which reflects their
  assigned color, the absence of matches follows by definition of $f$ being
  a tiling.

  \medskip

  Conversely, let us assume that there is a counterexample $\instance \supseteq \instance_0$ which satisfies
  $\Sigma$ and violates $Q$. Clearly, if the first two disjuncts of $Q$ are
  violated, then, for any element where $S$ holds, considering its $R$ and
  $D$ successors that exist by $\Sigma$, and respectively their $D$ and $R$
  successors, we reach the same element. Hence, from $a_0, \ldots, a_n$, we can
  consider the part of $\instance$ defined as a grid of the $R$ and $D$ relations,
  and it is indeed a full grid ($R$ and $D$ edges occur everywhere they
  should). Now, we observe that any element except the $b_j$ must be inserted at
  some position in the total suborder $b_1 < \cdots < b_{k-1}$, so that at least
  one predicate $K'_j$ holds for each element of the grid (several $K'_j$ may
  hold in case $\instance$ has more elements than the $b_i$ that are labeled
  with the $K_i$). Choose one of them, in a way
  that assigns to $a_0, \ldots, a_n$ their correct colors, and use
  this to define a function $f$ that extends $a_0, \ldots, a_n$. We claim
  that this $f$ indeed describes a tiling.

  Assume by contradiction that it does not. If there are   two horizontally
  adjacent values $(i, j)$ and $(i+1, j)$ realizing a configuration $(c, c')$
  from $\mathbb{H}$, by completeness of the grid there is an $R$-edge between
  the corresponding elements $u, v$ in $\instance$. Further, by the fact that $(i,j)$
  and $(i+1, j)$ were given the color that they have in $f$, we must have
  $K'_c(u)$ and $K'_c(v)$ in $\instance$, so that we must have had a match of a
  disjunct of $Q$, a contradiction. The absence of forbidden vertical patterns
  is proven in the same manner.

  \myparagraph{From UCQ to CQ}
  We now adapt the previous proof to use a CQ rather than a UCQ.
  Define the new signature $\sigma_{\mathrm{Or}}'$ as in the proof of
  Theorem~\ref{thm:tcdisj} by adding the relations of $\sigma_\Or$, and
  otherwise increasing the arity of each relation of $\sigmab$ by one. Rewrite
  the IDs $\Sigma$ as in the proof of Theorem~\ref{thm:tcdisj}, yielding:

  \begin{itemize}
    \item $\forall x ~ S(x, b) \rightarrow \exists y ~ R(x, y, b)$
    \item $\forall x ~ S(x, b) \rightarrow \exists y ~ D(x, y, b)$
    \item $\forall x y ~ R(x, y, b) \rightarrow S(y, b)$
    \item $\forall x y ~ D(x, y, b) \rightarrow S(y, b)$
  \end{itemize}
  We add to our initial set of facts $\instance_0$ the facts of
  $\instance_{\mathrm{Or}}$ as in the proof of
  Theorem~\ref{thm:tcdisj}.

  We construct the CQ from the original UCQ by the same process as in
  the proof of Theorem~\ref{thm:tcdisj}, and construct the initial sets of facts
  as in that proof as well.

  We then argue that this
  $\owqalin$ problem with the UCQ is equivalent to the one with the CQ. For the
  forward direction, from a solution to the initial instance $a_0, \ldots, a_n$
  of the tiling problem, we build a suitable
  $\instance \supseteq \instance_0$ from the previously defined $\instance$ by
  putting $\true$ as the last element of $R$- and $D$-facts.
  We complete the vacuous matches by chasing on them with the dependencies
  of~$\Sigma$ (as in the proof of Theorem~\ref{thm:tcdisj}; recall
  Definition~\ref{def:chase}),
we define $<$
  arbitrarily on each vacuous match, arbitrarily between them, arbitrarily with
  $\false$ and $\true$, and then as before on the true grid.
  To show that the query has no
  match, we first claim that any match of the query must be a match of one of
  the disjuncts where the free variable is bound to $\true$. Indeed, this is clear
  by definition of the $\mathrm{Or}$ and $\mathrm{True}$ relations.
  Now, we claim that none of the
  disjuncts have such a match. Indeed, if one disjunct has such a match, it
  implies that all facts of the match (except the order facts) have $\true$ in the
  last position, and, as these facts are the same as in the original
  proof (up to the last element), the absence of match is for the same reason as
  in the original proof.
  
  Conversely, let us consider a $\instance \supseteq \instance_0$ satisfying $\Sigma$
  and violating the query. We first observe that, for any disjunct of the query,
  it has a match where the free variable is bound to $\false$, as witnessed by
  the vacuous matches. Hence, if the query has no match, it must mean
  that none of the disjuncts has a match with the free variable bound to $\true$.
  Indeed, if there were one, then, from this match, from the vacuous matches,
  and using the facts which we know are present in the table of $\mathrm{Or}$
  and $\mathrm{True}$,
  we would obtain a match of the entire query. Hence, restricting our attention
  to the facts of $\instance'$ with $\true$ in their last position, using the fact that none
  of the disjuncts has a match there, we conclude as in the original proof.
\end{proof}

Now recall the statement of
Corollary~\ref{cor:undeclin}:
\begin{quote}
{\bf Corollary~\ref{cor:undeclin}.} 
  There is a signature $\sigma = \sigmab \sqcup \sigmad$ where $\sigmad$ is a single strict
  linear order relation, and a set $\Sigma'$ of $\afgtgd$ constraints,
  such that, letting $\top$ be the tautological query, the following problem is
  undecidable:
  given a finite set of facts $\instance_0$, decide $\owqalin(\instance_0,
  \Sigma', \top)$.
\end{quote}

To prove Corollary~\ref{cor:undeclin} from Theorem~\ref{thm:undeccq},
we take constraints $\Sigma'$
that are equivalent to $\Sigma \wedge \neg Q$,
where $\Sigma$ and $Q$ are as in the previous theorem.
Recall that $\Sigma$ is a set of inclusion dependencies on $\sigmab$,
and therefore are $\afgtgd$s.
Hence, it only remains to argue that $\neg Q$ can be written as a $\afgtgd$.
Indeed, write $Q$ as $\exists \vec{x} ~ \phi(\vec{x})$ and
consider the constraint 
\[
\forall \vec{x} ( \phi(\vec{x}) \rightarrow \exists y ( y < y ) )
\]
where $<$ is the distinguished relation.
Since $<$ must be a strict linear order in $\owqalin$,
$\exists y ( y < y)$ is equivalent to $\bot$ and
this new constraint is logically equivalent to $\neg Q$.
Moreover, this constraint is trivially in $\afgtgd$
since there are no frontier variables.
Hence, $\neg Q$ can be written as an $\afgtgd$ as claimed.

\end{document}